\documentclass[acmsmall, screen]{acmart}
\AtBeginDocument{%
  }

\setcopyright{cc}
\setcctype{by}
\acmDOI{10.1145/3763080}
\acmYear{2025}
\acmJournal{PACMPL}
\acmVolume{9}
\acmNumber{OOPSLA2}
\acmArticle{302}
\acmMonth{10}
\received{2025-03-25}
\received[accepted]{2025-08-12}

\usepackage{mathtools}
\usepackage{mathpartir}
\usepackage{quiver}
\usepackage{listings}
\usepackage{wrapfig}
\usepackage{cleveref}
\newcommand{\action}[2]{\langle #1 \mid #2 \rangle}
\newcommand{\actiontemp}[4]{(\langle #1 \mid #2 \rangle, \langle #3, #4\rangle)}
\newcommand{\tuple}[3]{(#1, #2_{#3})}
\newcommand{\LVars}{LVars}
\newcommand{\sem}[1]{\left\llbracket #1 \right\rrbracket} 
\newcommand{\triple}[4]{#1\; {#2}\; {#3}_{#4}}
\newcommand{\while}{{\textbf{for}}}

\usepackage{ stmaryrd }
\usepackage{booktabs}
\usepackage{mathtools}

\usepackage{amsthm}
\usepackage{thmtools}

\declaretheorem[name=Theorem,numberwithin=section]{theorem}
\declaretheorem[name=Lemma,numberwithin=section]{lemma}
\declaretheorem[name=Definition,numberwithin=section]{definition}

\newif\iflong
\longtrue
\newif\iffull

\iflong
\fulltrue
\else 
\fi

\begin{document}


\title{A Hoare Logic for Symmetry Properties}

\author{Vaibhav Mehta}
\orcid{0000-0003-2357-3023}
\affiliation{%
  \institution{Cornell University}
  \city{Ithaca}
  \country{USA}
}
\email{vm353@cornell.edu}

\author{Justin Hsu}
\orcid{0000-0002-8953-7060}
\affiliation{%
  \institution{Cornell University}
  \city{Ithaca}
  \country{USA}
}
\email{email@justinhsu.net}

\begin{abstract}
  Many natural program correctness properties can be stated in terms of
  symmetries, but existing formal methods have little support for reasoning
  about such properties. We consider how to formally verify a broad class of
  symmetry properties expressed in terms of group actions. To specify these
  properties, we design a syntax for group actions, supporting standard
  constructions and a natural notion of entailment. Then, we develop a
  Hoare-style logic for verifying symmetry properties of imperative programs,
  where group actions take the place of the typical pre- and post-condition
  assertions. Finally, we develop a prototype tool $\mathsf{SymVerif}$, and use
  it to verify symmetry properties on a series of handcrafted benchmarks. Our
  tool uncovered an error in a model of a dynamical system described by \citet{McLachlan_Quispel_2002}.
\end{abstract}

\maketitle

\section{Introduction}
\label{sec:intro}

Symmetry is a universal concept across mathematics and physics. In mathematics,
symmetries are transformations that preserve aspects of mathematical structure.
In physics, symmetries correspond to conserved quantities. In this work, we
develop methods to verify \emph{symmetry properties of programs}.

While this combination may seem strange, programs have mathematical and physical
aspects. For instance, programs can be given semantics as a mathematical object,
e.g., a relation or a trace. Programs can also describe a physical process of
computation, or model some physical process in the world. Furthermore, many
target properties can be naturally phrased in terms of some kind of symmetry.
For example, a program modeling a physical system, like the dynamics of a car,
often satisfies a safety property that is invariant under coordinate
transformations: shifting the entire system by a fixed offset should not change
whether the property holds or not. Likewise, we may be interested in showing
that a program's behavior does not depend on certain aspects of its input; for
instance, the winner computed by a voting mechanism should not depend on the
order of the votes.

\paragraph{Prior work.}
We are not the first to consider symmetries of programs. On the semantic side,
nominal logic~\citep{pitts2003nominallogic} reasons about objects that are
symmetric under renaming (e.g., program syntax with variable binding).  This
line of work considers symmetries induced by a single, fixed group action
capturing renaming. While this setup is natural for reasoning about variable
binding, it is too restrictive for general symmetry properties, which may
involve multiple kinds of symmetries.

On the algorithmic side, work on symmetry reduction has been highly successful
in automated reasoning. These methods apply to finite models of computation,
like formulas and transition systems, where symmetries can be found using
algorithmic techniques. However, it is unclear how to extend these methods to general
programs, where data may range over unbounded domains.

\paragraph{Our work.}
We propose a general class of symmetry properties phrased in terms of group
actions, a concept from group theory that can encode transformations of an
object. Our target properties require that a group action on the input of a
program leads to a (possibly different) group action on the output. Our
properties generalize prior symmetry properties, like equivariance, and
encompass symmetry properties in a variety of applications. 

Next, we consider how to verify symmetry properties for imperative programs.
There are several challenges to overcome. First, in order to specify our
properties, we need a compact yet expressive syntax for encoding group actions, which might be infinite.
Second, we would like to verify symmetry properties for larger programs by
establishing symmetry properties for smaller subcomponents. We also want to take
advantage of compositional structure in group actions: for actions that are
built out of simpler actions, we should be able to establish symmetry properties
for the simpler actions separately before combining our results. Finally,
verifying symmetry properties can involve a great deal of complex,
group-theoretic reasoning that we hope to simplify.

To overcome these challenges, we design a novel, Hoare-style logic for reasoning
about group-theoretic symmetries. In a nutshell, our logic captures symmetry
properties by replacing the usual pre- and post-conditions of Hoare logic with
\emph{group actions}. The proof rules of our logic allow combining symmetry
properties of simpler programs and group actions into more complex properties.
Finally, to simplify reasoning, we develop automated methods to ease
verification.

\paragraph*{Outline.}
After a technical overview in \Cref{sec:overview}, we present our contributions.
\begin{enumerate}
  \item We develop a compact syntax of groups and group actions, supporting
    standard constructions in group theory. We provide semantics for our
    assertions and a notion of entailment, validating typical logical properties
    (\Cref{sec:assertions}).
  \item We develop a Hoare-style program logic for an imperative programming
    language, using our group-action assertions as the pre- and post-conditions.
    We provide a set of proof rules and prove soundness. We also develop
    weakest-precondition and strongest-postcondition transformers for
    assignments (\Cref{sec:logic}).
  \item To demonstrate the expressiveness of our logic, we verify symmetry
    properties from a series of examples drawn from different application
    domains (\Cref{sec:examples}).
  \item We implement tools for verifying symmetry triples for programs and
    synthesizing pre-conditions. We evaluate our tools on a handcrafted set of
    benchmark programs (\Cref{sec: impl}). Our tool helped uncover a subtle
    error in a model dynamical system (the \emph{AAC flow}) described by
    \citet{McLachlan_Quispel_2002}, where the claimed symmetry property fails to
    hold.
\end{enumerate}
We finish by surveying related work (\Cref{sec:rw}) and concluding
(\Cref{sec:conc}).

\section{Technical Overview}
\label{sec:overview}
We begin with an overview of symmetry properties and our symmetry program logic.

\subsection{Groups and Group Actions}
We will need some basic concepts from group theory for our development (see,
e.g.,~\citep{rotman2012introduction}). To begin, we introduce
\emph{groups} and maps between groups.

\begin{definition}
  \label{def: group}
  A \emph{group} is a nonempty set equipped with an associative
  operation~$\cdot$ and an element $e$ such that (1) $e \cdot g = g = g \cdot
  e$, for all $g \in G$, and (2) for every $g \in G$, there is an element $h \in
  G$, called the \emph{inverse}, such that $g \cdot h = e = h \cdot g$.  We often elide
  the group operation, writing $gh$ for $g\cdot h$.

  Given $(G, \cdot_G)$ and $(H, \cdot_H)$ groups, a function $f: G \rightarrow
  H$ is a \emph{homomorphism} if, for all $a,b \in G$, $f(g_1 \cdot_G g_2) =
  f(g_1)\circ_H f(g_2)$ and $f(e_G) = e_H$. A bijective homomorphism is an
  \emph{isomorphism}.
\end{definition}

Groups can encode symmetries by their \emph{action} on a target set.
Informally, a group action associates each group element with a transformation
of some set $X$, representing a particular symmetry of $X$.

\begin{definition}
  Let $G$ be a group and $X$ be a set. An \emph{action of $G$ on $X$} is a
  function $a: G \times X \rightarrow X$ such that
  (1) $a_e(x) = x$ for all $x \in X$, and
  (2) $a_g(a_h(x)) = a_{gh}(x)$ for all $x \in X$.
\end{definition}

Group actions are intimately related to a group called the \emph{symmetric group}.
\begin{definition}
  The \emph{symmetric group} over a set $X$, denoted by $Sym(X)$ is the set of
  all bijections from $X \rightarrow X$, with function composition $\circ$ as
  the group operation, the identity function as the identity, and the inverse of
  a function as its inverse element. 
\end{definition}

\subsection{Symmetry Properties}
Our symmetry properties describe how the output states of a program are
transformed, given a transformation of the input state of a program. More
formally, a program satisfies a symmetry property if the transformations of the
input under a \emph{pre-group action} $(G, a)$ lead to transformations of the
output under a \emph{post-group action} $(H, b)$.
For instance, when $(H, b)$ is the trivial group $E$ with a single element that
acts by preserving the program state, the symmetry property ensures
\emph{invariance}: acting on the input with $(G, a)$ leaves the output
unchanged.

\paragraph*{A Warm-up Example}
Figure \ref{fig: car-trans} simulates a simple model of a car from time step $t
= 0$ up to $t = T$, where $dt$ is some small time increment.  The variables $x$
and $y$ represent the car coordinates updated at every time step depending on
the velocity $v$, which depends on the acceleration $a$. The variable
$\phi$ is the steering angle, and the turning rate is $u$. The angle $\theta$
the car makes with the x-axis changes at every time step based on $\phi$, the
velocity $v$, and a constant $L$.

Now, suppose we change the starting coordinates from $(x,y)$ to $(x+c, y+c)$.
Intuitively, after the transformation, the path taken by the car will be
identical to the original path shifted by $c$ in each coordinate.
In particular, if the starting position is translated by $c$, then the final
position of the car should also be translated by $c$.
In the rest of this section, we will frame this property as a symmetry property,
and describe how to verify this property in our system.

\subsection{Proving Symmetry Properties, Formally}
We propose a novel Hoare logic where assertions are group actions and judgments
are of the form
\[\triple{(G,a)}{C}{(H,b)}{\varphi}\]
In this triple, $(G,a)$ is a group action on the input, $C$ is a program, $\varphi: G \rightarrow H$, is a homomorphism and
$(H,b)$ is the group action on the output. The judgment
$\triple{(G,a)}{C}{(H,b)}{\varphi}$ captures a symmetry property: whenever the inputs to
the program $C$ are transformed by $(G,a)$, the outputs are related by the
action $(H,b)$. The homomorphism $\varphi: G \rightarrow H$ connects these two
actions: transforming the program input using the action of $g \in G$ must
transform the program output using the action of $\varphi(g) \in H$.
\begin{figure}
  \begin{center}
    \begin{tabular}{c}
\begin{lstlisting}
1. $\mathbf{for}(t \coloneq 0;t < T; t \coloneq t + dt)${
2. $\quad \quad x \coloneq v\cdot \cos{(\theta)}\cdot dt + x;$
3. $\quad \quad y \coloneq v\cdot \sin{(\theta)}\cdot dt + y;$
4. $\quad \quad  v \coloneq a \cdot dt + v;$
5. $\quad \quad \phi \coloneq u \cdot dt + \phi;$
6. $\quad \quad \theta \coloneq \frac{v}{L} \tan{(\phi)} \cdot dt + \theta$
7. }
\end{lstlisting}
    \end{tabular}
  \end{center}
  \vspace{-4.0mm}
  \caption{A program representing a car system}  \label{fig: car-trans}
\end{figure}

\paragraph*{Expressing Symmetry Properties as Assertions}
To verify the example of the car we just saw, we need first to define the
assertions $(G,a)$ and $(H,b)$ that model the translation of the car. For the
pre-group action, we select a group action that acts on the program state by
translating the coordinates $x$ and $y$. The set of integers $\mathbb{Z}$ forms
a group under addition and we want each integer $c$ to act on $(x, y)$ by
mapping to $(x + c, y + c)$. To complete the triple, we also need to specify a
homomorphism that captures the relation from the input to the output. In our
case, we would like that if the starting position of the car is translated by
$c$, then the final position of the car will also be
translated by $c$. This can be expressed as the identity homomorphism
$\varphi(c) = c$. We will denote this homomorphism as $\mathbf{eq}$.
At this point, we know intuitively what the specification looks like, but it is
unclear how to describe group actions and homomorphisms explicitly; for
instance, the group $\mathbb{Z}$ and set of program states $\Sigma$ are
infinite.

\paragraph*{A Finite Syntax}
To address this problem, we develop a finite syntax for expressing actions using
the idea of \emph{finitely presented groups}. In a nutshell, groups can often
be expressed using finitely many \emph{generator} elements and equations
(\emph{relations}).  For example, the group of integers $\mathbb{Z}$ has one
generator $g$, and no equations. Every element is a word consisting of $g$ or
its inverse $g^{-1}$. For example, $g$ corresponds to $1$, $gg$ to $2$, and so
on. The group $\mathbb{Z}$ can be expressed (\emph{presented}) as
\(
  \langle g \mid - \rangle
\).

The action of the group can also be defined in terms of
these elements, which enables a finite description of the group action.
For instance, $\mathbb{Z}$ acts on program state by letting
$c \in \mathbb{Z}$ map $(x,y$) to $(x + c, y + c)$.  In our syntax, we define
this action $a: \mathbb{Z} \times \Sigma \rightarrow \Sigma$ via:
\[
  (\mathbb{Z}, a_{\{x,y\}})  \triangleq \actiontemp{g}{-}{\{x,y\}}{ g \cdot (\{x \rightarrow \alpha_x, y \rightarrow \alpha_y\}) = \{x \rightarrow \alpha_x + 1 , y \rightarrow \alpha_y + 1\}}
\]
To read this definition, $\langle g \mid - \rangle$ defines
the group: the integers. The $\{x,y\}$ annotation
specifies which program variables the action is transforming, and the final
equation $g(\{x \rightarrow \alpha_x, y \rightarrow \alpha_y\}) = \{x
\rightarrow \alpha_x + 1 , y \rightarrow \alpha_y + 1\}$ describes how the
generator transforms these variables: the generator $g$ acts on the variables
$(x, y)$ by adding $1$ to each variable. Finally, to extend this definition to
act on whole program states, we simply require that the action preserves the
value of all variables other than $x$ and $y$; we call this resulting group action
$(\mathbb{Z}, a_{\mathit{Vars}})$, where $Vars$ indicates that $a$ acts on
all variables. \Cref{sec:assertions} introduces the syntax for groups, group
actions, and homomorphisms and considers how to build complex groups and group actions out of
simpler parts. 

Returning to our example, let $C_{car}$ denote the program in Figure \ref{fig: car-trans}.
The symmetry property of the car system can now be expressed as the judgment:
\begin{equation}\label{eq:car-triple}
  \triple{(\mathbb{Z}, a_{\mathit{Vars}})}{C_{car}}{(\mathbb{Z}, a_{\mathit{Vars}})}{\mathbf{eq}}
\end{equation}
In words, whenever the initial coordinates $x$ and $y$ are shifted by an
integer, the output coordinates are also transformed by the same integer and all
other variables are unchanged. The homomorphism $\mathbf{eq}$ specifies that the output coordinates are also transformed by the same integer as the input.

\paragraph*{Proving Judgements}
We can now formally specify symmetry properties, but proving them is
challenging. To make this task easier, we develop a Hoare-style proof system for
our logic (\Cref{sec:logic}). The rules for program commands are similar to
rules in standard Hoare logic, while the structural rules use constructions from
group theory to derive more complicated assertions from simpler ones.  Together,
these rules allow us to prove judgments \emph{compositionally}.
For instance, the assignment rule can be applied to derive the judgment
\[
\triple{(\mathbb{Z}, a_{\mathit{Vars}})}{ x \coloneq v\cdot \cos{(\theta)}\cdot dt + x}{(\mathbb{Z}, a_{Vars})}{\mathbf{eq}}
\]
This judgment says that whenever the input program state has $x$ and $y$ shifted
by a number $c$, after the assignment statement, $x$ and $y$ are still shifted
by $c$ and the other variables remain unchanged. By applying this rule for every
assignment in Fig. \ref{fig: car-trans} and applying our sequencing and loop rules, we can conclude the Hoare judgment \eqref{eq:car-triple}, establishing our
target symmetry property.

\subsection{Symmetry Properties for Existing Proofs}
While symmetry properties themselves are interesting targets for verification,
they can also be highly useful for simplifying proofs of general properties that
have nothing to do with symmetry.
For example, in cyberphysical systems, we may assume a canonical coordinate
system—e.g., placing the origin at a convenient location, or aligning the axes
in a standard way (see \cite{kheterpal2022automatinggeometricproofscollision,
DBLP:journals/jais/GhorbalJZPGC14}). In such treatments, one typically argues
that because the system is invariant under certain transformations
(translations, rotations, etc.), it suffices to check correctness in fewer
representative configurations. While this kind of argument can greatly simplify
verification, it is often carried out informally (e.g., ``without loss of
generality''). Our proposed symmetry properties and verification methods can
help formalize this approach.

For instance, consider a simple model of a car with horizontal position \({x}\)
and vertical position \({y}\). Suppose the car’s control code increments \({x}\)
by 100 units and \({y}\) by 50 units every time step. In standard Hoare logic,
to prove \(\{\mathit{true}\}\,C\,\{Q\}\), where \(C \coloneq {x'} =
{x}+100;\;{y'} = {y}+50\) and \(Q\) asserts that final positions \({x'},{y'}\)
equal \({x}+100,{y}+50\), we would show directly that for any state \(\sigma\),
if \(\sigma\) satisfies the precondition, then \(C(\sigma)\) satisfies \(Q\).

By contrast, if we capture translation by integers as a group action on the
program state, we can decompose the verification task into two stages. First,
we prove the  symmetry triple
\[
      \triple{(\mathbb{Z}, a_{\mathit{Vars}})}{C}{(\mathbb{Z}, a_{\mathit{Vars}})}{\mathbf{eq}}
\]
saying translating the input by $c$ yields a consistent translation in the
output \(\langle {x'}, {y'}\rangle\).
Once established, we note that the Hoare logic assertions \(P\equiv \mathit{true}\) and 
\[
Q(\sigma)\;\equiv\;\bigl(\sigma({x'}) = \sigma({x}) + 100\bigr)\,\land\,\bigl(\sigma({y'}) = \sigma({y}) + 50\bigr)
\]
are closed under these translations: shifting \({x},{y}\) by \(c\) preserves
\(P\) and \(Q\).

Then, we can use the symmetry property: In this example, rather than verify the Hoare triple
\(\{\mathit{true}\}\,C\,\{Q\}\) for every input \(\langle {x},{y}\rangle\), we
need only check a single canonical input, like \(\langle 0,0\rangle\), so it
suffices to verify \(\{x = 0 \land y = 0\}\,C\,\{Q\}\).
Our symmetry property ensures that this triple holds for all inputs. In this
way, our work can be used as a complement to existing formal verification
methods for more general properties.

\section{Assertion Language for Group Actions}
\label{sec:assertions}
In order to specify group actions precisely, we develop syntax for group
actions. These will serve as assertions in our program logic.

\subsection{Syntax of Groups}
Now, we develop the syntax for defining groups. Groups can be defined by a
\textit{presentation} consisting of a set $S$ of generators and a set $R$
equations between certain group elements. More formally:
\begin{definition}
  Let $S$ be a set of symbols, $S^{-1}$ be a disjoint copy of $S$, and $R$ be a
  set of equations between words in $(S \cup S^{-1})^*$. We call the pair
  $\langle S \mid R \rangle$ a \emph{group presentation}.
\end{definition}
Intuitively, $G = \langle S \mid R \rangle$ is the smallest group generated by
$S$ that obeys the equations in $R$. As an example, the group presentation
$\langle x \mid x^2 = e \rangle$
describes the \emph{symmetric group} on two elements, usually denoted $S_2$.
This group has two elements: the identity $e$ and the generator $x$; the
relation $x \cdot x = e$ ensures that $x$ is its own inverse.
For a more complex example, the group presentation
\[
  \langle x,y \mid x^4 = e, y^2 = e, xy = yx^{-1} \rangle
\]
describes the \emph{dihedral group} on four elements, usually denoted $D_4$
\cite{rotman2012introduction}. This
group describes the symmetries of a square: the generator $x$ represents a
rotation by $90$ degrees, while the generator $y$ represents a reflection. The
first relation $x^4 = e$ states that rotating four times returns to the original
state; the relation $y^2 = e$ states that reflecting twice returns to the
original state; and the last relation $xy = yx^{-1}$ states that rotating and
then reflecting is the same as reflecting and then rotating in the
\emph{opposite} direction.

A group $G = \langle S \mid R \rangle$ is said to be \emph{finitely presented}
when $S$ and $R$ are both finite sets; note that a finitely presented group may
still have infinitely many elements. In the remainder of this paper, we only
consider finitely presented groups.

\subsection{Syntax of Group Actions}
We will be concerned with groups $G$ acting on program states $\sigma \in
\Sigma$, where $\Sigma$ is the set of maps from $\mathit{Vars} \rightarrow
\mathbb{Z}$ and $\mathit{Vars}$ is a finite set of program variables.
Often, we will consider group actions that transform just a few select
variables $V \subseteq \mathit{Vars}$, rather than all variables.
Let the set $\Sigma_V$ denote the set of program states with domain $V$:
\[
  \Sigma_V = \{\sigma \mid \sigma: V \rightarrow \mathbb{Z} \}
\]
Given this information, a group action consists of a function $a$ from $G \times
\Sigma_V \rightarrow \Sigma_V$. When we have a finite presentation $G = \langle
S \mid R \rangle$, we can define $a$ by specifying the action of each generator
$g \in S$:
\[
  \actiontemp{S}{R}{V}{g \cdot (v_1 \rightarrow \alpha_1, \dots, v_k \rightarrow \alpha_k) = (v_1 \rightarrow e_1, \dots, v_k \rightarrow e_k)
    \text{ for each } g \in S, v_i \in V \rangle}
\]
where each equation
\[
  g \cdot (v_1 \rightarrow \alpha_1, \dots, v_k \rightarrow \alpha_k) = (v_1 \rightarrow e_1, \dots, v_k \rightarrow e_k)
\]
describes how the generator $g \in S$ acts on the variables $v_i \in V$ in a program
state $\sigma$. In more detail, $\alpha_i$ is a logical variable representing
$\sigma(v_i)$, and $g$ maps $\sigma$ to a program state where variable $v_i$
contains the value corresponding to a symbolic expression $e_i$ drawn from the
following grammar:
\begin{align*}
  \mathcal{E} \ni e \coloneqq \alpha \in \LVars \mid n \in \mathbb{Z} \mid  f(e_1, \dots, e_n) , f \in Ops
\end{align*}
We fix a set of operations $Ops$; an operation $f \in Ops$ with arity $n \in
\mathbb{N}$ can be applied to $n$ arguments $e_1, \dots, e_n$. For example,
addition is a $2$-ary operation $f$, where $f(e_1, e_2)$ represents $e_1 + e_2$.

We will write $(G,a_V)$ for an action of $G$ on $\Sigma_V$, and $a_g(\cdot) :
\Sigma_V \to \Sigma_V$ for the action of $g \in G$.
\begin{example}
  \label{ex:s2-ex}
  The group $S_2$ can act on a pair of variables $\{v_1, v_2\}$ by swapping
  their values. In our syntax, this action can be expressed as:
  \[
    a \triangleq \actiontemp{x}{x^2 = e}{\{v_1, v_2\}}{x \cdot (v_1 \rightarrow \alpha_1, v_2 \rightarrow \alpha_2) = (v_1 \rightarrow \alpha_2, v_2 \rightarrow \alpha_1)}
  \]
  We will refer to this action as $(S_2,a_{\{v_1,v_2\}})$.
  The group element $x$ acts by changing the value of $v_1$ to $\alpha_2$ and the
  values of $v_2$ to $\alpha_1$, i.e., it \emph{swaps} the contents of variable
  $v_1$ with variable $v_2$.
\end{example}

Note that maps $a$ defined in our syntax are not guaranteed to be valid group
actions. To verify that $a$ is a group action, it suffices to check that every
generator acts in a way that obeys the equations in $R$.

\begin{restatable}{theorem}{FiniteGrpThm}
  Let $G$ be a finitely generated group presented as $\action{S_G}{R}$, and let
  $a: S_G \rightarrow Sym(\Sigma)$ be a map from the set of generators to the
  group of bijections over $\Sigma$. If for every equation $x = y \in R$ and
  every program state $\sigma \in \Sigma$ we have
  $a_x(\sigma) = a_y(\sigma)$, then $a$ is a group action.
\end{restatable} 

For example, the group action of $S_2$ defined in Example \ref{ex:s2-ex} is
valid because swapping the values $v_1$ and $v_2$ twice results in the starting
configuration: for all program states $\sigma \in \Sigma_{V}$, we have
$a_{x^2}(\sigma) = a_{e}(\sigma) = \sigma$.

For every group action $(G, a_V)$ defined on $\Sigma_{V}$, we can lift it to an
action $(G, \overline{a}_{\mathit{Vars}})$ on whole program states $\Sigma$ as
follows: for the variables in $V$, $\overline{a}$ acts like $a$, but for all
other variables, $\overline{a}$ acts like the identity. Formally, for a set of
variables $V = \{ v_1, \dots, v_k \}$, the lifting of $a$ is defined as:
\begin{align*}
     \actiontemp{S}{R}{\mathit{Vars}}{\forall g_i \in S, g_i\cdot &(v_1 \rightarrow \alpha_1, \dots, v_k \rightarrow \alpha_k, \dots, v_n \rightarrow \alpha_n) = \\& v_1 \rightarrow \alpha_1', \dots, v_k \rightarrow \alpha'_k, \dots, v_n \rightarrow \alpha_n}
\end{align*}

For example, the lifting of the group action $\tuple{S_2}{a}{\{v_1, v_2\}}$ is
defined by:
\[
  \actiontemp{x}{x^2=1}{\mathit{Vars}}{x\cdot(v_1 \rightarrow \alpha_1,  v_2 \rightarrow \alpha_2, \dots, v_n \rightarrow \alpha_n) =  (v_1 \rightarrow \alpha_2, v_2 \rightarrow \alpha_1, \dots, v_n \rightarrow \alpha_n) }
\]
The lifting of any group action remains a group action.

\begin{restatable}{theorem}{LiftingThm}
  If $(G, a_V)$ is a group action then $( G, \overline{a}_{\mathit{Vars}})$ is also a group action.
\end{restatable}  
\iflong
  \begin{proof}
    This follows from the fact the $(G, a_V)$ is a group action, and the lifting leaves everything outside of $V$ unchanged.
  \end{proof}
\fi
When clear from the context, we will denote the lifting of $(G, a_V)$ as simply $(G, a_{\mathit{Vars}})$.
\subsection{Group Constructions}
So far, we've seen how to define group constructions by writing down a
presentation. This can become tedious for more complicated groups.  In this
section, we show a simpler, more compositional way of defining group actions
using group constructions. Starting from smaller groups, we can define more
complex groups using group constructions. Our logic will use these constructions
to derive more complex assertions (i.e., group actions). In this section, we
define some of these constructions and show how to express them in our syntax.

\paragraph*{Direct product}
One basic way to combine two groups is via the \emph{direct product}.
\begin{restatable}{definition}{DirProd}
  If $G$ and $H$ are both groups, their \emph{direct product}, denoted by $G
    \times H$, is the group with elements pairs $(g,h)$ where $g \in G$ and $h \in
    H$, with group operation $
    (g,h)\cdot (g',h') = (gg', hh')$.
\end{restatable}
Suppose we have two groups $G$ and $H$ presented as $
G = \langle S_G \mid R_G \rangle$ and $H =\langle S_H \mid R_H \rangle$. Then
a presentation of the direct  product is:
\[
  G \times H = \langle S_G , S_H \; \mid \;  R_G , R_H , R_{P}\rangle\
\]
where $R_P$ is the set of relations that say that each element of $S_G$ commutes with $S_H$.
For example, the direct product of $S_2$ with itself can be presented as
\[
  S_2 \times S_2 = \langle x, y\; \mid \; x^2 = y^2 = e, xy = yx \rangle
\]

\paragraph*{Free Product.}
Our syntax also supports a group construction called the free product, which is shown in \iffull \Cref{app: fp-cons}. \else the appendix, which can be found in the supplementary materials. \fi
\subsection{Group Actions of Constructions}
Just like we can compose groups together, we can also compose group actions together.
\subsubsection*{Direct Product} Let $I = G \times H$ be a direct product group, where the actions of $G$ and $H$ are $\tuple{G}{a}{V_1}$ and $\tuple{H}{b}{V_2}$ respectively. We would like the new action of $I$ to be compatible with the actions $\tuple{G}{a}{V_1}$ and $\tuple{H}{b}{V_2}$.  However, an action compatible with $\tuple{G}{a}{V_1}$ and $\tuple{H}{b}{V_2}$ does not always extend to $I$. When it does extend to $I$, it
is defined as follows:  each generator of $G$ acts as $a_g$ for each variable in $V_1$, and fixes every variable in $V_2\setminus V_1$, whereas each generator of $H$ acts as $b_h$ for every variable in $V_2$, and fixes every variable in $V_1 \setminus V_2$.

\begin{restatable}{definition}{DirProdAction}
  \label{def: dir-prod-action}
  Suppose we have two groups $G$ and $H$ presented as $ G = \langle S_G \mid R_G
  \rangle$ and $H =\langle S_H \mid R_H \rangle$ with group actions $(G,
  a_{V_1})$ and $(H,b_{V_2})$.  Let $V_1 = \{v_1, \dots, v_k\}$, $V_2 = \{v'_1,
  \dots, v'_n\}$ and let $a_g$ map each integer $\alpha_i$ to $\alpha'_i$, and
  $b_h$ map each integer $\beta_j$ to $\beta_j'$.  Then the action of the direct
  product of $G$ and $H$ on $V_1 \cup V_2$ is defined to be:
  \begin{align*}  \actiontemp{S_G, H_G}{R_S, R_H, R_P}{V_1 \cup V_2}{ &
      \\ \forall g \in S_G ,{} &g\cdot(v_1 \rightarrow \alpha_1, \dots, v_k \rightarrow \alpha_k, v'_1 \rightarrow \beta_1, \dots, v'_n \rightarrow \beta_n) = \\
                                                                    & v_1 \rightarrow \alpha_1', \dots, v_k \rightarrow \alpha'_k, v'_1 \rightarrow \beta_1, \dots, v'_n \rightarrow \beta_n           \\
      \forall h \in S_H ,{}                                 & h\cdot(v_1 \rightarrow \alpha_1, \dots, v_k \rightarrow \alpha_k, v'_1 \rightarrow \beta_1, \dots, v'_n \rightarrow \beta_n) = \\
                                                                    & v_1 \rightarrow \alpha_1, \dots, v_k \rightarrow \alpha_k, v'_1 \rightarrow \beta'_1, \dots, v'_n \rightarrow \beta'_n}
  \end{align*}
\end{restatable}

We will often use the shorthand $(G \times H, a \times b_{V_1 \cup V_2})$ to denote the action of the direct product.
The condition for when $(G \times H, a \times b_{V_1 \cup V_2})$ is a valid action is characterized by the following lemma:
\begin{restatable}{lemma}{LemDirProdCond}
\label{lem:dir-prod-cond}
  Let $I = G \times H$ be a direct product of two groups G and H with actions $\tuple{G}{a}{V_1}$, and  $\tuple{H}{b}{V_2}$. 
  Then $(G \times H, a \times b_{V_1 \cup V_2})$ is an action of $I$ on
  $\Sigma_{V_1 \cup V_2}$ if and only if for all program states $\sigma \in
  \Sigma_{V_1 \cup V_2}$, we have $\overline{a}_g(\overline{b}_h(\sigma)) =
  \overline{b}_h(\overline{a}_g(\sigma))$.
\end{restatable} 
When $V_1$ and $V_2$ are disjoint, then the condition in Lemma
\ref{lem:dir-prod-cond} is automatically satisfied. In this case, we will denote
the product group action as $\tuple{G \times H}{a \times b}{V_1  \uplus  V_2}$.

\subsubsection{Free Product}
We refer the reader to \iffull \Cref{app:fp-action} \else the full version of the paper \fi for the details of the free product action. 
\subsection{Syntax of Homomorphisms}
Recall from \Cref{sec:overview} that our Hoare triples were of the form
\(\triple{(G,a_{V_1})}{C}{(H,b_{V_2})}{\varphi} \). Now that we know how to
specify the group actions $(G, a)$ and $(H, b)$, we focus on the homomorphism
$\varphi$.
Much like for our group actions,
we provide a syntax for defining homomorphism. Let $\varphi : G \rightarrow H$ be a homomorphism where $G$ and $H$ are finitely presented groups, with generators $S_G$ and $S_H$, respectively. Then we define $\varphi$ in our syntax as: 
\[
\varphi \triangleq  \Big\langle \forall g \in S_G, \varphi(g) = h_i \cdots h_k, \text{ where } h_j \in S_H  \Big\rangle
\]
This syntax describes how $\varphi$ maps each generator of $G$ into some element
of $H$. For example, if we wanted to describe the homomorphism $\varphi: G
\rightarrow E$ that maps every element of $G$ to the identity in $E$, we would
define it as: 
\[
\varphi \triangleq \bigl\langle \forall g \in S_G, \varphi(g) = e   \bigr\rangle
\]
Since we will use the homomorphism from any group into the identity group $E$ often, we will denote it by $\mathbf{e*}$.
Another example is the homomorphism that maps every element to itself. For any group with generating set \(S_G\), we can define this map in our syntax as follows:
\[
\varphi \triangleq \bigl\langle \forall g \in S_G, \varphi(g) = g \bigr\rangle
\]
We also refer to this homomorphism extensively, so we will denote it as
\(\mathbf{eq}\), for \emph{equivariant}.

As with group actions, to verify if a map on generators is a homomorphism, it suffices to check if the map respects the relations: 
\begin{restatable}{lemma}{HomomorphismCheck}
\label{lem:johnson1}
\cite{johnson1997presentations}
Given $G=\langle S_G \mid R_G\rangle$ and $f:S_G\to H$, $f$ extends to a
homomorphism iff $f(r_j)=e_H$ for each $r_j\in R_G$.
\end{restatable}

\subsubsection{Direct Product}
To extend our syntax to capture a homomorphism from \(G\) to the direct product \(H \times I\), we can combine the data of the two homomorphisms, \(\varphi: G \to H\) and \(\psi: G \to I\) into one mapping that sends each generator \(g \in S_G\) to the pair \((\varphi(g), \psi(g))\). In our syntax, we can express it as follows:
\[
(\varphi\times \psi) \triangleq \Big\langle\forall g\in S_G,\, (\varphi\times \psi)(g)= \varphi(g)\,\psi(g)  \Big\rangle.
\]
Here, \( \varphi(g) \) is a word in the alphabet \( S_H \), \( \psi(g) \) is a word in the alphabet \( S_I \) and their concatenation \( \varphi(g)\,\psi(g) \) is a word in the combined alphabet \( S_H\cup S_I \). 
As expected, \(\phi \times \psi\) is a homomorphism. 
\begin{restatable}{lemma}{DirProdHomo}
        \(\phi \times \psi\) extends uniquely to a homomorphism from \(G\) to \( H 
    \times I\). 
\end{restatable}

We remark that similar constructions can be defined for homomorphisms from $G_1 \times G_2 \rightarrow H_1 \times H_2$, given the homomorphism $\varphi_1: G_1 \rightarrow H_2$, and $\varphi_2: G_2 \rightarrow H_2$.
\subsubsection{Free Product} We refer the reader \iflong to \Cref{app: fp-hom} \fi
for the free product.
\subsubsection{Composition of Homomorphisms}
Given two homomorphisms $\varphi: G \rightarrow H$, and $\psi: H \rightarrow I$, we define their composition, denoted by $\psi \circ \varphi : G \rightarrow I$ as a map that first applies $\varphi$ and then applies $\psi$. We define the syntax formally in \iflong \Cref{app: compose}. \fi
\subsection{Entailment of Group Actions}

When designing our Hoare logic, it will be useful to consider a notion of
weakening or entailment for group actions, analogous to the usual logical
implication. Intuitively, we think of one group action $\tuple{G}{a}{V}$ as
implying another group action $\tuple{H}{b}{V}$ if every transformation of
programs states in $\tuple{G}{a}{V}$ also exists in $\tuple{H}{b}{V}$. This
means that for every group element $g \in G$, there exists an element in $h$
such that they have the same action on every program state.  More formally:
\begin{definition}
  \label{def: implication}
  We say that $(G, a_{V_1}) \overset{\varphi}{\rightarrow} (H, b_{V_2})$ if and
  only if, $V_2 \subseteq V_1$ and $\varphi: G \rightarrow H$ is a homomorphism
  such that for all group elements $g \in G$ and program states $\sigma \in
  \Sigma$, we have
  \[
    \forall v_2 \in V_2, \overline{a_g}(\sigma)(v_1) = \overline{b_{\varphi(g)}}(\sigma)(v_2) .
  \]
\end{definition}
Intuitively, $(G, a_{V_1}) \overset{\varphi}{\rightarrow} (H, b_{V_2})$ holds
when $\tuple{H}{b}{V_2}$ contains more transformations of program state than
$\tuple{G}{a}{V_1}$ does.

\begin{example}
  Consider the group of integers, $\mathbb{Z}$,  which acts on each variable by integer addition. Let ${\tuple{\mathbb{Z}}{a}{\{x\}}}$ denote this action where for every $g \in \mathbb{Z}$, $x$ is transformed to $x+g$.

  We can also consider the group $2\mathbb{Z}$, which is the additive group of even integers. We define another action of this group:
  $ {\tuple{2\mathbb{Z}}{b}{\{x\}}}$ such that for every $g \in 2\mathbb{Z}$, $x$ is transformed to $x+g$.

  Consider the map $\varphi: 2\mathbb{Z} \rightarrow \mathbb{Z}$ such that
  $\varphi(g \in 2\mathbb{Z}) = g \in\mathbb{Z}$. The map $\varphi$ is a homomorphism, and
  \[
    \forall g \in 2\mathbb{Z}\; \forall \sigma \in \Sigma,
    \overline{a}_g(\sigma)(x) = b_{\varphi(g)}(\sigma)(x) ,
  \]
  so $\tuple{2\mathbb{Z}}{b}{\{x\}} ) \overset{\varphi}{\rightarrow} \tuple{\mathbb{Z}}{a}{\{x\}}$.

\end{example}
This definition of implication satisfies several convenient properties resembling
properties of logical implication. For instance, if we view the direct product
as logical conjunction, we would expect $\tuple{G \times H}{a \times
    b}{V_1\uplus V_2} \overset{\varphi}{\rightarrow} \tuple{G}{a}{V_1}$, since $A \land B \implies A$.
Indeed, we have:

\begin{restatable}{lemma}{ImpProp}
  Let $\tuple{G}{a}{V_1}$ and $\tuple{H}{b}{V_2}$ be any group actions. Then the following properties hold:
  \begin{enumerate}
    \item $\tuple{G \times H}{a \times b}{V_1\uplus V_2} \overset{\mathbf{proj_1}}{\rightarrow} \tuple{G}{a}{V_1}$
    \item $\tuple{G \times H}{a \times b}{V_1\uplus V_2} \overset{\mathbf{proj_2}}{\rightarrow} \tuple{H}{b}{V_2}$
    \item  $\tuple{E}{e}{V_1} \overset{\mathbf{e^-}}{\rightarrow} \tuple{G}{a}{V_1}$
    \item  $  \tuple{G}{\overline{a}}{\mathit{Vars}} \overset{\mathbf{eq}}{\rightarrow} \tuple{G}{a}{V_1}$
  \end{enumerate}
where \(\mathbf{proj_1} \), and \(\mathbf{proj_2}\) are the first and second projections, respectively, and $\mathbf{e^-}$ is the map $\mathbf{e^-}(e) = e_G$.
\end{restatable}

We will use these properties extensively when reasoning about group actions.

\section{Symmetry Program Logic}
\label{sec:logic}
With the syntax for group actions and homomorphisms set, we are now ready to
present our Hoare-style logic for symmetry properties.

\subsection{Programming Model}
We start by defining a simple imperative programming language with no side
effects: The syntax of commands in the language is: 
\begin{align*}
  \mathcal{COM}  \ni C \coloneqq{}& \textbf{skip} \mid x \coloneqq e  \mid C_1; C_2 \mid \textbf{if }x\;\textbf{then}\;C_1\;\textbf{else}\;C_2\\& \mid \textbf{for } {(t \coloneq 0; t < \alpha_1; t \coloneq t + \alpha_2)} \textbf{ do } C  \mid \textbf{while } x \textbf{ do } C
\end{align*}
Our syntax is standard. The language supports simple commands like assignment,
sequencing, branching, and iteration through for-loops and while loops, however,
we require all control flow to branch on boolean \emph{variables}; this can be
achieved by assigning the test to an auxiliary variable. The expressions are
standard boolean and arithmetic expressions with operations from a fixed set
$Ops$. Loops are assumed to be terminating.  We freely use logical variables,
such as $\alpha_1$ and $\alpha_2$ in the loop constructs, to refer to parameters
that remain unchanged throughout the program.

The full syntax and (entirely standard) semantics can be found in \iffull
\Cref{app:sem,app:syntax,}. \else the full version. \fi

\subsection{Symmetry Triples}
The formulas in our logic are Hoare-style triples
consisting of a group action $\tuple{G}{a}{V}$ defined on the input states, a
program $C$, and a group action $\tuple{H}{b}{V}$ on the output program states.
This triple is valid when for all input states acted on by $\tuple{G}{a}{V}$,
the output states after executing $C$ are related by the action of $\tuple{H}{b}{V}$. More formally:
\begin{definition}
  \label{def: valid}
  The triple $\triple{(G, a_{V_1})}{C}{(H, b_{V_2})}{\varphi}$ is valid, written
  $\vDash \triple{\tuple{G}{a}{V_1}}{C}{\tuple{H}{b}{V_2}}{\varphi}$,
  if and only if for all $g \in G$, we have
  \[ \forall \sigma \in \Sigma,\forall v \in
    V_2\;\overline{b}_{\varphi(g)}(\sem{C}_{\sigma})(v)=
    \sem{C}_{\overline{a}_{g}(\sigma)}(v) .
  \]
\end{definition}
Intuitively, the homomorphism $\varphi$ witnesses the symmetry property:
for every group element $g \in G$, transforming the input state $\sigma$ by $a_g$ and
running the program leads to the same program state as running the program on
$\sigma$ and transforming by $b_{\varphi(g)}$.

\begin{example}
  \label{ex:triple-int}
  Consider the program $x \coloneqq x + 5$. Suppose before executing this
  statement, the inputs were transformed by the group of integers acting on
  each variable by integer addition, i.e., for any integer $c$, $x$ is
  transformed to $x + c$. This action is defined in our syntax as follows:
  \[
    {\tuple{\mathbb{Z}}{a}{\{x\}}} \triangleq
    \actiontemp{g}{-}{\{x\}}{g(x \rightarrow \alpha) = x \rightarrow \alpha + 1 }
  \]
  Then, after executing the program, the value of $x$ without the group action will be $x+5$, and with the group action for any $g
    \in \mathbb{Z}$ will be $x + 5 + g$. Therefore, for any $g$, all the outputs
  are related by integer addition. Thus, the triple $\triple{\tuple{\mathbb{Z}}{a}{\{x\}}}{x \coloneqq x + 5}{\tuple{\mathbb{Z}}{a}{\{x\}}}{\mathbf{eq}}$ is valid. 
\end{example}

\subsection{Proof Rules}
Next, we consider the proof rules for our logic. As usual, we write $\vdash
\triple{\tuple{G}{a}{V_1}}{C}{\tuple{H}{b}{V_2}}{\varphi}$ when the triple
$\triple{\tuple{G}{a}{V_1}}{C}{\tuple{H}{b}{V_2}}{\varphi}$ is derivable from
the set of inference rules.

\subsubsection{Program Statement Rules}
We first present the basic rules for each construct in our programming language, which are summarized in Figure \ref{fig:rules-stmt}.
All rules implicitly require that assertions are valid group actions.
The $\textsc{SKIP}$ rule says that if a group action transforms the variables in $V$, then after a $\mathbf{skip}$ statement, the variables are still related by the same group action.

To understand the $\textsc{ASSGN}$ rule, we first introduce a predicate
transformer $\mathsf{POST}$ that takes in a group action $(G, a)$, and function
$f$ and returns a group action $\mathsf{POST}((G, a), f)$. This group action
will be based on the pre-condition group $G$, but its action will be different
from $a$.

Let us think about what this group action should do. Let $C$ be the assignment
command $v_1 \coloneq f(v_1, \dots , v_k)$; assigning to other variables is
entirely similar. For any group element $g \in G$, and
starting program state $\sigma$, we want the post-condition group action at $g$ to map
the output program state $\sem{C}_{\sigma}$ to the output program state
$\sem{C}_{a_g(\sigma)}$. To define this post-condition group action, we first
note that the program $C$ is invertible because $f$ is assumed to be injective
with respect to $v_1$, i.e.,
\[
  \forall v_1, v'_1 \dots, v_k. f(v_1, \dots, v_k) = f(v_1', \dots, v_k) \implies v_1 = v'_1 .
\]
Thus at a high level, the group action can be defined by inverting $C$, applying
the original action $(G, a)$, and then applying $C$. In pictures, we want:
\[\begin{tikzcd}
    \Sigma & {} & \bullet && \bullet && \Sigma
    \arrow["{C^{-1}}", from=1-1, to=1-3]
    \arrow["{a_g}", from=1-3, to=1-5]
    \arrow["C", from=1-5, to=1-7]
    \arrow["{POST((G, a), f)_g}"', curve={height=18pt}, from=1-1, to=1-7]
  \end{tikzcd}\]
Before we define $\mathsf{POST}$ formally, we introduce some notation.
Let $\vec{z}$ denote $v_2, \dots, v_k$.

Let $n$ be $|V|$, and $m$ be $|\mathit{Vars}|$.
Let $\Sigma_f$ be the set of all states reachable after $C$.
We can define an inverse map $C^{-1} : \Sigma_f \to \Sigma$ via:
\begin{multline*}
  C^{-1}(v_1 \rightarrow \alpha_1, \dots, v_n \rightarrow \alpha_n, \dots, v_m \rightarrow \alpha_m)  \\
  \triangleq
  (v_1 \rightarrow \hat{f}(\alpha_1, \alpha_2, \dots, \alpha_k) , \dots, \alpha_n \rightarrow \alpha_n, \dots, v_m \rightarrow \alpha_m)
\end{multline*}
where $\hat{f} : \mathbb{Z}^k \to \mathbb{Z}$ is the function satisfying $\hat{f}(f(\alpha_1, \dots, \alpha_k), \alpha_2, \dots, \alpha_k) = \alpha_1$ for
all $\alpha_i \in \mathbb{Z}$.

Now, we can define the action returned by $\mathsf{POST}$ as the action that inverts the program state, applies the action $\tuple{G}{a}{V}$, and then executes the assignment:

\begin{restatable}{definition}{POSTDef}
  \label{def:post}
  Let $(G,a_V)$ be any group action with \(|V| = n\), $f$ be any injective function, and $\sigma$
  be any program state $(v_1 \rightarrow \alpha_1, \dots, v_n
    \rightarrow \alpha_n, \dots, v_m \rightarrow \alpha_n)$. Then $\mathsf{POST}((G,a_V), f)$ is defined as:
    \begin{align}
      \mathsf{POST}((G,a_V), f) &\triangleq \actiontemp{S_G}{R_G} {\mathit{Vars}}
      { && \forall g \in S_G , g \cdot \sigma =
        \notag \\
        & && v_1 \rightarrow f(\overline{a}_g(C^{-1}(\sigma))(v_1), \dots,
        \overline{a}_g(C^{-1}(\sigma))(v_k)),
        \notag \\
        & && \dots, v_n \rightarrow \overline{a}_g(C^{-1}(\sigma))(v_n), \dots, v_m
      \rightarrow \alpha_m}
      \tag{if $\sigma \in \Sigma_f$} \\
        \mathsf{POST}((G,a_V), f) &\triangleq \actiontemp{S_G}{R_G}
        {\mathit{Vars}} {&&\forall g \in S_G , g \cdot \sigma = \sigma}
      \tag{if $\sigma \notin \Sigma_f$}
      \end{align}
\end{restatable}

The action first computes the starting state and then updates $v_1$ with the result of applying $f$ to the relevant variables in the start state, which gives the program state. In the case that the program state does not belong to $\Sigma_f$, the action returned by $\mathsf{POST}$ behaves like the identity.
We can show that the map produced by the $\mathsf{POST}$ transformer is a group action.

\begin{restatable}{theorem}{PostGpActionThm}
  \label{thm:post-is-grp-action}
  If $(G,a_V)$ is a group action and $f$ is injective, then $\mathsf{POST}((G,a_V),f)$ is a group action.
\end{restatable}  

\begin{example}
  Using the $\textsc{ASSGN}$ rule, we show how to derive the post-condition in Example \ref{ex:triple-int} systematically.  The statement $x \coloneq x+5$ is injective, and its inverse is $\hat{f}(y) = y - 5$.  Now we can write this action in the syntax as:
  \begin{align*}
    \actiontemp{g}{-} {\{x\}}{ g \cdot (x \rightarrow \alpha) = x \rightarrow f(g(x \rightarrow \alpha -5)(x))} ,
  \end{align*}
  which simplifies to the action we derived in Example \ref{ex:triple-int}:
  $ \actiontemp{g}{-} {\{x\}}{ g \cdot (x \rightarrow \alpha) = x \rightarrow \alpha + 1}$
\end{example}

\begin{figure}[t]
  \begin{mathpar}
    \inferrule*[left={$\textsc{SKIP}$}, right={}]{
      \;
    }
    {\triple{(G, a_{V})}{\mathbf{skip}}{(G, a_{V})}{\mathbf{eq}}}
    \and
    \inferrule*[left={$\textsc{ASSGN}$}, right={}]{f \text{  is injective}}
    {\triple{(G, a_{V})}{ v_i \coloneq f(v_1, \dots , v_k)}{ \mathsf{POST}(\tuple{G}{a}{V}, f)}{\mathbf{eq}}}
    \and
    \inferrule*[left={$\textsc{SEM-ASSGN}$}, right={}]{\;\;\vDash \triple{(G, a_{V_1})}{  v_1 \coloneq f(v_1, \dots , v_k) }{ \tuple{H}{b}{V_2}}{\varphi}\;\;}
    {\triple{(G, a_{V_1})}{  v_1 \coloneq f(v_1, \dots , v_k)  }{\tuple{H}{b}{V_2}}{\varphi}}
    \and
    \inferrule*[left={$\textsc{SEQ}$}, right={\;}]{\triple{(G, a_{V_1})}{ C_1
    }{(H, b_{\mathit{Vars}}) }{\varphi} \\\\ \triple{(H, b_{\mathit{Vars}})}{ C_2}{(I, c_{V_2})}{\tau}}
    {(\triple{G, a_{V_1})}{C_1;C_2}{(I, c_{V_2})}{\tau \circ \varphi}}
    \and
\inferrule*[left={IF}, right={}]{ \triple{(G, a_{V_1})}{C_1}{(H, b_{V_2})}{\varphi} \quad 
     \triple{(G, a_{V_1})}{C_2}{(H, b_{V_2})}{\varphi} {(G, a_{V_1})} \quad (G, a_{V_1}) \overset{\mathbf{e*}}{\rightarrow}{(E, e_{\{x\}})}}{\triple{(G, a_{V_1})}{\textbf{if }x\;\textbf{then}\;C_1\;\textbf{else}\;C_2}{(H, b_{V_2})}{\varphi}}
     \and
    \inferrule*[left={$\textsc{FOR}$},]{ \triple{(G, a_{V})}{C }{(G,
    a_{V})}{\varphi} \quad  \;\; \mathsf{VARS}(C) \subseteq V \quad \text{Loop
Terminates in $n$ steps}}{\triple{(G, a_{V})}{ \textbf{for } {(t \coloneq 0; t <
\alpha_1; t \coloneq t + \alpha_2)} \textbf{ do } C }{(G, a_{V})}{\varphi^{n}}}
    \and
    \inferrule*[left={WHILE}, right={}]{ \triple{(G, a_{V})}{C}{(G, a_{V})}{\mathbf{eq}} \quad {(G, a_{V})} \overset{\mathbf{e*}}{\rightarrow}{(E, e_{\{x\} })}  \quad \{b\} \subseteq  V \quad \mathsf{VARS(C)} \subseteq V}{\triple{(G, a_{V})}{\textbf{while }x\;\textbf{do }C} {(G, a_{V})}{\textbf{eq}}}
  \end{mathpar}
  \caption{Basic Inference Rules}
  \label{fig:rules-stmt}
\end{figure}

The $\textsc{ASSGN}$ rule only applies to injective assignment statements. For
non-injective assignments, the situation is more challenging: we will show later
in this section that there are some pre-conditions that have \emph{no}
valid post-condition. For these assignment statements, the best we can do is
fall back on the general rule $\textsc{SEM-ASSGN}$, which has a semantic side
condition. In \Cref{sec: impl}, we will consider how to check this semantic
condition automatically.

The $\textsc{SEQ}$ rule is similar to the standard Hoare logic rule. It allows us to compose triples when the post-condition and pre-condition match. However, the intermediate assertion (denoted by $\tuple{H}{b}{\mathit{Vars}}$) must be about the set $\mathit{Vars}$, and not any arbitrary set of variables.

The $\textsc{IF}$ rule requires that both branches have the same symmetry property. The condition $x$ must be a boolean variable which is preserved by the group action as ensured by the premise ${(G, a_{V_1})} \overset{\mathbf{e*}}{\rightarrow}{(E, e_{\{x\}})}$.

The iteration rule $\textsc{FOR}$ requires that the body of the loop, $C$, must have the same pre- and post-condition which means that if the variables are acted on by a group action $\tuple{G}{a}{V}$, then after running $C$, the variables in $V$ must still be related by the same group action. This is reminiscent of a loop invariant in Hoare logic. The side condition states that ${\mathsf{VARS}(C)}$, the set of variables modified by $C$ is a subset of $V$. Since our language does not have side effects, this set can be easily computed. When all these conditions are met,  once the loop terminates,  the variables in $V$ are still related by the group action on the inputs. Additionally, the homomorphism is computed by composing $\varphi$ times with itself $n$ times.  

 The $\textsc{WHILE}$ rule requires the loop body to preserve the input group action with the homomorphism $\mathbf{eq}$. This is to ensure we can construct the same homomorphism for a different number of iterations. Similar to the $\textsc{IF}$ rule, the guard must be a variable that is preserved by the group action, and similar to the $\textsc{FOR}$ rule, ${\mathsf{VARS}(C)}$, the set of variables modified by $C$ must be a subset of $V$.
\begin{figure}[t]
  \begin{mathpar}
    \inferrule*[left={$\mathsf{CONST}$}, right={}]{\mathit{VARS}(C) \cap V = \emptyset }{\triple{(G, a_{V})}{ C }{(G, a_{V})}{\mathbf{eq}}}  \quad
    \inferrule*[left={$\textsc{LIFT}$}]{\triple{\tuple{G}{a}{V_1}}{C}{\tuple{H}{b}{V_2}}{\varphi}}{\triple{\tuple{G}{\overline{a}}{\mathit{Vars}}}{C}{\tuple{H}{b}{V_2}}{\varphi}}   \quad
    \inferrule*[left={$\mathsf{ID}$}, right={}]{\;}{\triple{(G, a_{V})}{C}{ (E, e_{\mathit{Vars} \setminus (V \cup \mathsf{Vars}(C))})}{\mathbf{e*}}} \\
    \inferrule*[left={$\textsc{DIR-PROD}$}, right={}]{\triple{(G, a_{V_1})} {C_1}{(H,b_{V_2}) }{\varphi_1}\;\;\;\;\ \triple{(G, a_{V_1})}{C_1}{(I, c_{V_3})}{\varphi_2} \quad V_2 \cap V_3 = \emptyset}
    {\triple{(G, a_{V_1})}{ C_1 }{(H \times I, b \times c_{V_2 \uplus V_3})}{(\varphi_1 \times \varphi_2)}} \\
    \inferrule*[left={$\textsc{CONS-1}$},right={}]{\triple{\tuple{G}{a}{V_0}}{C}{\tuple{H}{b}{V_1}}{\varphi}
      \quad  \tuple{H}{b}{V_1} \overset{\tau}{\rightarrow} \tuple{I}{c}{V_2}}{
      \triple{\tuple{G}{a}{V_0}}{C}{\tuple{I}{c}{V_2}}{\tau \circ \varphi}}   \and
    \inferrule*[left={$\textsc{CONS-2}$},right={}]{\triple{\tuple{G}{a}{V_1}}{C}{\tuple{H}{b}{V_2}}{\varphi}
      \quad  \tuple{I}{c}{V_1} \overset{\tau}{\rightarrow} \tuple{G}{a}{V_1}}{
      \triple{\tuple{I}{c}{V_1}}{C}{\tuple{H}{b}{V_2}}{\varphi \circ \tau}}
  \end{mathpar}
  \caption{Structural Rules}
  \label{fig:rules-struct}
\end{figure}
\subsubsection{Structural Rules}
We now present the structural rules of our logic in Figure \ref{fig:rules-struct}. The rule of constancy, $\textsc{CONST}$ says if the variables in $V$ are not modified by the command $C$, then, after executing the command, the variables in $V$ will still be related by $\tuple{G}{a}{V}$.  The $\textsc{LIFT}$ rule allows the pre-condition to be replaced by its lifting. Since many rules require the set $\mathit{Vars}$ in the pre-condition, this rule is useful in ensuring the assertions are of the correct form. The $\mathsf{ID}$ rule says that any variable that is unchanged by the command $C$, and the group action $(G,a_{V})$, can be mapped by the trivial group action.

The product rules allow us to build more complicated group actions in
the post-condition. The \textsc{DIR-PROD} rule allows us to combine two
post-condition group actions, so long as the set of variables they act on are
disjoint. The disjointness also guarantees that the direct product group action
is indeed a group action (\Cref{lem:dir-prod-cond}).  This rule is useful for reasoning about different parts of the program memory. 
\iffull We also have a free product rule presented in \Cref{app:rules}. \else We also have a free product rule presented in the full version.  \fi

Finally, we consider the rules of consequence. The rule $\textsc{CONS-1}$ says
that if the group action relating the output states implies another group
action, then we can replace it with the action it entails. $\textsc{CONS-2}$ is
similar but strengthens the pre-condition: if a triple holds for any transformation by the group action $\tuple{G}{a}{V_1}$, then the triple holds for any smaller set of transformations. This rule requires that the set of variables being acted on by the pre-condition group action remain the same.

As expected, the rules presented are sound.

\begin{restatable}{theorem}{Soundness} (Soundness)
  \label{thm:soundness}
  If  \;$\vdash \triple{\tuple{G}{a}{V}}{C}{\tuple{H}{b}{V}}{\varphi}$ then $
    \vDash  \triple{\tuple{G}{a}{V}}{C}{\tuple{H}{b}{V}}{\varphi}.$
\end{restatable}

\subsection{Weakest Pre-condition and Strongest Post-Condition}
In standard Hoare logic, the weakest pre-condition and strongest
post-condition transformers are powerful tools that aid verification. 
In our setting, the weakest pre-condition and strongest post-condition transform
\emph{group actions}, and it is not clear how to define these notions. In this
section, we show how to compute these actions for assignment commands.

\subsubsection{Weakest Pre-Condition}
Intuitively, one group action is weaker than another if it contains all the
transformations in the other group action. In other words, the group action
$\tuple{H}{b}{V}$ is weaker than $\tuple{G}{a}{V}$ if $ \tuple{G}{a}{V}
\overset{\varphi}{\rightarrow} \tuple{H}{b}{V}$. Therefore, the weakest
pre-condition is the ``largest'' group action for which a post-condition is valid. The trivial group action that leaves the program state unchanged is always a valid pre-condition. Our discussion of the weakest pre-condition focuses on a class of assertions called \emph{faithful} group actions.


\begin{definition}
  \label{def: faithful}
  A group action $\tuple{G}{a}{V}$ is \emph{faithful} if
  $a_g(\sigma) = \sigma$ for all $\sigma$ implies that $g = e_G$.
  Equivalently, no two elements of $G$ act the same
  way on every element of the set $\Sigma_V$.
\end{definition}
Most natural group actions are faithful. Further, any group action can be converted to an equivalent faithful group action. When the post-condition action is faithful, then there is a unique homomorphism that makes a triple valid. We will use this property in our weakest pre-condition, and our automated verifier.
\begin{restatable}{theorem}{FaithfulUnique}
\label{thm: faithful-unique}
    If $\triple{(G,a_{Vars})}{C}{(H,b_{Vars})}{\varphi}$, and $(H, b_{Vars})$ is faithful, then $\varphi$ is unique.
\end{restatable}
We consider two versions of the weakest pre-condition. We focus on one variant called the weakest faithful pre-condition here, presenting the other one in \iffull \Cref{app: wp}. \else the appendix of the full version. \fi
In this variant of the weakest pre-condition, we assume \emph{the post-condition group acts faithfully on the set of reachable states}. Let $\Sigma_C$ be the set of all states reachable by a command $C$. We assume $\tuple{H}{b}{V}$ is a faithful group action on the set $\Sigma_C$. This means if  for all  $\sigma \in \Sigma_C, {\overline{b}}_h(\sigma) = \overline{b}_{h'}(\sigma)$, then $h' = h$.
\begin{definition}
  \label{def: wp-post-faithful}
  The \emph{weakest faithful pre-condition} for a command $C$ and a post-condition $\tuple{H}{b}{\mathit{Vars}}$ that acts \emph{faithfully} on the set $\Sigma_C$ is a group action, denoted by $\mathsf{fwp}(\tuple{H}{b}{\mathit{Vars}}, C)$, such that\\ $\triple{(\mathsf{fwp}(\tuple{H}{b}{\mathit{Vars}}, C))}{C}{\tuple{H}{b}{\mathit{Vars}}}{\varphi}$ is valid for some homomorphism $\varphi$ and
  for any group action $\tuple{I}{c}{\mathit{Vars}}$ and homomorphism $\tau$, if $\triple{\tuple{I}{c}{\mathit{Vars}}}{C}{\tuple{H}{b}{\mathit{Vars}}}{\tau}$ is valid then there exists a homomorphism $\eta$, such that $\tuple{I}{c}{\mathit{Vars}} \overset{\eta}{\rightarrow} \mathsf{fwp}(\tuple{H}{b}{\mathit{Vars}}, C)$ and $\tau = \varphi \circ \eta$.
\end{definition}
This definition is analogous to the traditional definition of weakest pre-condition in the sense that any valid pre-condition entails the weakest pre-condition. 

We will construct a group action that will serve as the weakest faithful
pre-condition. We begin a group $\mathsf{G^*}(\tuple{H}{b}{\mathit{Vars}}, C)$
by assembling all valid transformations of the input that ensure the
post-condition. More formally:  
\begin{restatable}{definition}{GStar}
  \label{def :gstar}
  Let $C$ be a command and $\tuple{H}{b}{\mathit{Vars}}$ be any faithful action on the set
  $\Sigma_C$. Let:
  \begin{align*}
    \mathsf{G^*}(\tuple{H}{b}{\mathit{Vars}}, C) \triangleq
    \{ & g \mid g \text{ is a bijection from } \Sigma \rightarrow \Sigma  \text{
    such that } \\& \qquad\exists h \in H, \forall \sigma \in \Sigma, b_h(\sem{C}_{\sigma}) = \sem{C}_{g(\sigma)}\}
  \end{align*}
\end{restatable}

Next, we show that this is indeed a group.
\begin{restatable}{theorem}{WPGroup}
  The set $\mathsf{G^*}(\tuple{H}{b}{\mathit{Vars}}, C)$ is a group under function composition.
\end{restatable}  
We now define the action of $\mathsf{G^*}(\tuple{H}{b}{\mathit{Vars}}, C)$ based on the group $\mathsf{G^*}(\tuple{H}{b}{\mathit{Vars}}, C)$.
This will be the \emph{weakest faithful pre-condition}, denoted by $\mathsf{fwp(\tuple{H}{b}{\mathit{Vars}}, C)}$.
\begin{restatable}{definition}{WpTwoAction}
  \label{def: wp2-action}
The action $\mathsf{fwp}(\tuple{H}{b}{\mathit{Vars}}, C) \triangleq (\mathsf{G^*}(\tuple{H}{b}{\mathit{Vars}}, C), a_{\mathit{Vars}})$ is defined for any $g \in \mathsf{G^*}(\tuple{H}{b}{\mathit{Vars}}, C)$ as $a_g(\sigma) = g(\sigma)$.
\end{restatable}
The constructed group action is indeed the weakest faithful pre-condition.
\begin{restatable}{theorem}{GPWFaithfulPre}
  $\tuple{\mathsf{G^*}(\tuple{H}{b}{\mathit{Vars}}, C)}{a}{\mathit{Vars}}$ is the weakest faithful pre-condition.
\end{restatable}  

\paragraph*{Expressibility of Weakest Pre-Condition}
In general, the weakest faithful pre-condition group
$\mathsf{G^*}(\tuple{H}{b}{\mathit{Vars}}, C)$ may not be finitely generated.
For example, if we consider a constant assignment $x \coloneq 5$ with the
identity group as the post-condition, then for every bijection on program
states, the triple will hold. Therefore, $G^*(\tuple{E}{e}{\mathit{Vars}}, x
\coloneq 5)$ is the group $Sym(\Sigma)$, which is not finitely generated.

Although the full group might not be expressible in our syntax, any finite
subset of $G^*(\tuple{H}{b}{\mathit{Vars}}, C)$ generates a group action that
is a valid pre-condition. We formalize this observation as
follows:
\begin{definition}
  \label{def:grp-finite}
  Let $H$ be a finitely presented group with presentation
\(H = \langle h_1, \dots, h_n \mid R_H \rangle ,
\)
and let $G = \{g_1, \dots, g_m\}$ be a set of symbols, each $g_i \in G$ co-responding to a bijection
\(
g_i \in \mathsf{G^*}(\langle H,b,V_1\rangle, C)
\)
such that for some generator $h_i \in \{h_1,\dots,h_n\}$ of $H$,
\(
b_{h_i}(\sem{C}_{\sigma}) =\sem{C}_{\,g_i(\sigma)} ,
\)
for all states $\sigma$. We then define the group
\(
\hat{G}((H,b),C) \triangleq \langle G \mid - \rangle ,
\)
that is, the free group generated by the symbols $G$.
\end{definition}

We are now ready to define the action of $\hat{G}((H,b),C)$. 
\begin{definition}
  \label{def:grp-finite-action}
  We define the action of the free group $\hat{G}((H,b),C)$, denoted by
  $(\hat{G}((H,b),C), g^*_{\mathit{Vars}})$  as follows: For every word $g
  \triangleq g_1 \cdots g_n$, we define the action
  $g^*_g(\sigma) \triangleq g \triangleq g_1 \circ \cdots \circ g_n(\sigma)$.
\end{definition}
Indeed, this forms a valid pre-condition.
\begin{restatable}{theorem}{FiniteWP}
\label{thm: finite-wp}
  If ($H,b_{\mathit{Vars}})$ is faithful then $\triple{(\hat{G}((H,b),C), g^*_{\mathit{Vars}})}{C}{(H,b_{\mathit{Vars}})}{\varphi}$ is a valid triple, where $\varphi : \hat{G}((H,b),C) \rightarrow H$ is the unique homomorphism for which the triple can be sound.
\end{restatable}  

Thus, our construction of the weakest faithful pre-condition gives us a way to
construct finitely presented, valid pre-conditions, which we use in our
implementation in \Cref{sec: impl}.
\subsubsection{Strongest Post-Condition}
Next, we consider the strongest post-condition for an assignment statement. A
group action $\tuple{H}{b}{V}$ is stronger than $\tuple{G}{a}{V}$ if $\tuple{H}{b}{V} \overset{\varphi}{\rightarrow} \tuple{G}{a}{V}$, i.e., $\tuple{G}{a}{V}$ contains at least as many transformations as $\tuple{H}{b}{V}$.  The strongest post-condition is the smallest group action which is a valid post-condition.
\begin{definition}
  \label{def: sp-def}
  The \emph{strongest post-condition} with respect to a pre-condition $\tuple{G}{a}{V_1}$, a command $C$,  denoted by $\mathsf{sp}(\tuple{G}{a}{V_1}, C)$,  is a group action such that $\triple{\tuple{G}{a}{V_1}}{C}{(\mathsf{sp}(\tuple{G}{a}{V_1}, C))}{\varphi}$ and
  for any \emph{faithful} group action $(I,c_{V_2})$ if $\triple{\tuple{G}{a}{V_1}}{C}{\tuple{I}{c}{V_2}}{\eta}$ then $\mathsf{sp}(\tuple{G}{a}{V_1}, C) \overset{\tau}{\rightarrow} \tuple{I}{c}{V_2}$ where $\eta = \tau \circ \varphi$
\end{definition}
The situation for the strongest post-condition is different from the weakest pre-condition -- for a given pre-condition, a valid post-condition may not exist.
\begin{restatable}{theorem}{PostDoesNotExist}
  \label{thm: post-not-exist} Let  $v_1 \coloneq f(v_1, \vec{z})$  be an assignment statement, where $\vec{z}$ denotes variables $v_2 \dots, v_n$, which are the rest of the arguments to $f$.
  If $f$ is not injective, then there exists a group action $\tuple{G}{a}{V}$
  such that for any homomorphism $\varphi$, there is no group action $(H,b_{V})$, making the following triple valid:
  \[\triple{\tuple{G}{a}{V}}{v_1 \coloneq f(v_1, \vec{z})}{(H,b_{V})}{\varphi}.\]
\end{restatable}  
This negative result shows that we generally cannot hope for a strongest
post-condition transformer. However, from the $\textsc{ASSGN}$ rule, we know
that when $f$ is injective a post-condition always exists. In this situation, we
can show that the predicate transformer $\mathsf{POST}((G,a_{V}), f)$ is indeed
the strongest post-condition in the sense of Definition \ref{def: sp-def}.

\begin{restatable}{theorem}{PostIsSp}
  \label{thm: post-is-sp}
  $\mathsf{POST}((G,a_V), f)$ is the strongest post-condition for injective assignment statements of the form $v_1 \coloneq f(v_1 \dots, v_n)$.
\end{restatable}  

\section{Case Studies}
\label{sec:examples}
In this section, we use our logic to verify the symmetry properties formally for a range of programs. Many of our examples are taken from \citet{10.1007/978-3-030-31784-3-6}, where these symmetry properties are used to verify safety properties of systems. We highlight the interesting parts of the proofs for four examples,
showing the details; more examples can be found in \iffull \Cref{app: Examples}. \else the full version of the paper. \fi 
\subsection{Translation of a Car System}
As a warm-up, we return to a simple car system from Section \ref{sec:overview}.
Recall the program in Figure \ref{fig: car-trans} where $x$ and $y$ represent the coordinates of the car, $v$ is the velocity, the variable $\phi$ is the steering angle, and $\theta$ is the angle between the car and the x-axis. At every time step, these variables are updated based on the control parameters $u$, $a$, and $L$. 

The program simulates the motion of the car by updating the car's state for $T$, time steps. The variable $dt$ refers to a small constant. 

\paragraph*{Symmetry Property}
Recall from Section \ref{sec:overview}, the program simulating the motion of the
car is equivariant under translation --- if the starting coordinates $(x,y)$ are
shifted by some constant $(c,c)$, the final coordinates $(x',y')$ will also be
shifted by $(c, c)$ to $(x'+c, y'+c)$. To prove this formally, we model
coordinate transformations as a particular
action of the group of integers.
Consider the action of the group of integers on the variables $x$ and $y$, where the generator acts by adding 1 to $x$ and $y$:
\[\tuple{\mathbb{Z} 
    }{a}{\{x,y\}} \triangleq\actiontemp{g}{-}{\{x,y\}}{g(x \rightarrow \alpha_1, y \rightarrow \alpha_2) = x \rightarrow \alpha_1 + 1, y \rightarrow \alpha_2 + 1 }\]
The lifting of this action is denoted by $\tuple{\mathbb{Z}
}{a}{\mathit{Vars}}$; this action modifies $x$ and $y$, and leaves all other
variables unchanged.  Then, we would like to prove that:
\[ 
  \triple{\tuple{\mathbb{Z} 
    }{a}{\mathit{Vars}}}{\textbf{for } {(t \coloneq 0; t < T; t \coloneq t +dt)} \textbf{ do } C}{\tuple{\mathbb{Z}}{a}{\mathit{Vars}}}{\mathbf{eq}}
\]
This triple states that if integer $c$ acts on the starting coordinates of the car $(x,y)$ by adding $c$ to $x$ and $y$, then the final coordinates are also shifted by $c$. 
\paragraph{Line 2}
We start with the first statement in the loop body on line 2: $c_2 \triangleq x \coloneq v \cdot \cos{(\theta)}\cdot dt + x$. Since this injective in the argument $x$, we can apply our $\textsc{ASSGN}$ rule to compute $\mathsf{POST((\mathbb{Z},a), c_2)}$. From the definition of $\mathsf{POST}$, this is 
\begin{align*}
     \mathsf{POST}(\tuple{\mathbb{Z}}{a}{\mathit{Vars}}, c_2) \triangleq \actiontemp{g}{-} {\mathit{Vars}} {g \cdot &(x \rightarrow \alpha_f, y \rightarrow \alpha_y, \dots \theta \rightarrow \alpha_{\theta}) =\\&
                       x \rightarrow c_2(a_g(c_2^{-1}((x \rightarrow \alpha_f, y \rightarrow \alpha_y, \dots \theta \rightarrow \alpha_{\theta})))), \\&  y \rightarrow \alpha_y + 1, 
                        \dots, \theta \rightarrow \alpha_{\theta}}
\end{align*}
This simplifies to
\begin{align*}
        \mathsf{POST}(\tuple{\mathbb{Z}}{a}{\mathit{Vars}}, c_2) \triangleq \actiontemp{g}{-} {\mathit{Vars}} {g \cdot &(x \rightarrow \alpha_f, y \rightarrow \alpha_y, \dots \theta \rightarrow \alpha_{\theta}) =\\&
                       x \rightarrow \alpha_f + 1,  y \rightarrow \alpha_y + 1, 
                        \dots, \theta \rightarrow \alpha_{\theta}},
\end{align*}
which is the same as $\tuple{\mathbb{Z}}{a}{\mathit{Vars}}$. Therefore, we can conclude:
$
    \triple{\tuple{\mathbb{Z}}{a}{\mathit{Vars}}}{ x \coloneq v \cdot \cos{(\theta)} \cdot dt + x}{\tuple{\mathbb{Z}}{a}{\mathit{Vars}}}{\mathbf{eq}}$.
\paragraph{Line 3}
We can use a similar computation on Line 3 to prove: \[
\triple{\tuple{\mathbb{Z} 
    }{a}{\mathit{Vars}}}{y \coloneq v \cdot \sin{(\theta)}\cdot dt + y}{\tuple{\mathbb{Z}}{a }{\mathit{Vars}}}{\mathbf{eq}}.
\]
\paragraph{Line 4}
Next, we focus on the statement $c_4 \triangleq v \coloneq a \cdot dt + v$. This can be rewritten as $v \coloneq f(a,dt,v)$.  This assignment statement is injective for $v$, i.e., $\forall v_1, v_2, a,dt. f(a,dt,v_1) = f(a,dt,v_2)$ implies $v_1 = v_2$. We can therefore define a function $\hat{f} : \mathbb{Z} \times \mathbb{Z} \times \mathbb{Z} \rightarrow \mathbb{Z}$ such that $\hat{f}(a,dt,f(a,dt,v)) = v$. Using this, we can compute the group action $\mathsf{POST}$:
\begin{align*}
  \mathsf{POST}(\tuple{\mathbb{Z}}{a}{\mathit{Vars}}, c_4) \triangleq \actiontemp{g}{-} {\mathit{Vars}} {g \cdot &(v \rightarrow \alpha_f, x \rightarrow \alpha_x, \dots \theta \rightarrow \alpha_{\theta}) =\\&
                       v \rightarrow f(\alpha_a, dt, \hat{f}(\alpha_a, dt, \alpha_f)), \\ &  x \rightarrow \alpha_x + 1, 
                        y \rightarrow \alpha_y +1, \dots, \theta \rightarrow \alpha_{\theta}}
\end{align*}
This simplifies to 
\begin{align*}
\actiontemp{g}{-} {\mathit{Vars}} {g \cdot &(v \rightarrow \alpha_f, x \rightarrow \alpha_x, \dots \theta \rightarrow \alpha_{\theta}) =\\&
                       v \rightarrow \alpha_f , x \rightarrow \alpha_x + 1, 
                         y \rightarrow \alpha_y +1, \dots, \theta \rightarrow \alpha_{\theta}} .
\end{align*}
For every integer $c$, $\mathsf{POST}(\tuple{\mathbb{Z}}{a }{\mathit{Vars}}, c_4)$ acts  like the identity on all variables except $x$ and $y$, which are translated by $c$ to $x+c$, and $y+c$. This is the same action as $\tuple{\mathbb{Z}}{a}{\mathit{Vars}}$. 
Therefore, $\mathsf{POST}(\tuple{\mathbb{Z}}{a }{\mathit{Vars}}, c_4) \overset{\mathbf{eq}}{\rightarrow}
\tuple{\mathbb{Z}}{a }{\mathit{Vars}}$. Now by the $\textsc{ASSGN}$ rule,
$\triple{\tuple{\mathbb{Z}}{a }{\mathit{Vars}}}{c_4}{\mathsf{POST}(\tuple{\mathbb{Z}}{a}{\mathit{Vars}}, c_4)}{\mathbf{eq}}$ is valid, and by $\textsc{CONS-1}$, $
    \triple{\tuple{\mathbb{Z}}{a}{\mathit{Vars}}}{c_4}{\mathsf{\tuple{\mathbb{Z}}{a}{\mathit{Vars}}}}{\mathbf{eq}}$
is also valid. 
The proof for the rest of the loop body is similar since the statements are
injective; we present details in
\iffull \Cref{app: Examples}. \else the full version of the paper. \fi

If we let $C$ denote lines 2-6 in Figure \ref{fig: car-trans}, 
by the $\textsc{SEQ}$ rule we can prove
$\triple{\tuple{\mathbb{Z} }{a }{\mathit{Vars}}}{C}{\tuple{\mathbb{Z}}{a}{\mathit{Vars}}}{\mathbf{eq}}$. 
Now since $\mathit{Vars}(C) \subset \mathit{Vars}$, by the $\textsc{FOR}$ rule, we have the
desired triple:
\[ 
  \triple{\tuple{\mathbb{Z} 
    }{a}{\mathit{Vars}}}{\textbf{for } {(t \coloneq 0; t < T; t \coloneq t +dt)} \textbf{ do } C}{\tuple{\mathbb{Z}}{a}{\mathit{Vars}}}{\mathbf{eq}}
\]

\subsection{Car in a Straight Line}
\label{ex:dihedral}
We return to a car system again. This time, we imagine a car moving in a
straight line without the ability to turn. In this program, presented in
\Cref{fig: car-straight}, the angle from the x-axis $\theta$ remains constant;
this variable has type $\mathbb{Z}_{360}$ ranging over $0$ to $360$.
We consider a different symmetry property this time.
Imagine a car that starts at initial coordinates $(x,y)$, and makes an
angle $\theta$ with the x-axis. Intuitively, if we rotate the angle $\theta$ by
$90^{\circ}$, and the car moves in a straight line, at the end of its path, the
car's position will be transformed by a $90^{\circ}$ rotation. Similarly, if we
reflect the car along an axis, and then the car moves in a straight line, the
original car's final position and the reflected car's final position will be
related by the same reflection.

To prove these symmetry properties formally, we need to find an input
transformation that captures the transformations we described.
The \emph{Dihedral Group} $D_4$ captures the rotation by $90^\circ$, and reflection. Recall the presentation of this group: $\langle r,s \mid r^4 = e, s^2 = e, rs = sr^{-1} \rangle$. The generator $r$ corresponds to rotation, and $s$ to reflection. Therefore, we would like the action of $r$ to be a rotation, and the action of $s$ to be a reflection.  Thus, we define a group action:
\begin{align*}
  \tuple{D_4}{a}{\{x,y,\theta\}}
  \triangleq (&\langle r, s \mid r^4 = e, s^2 = e, rs = sr^{-1} \rangle, 
                           \langle {\{x,y,\theta\}},\\& r\cdot(x \rightarrow \alpha_x, y \rightarrow \alpha_y, \theta \rightarrow  \alpha_\theta)
                           = (x \rightarrow -\alpha_y, y \rightarrow \alpha_x,
                           \theta \rightarrow  (\alpha_\theta + 90)\ \text{mod }{360}) \\
                           & s\cdot(x \rightarrow \alpha_x, y \rightarrow \alpha_y, \theta \rightarrow \alpha_\theta)
                           = x \rightarrow \alpha_x, y \rightarrow -\alpha_y, \theta \rightarrow -\alpha_\theta)\rangle).
\end{align*}
Here, $r$ acts by adding $90^{\circ}$ to the angle $\theta$, and also changing the $x$ and $y$ co-ordinates accordingly. $s$ acts by reflecting the starting point across the $x$-axis.
It is routine to verify that this is a group action. Once again, we want the transformation on the input to be preserved, so we would like the homomorphism to be the identity map: $\mathbf{eq}$ ($\varphi(g) = g$).
Let $C$ denote the loop body from Lines 4-8 in Figure \ref{fig: car-straight}. We seek to prove the triple 
\[
\triple{\tuple{D_4}{a}{\mathit{Vars}}}{t \coloneq 0; b \coloneq t < T; \textbf{while } {b } \textbf{ do }C}{\tuple{D_4}{a}{\mathit{Vars}}}{\mathbf{eq}}.
\]
\iflong We present the proof that the loop body is preserved in the appendix. For now, we focus mainly on the application of the \textsc{WHILE} rule. 
\fi

The loop body, denoted by $C$, preserves $\tuple{D_4}{a}{\{x,y,\theta\}}$. We can prove this by applying $\textsc{ASSGN}$ to each statement, and then the 
sequencing rule and lift rule, which allows us to conclude \[\triple{\tuple{D_4}{\overline{a}}{\mathit{Vars}}}{C}{\tuple{D_4}{\overline{a}}{\mathit{Vars}}}{\mathbf{eq}}\]
To finish the proof, we need to apply the $\textsc{WHILE}$ rule.
We observe the loop body does not modify any variables outside of $\mathit{Vars}$.  Additionally, since the group action $\tuple{D_4}{\overline{a}}{\mathit{Vars}}$ does not change $t$ and $T$, we can prove $\tuple{D_4}{\overline{a}}{\mathit{Vars}} \overset{\mathbf{e*}}{\rightarrow} (E,e_{\{b\}})$. Therefore, we can apply the $\textsc{WHILE}$ rule to conclude:
\[\triple{\tuple{D_4}{a}{\mathit{Vars}}}{\textbf{while } {b} \textbf{ do } C}{\tuple{D_4}{a}{\mathit{Vars}}}{\mathbf{eq}}. \]
\subsection{Voting}
\begin{figure}
\begin{minipage}[c]{0.4\linewidth}
    \begin{tabular}{c}
    \begin{lstlisting}
1. $t \coloneq 0$
2. $b \coloneq t < T$
3. $\mathbf{while}(b)${
4. $\quad \quad x \coloneq v \cdot \cos{(\theta)}\cdot dt + x;$
5. $\quad \quad y \coloneq v \cdot \sin{(\theta)}\cdot dt + y;$
6. $\quad \quad v \coloneq a \cdot dt + v;$
7. $\quad \quad t \coloneq t + dt$
8. $\quad \quad b \coloneq  t < T$
9. }
     \end{lstlisting}
    \end{tabular}
\caption{A car moving in a straight line}
\label{fig: car-straight}
\end{minipage}%
\begin{minipage}[c]{0.5\linewidth}
  \centering
     \begin{tabular}{c}
    \begin{lstlisting}
1. $count_1 \coloneq v_1 == 0\; ? \; count_1 + 1 : count_1$; 
2. $count_1 \coloneq v_2 == 0 \;? \; count_1 + 1 : count_1$;
3. $count_2 \coloneq v_1 == 1\; ? \; count_2 + 1: count_2 $; 
4. $count_2 \coloneq v_2 == 1\; ? \; count_2 + 1: count_2 $; 
5. $b \coloneq count_1 > count_2$;
6. $\mathbf{if}~b$
7. $\quad winner \coloneq 0$;
8. $\mathbf{else}$
9. $\quad winner \coloneq 1$;
     \end{lstlisting}
    \end{tabular}\caption{A voting system with two voters}
    \label{fig:voting}
\end{minipage}
\end{figure}
We now look at an application from a different domain. We verify a symmetry property of a voting system as considered by \citet{beckert2016automated}. They prove anonymity (among other properties) using a bounded model checker for a bounded number of agents.
Here, we consider a simplified voting system with two candidates and two voters
modeled as a program in \Cref{fig:voting} (later in \Cref{sec: impl}, we will
consider larger versions of this program for more voters).
The voting system counts the votes of both the voters and the candidate with more votes is declared the winner. The variable $v_1$ represents Voter-1's vote, and $v_2$ represents Voter-2's vote.
The constant $0$ symbolizes a vote for Candidate-1, and $1$ symbolizes a vote for Candidate-2.
The votes are counted in lines 1-4. The variable $count_1$ is the number of
votes for candidate 1, and $count_2$ is the number of votes for candidate 2.
Line 5 compares the two counts, lines 6-9 set the winner to be the candidate
with the higher vote count.

Ideally, we would like the voting system to satisfy \emph{anonymity} i.e., the identity of the voter should not affect the outcome. In other words, if we permute the votes, the outcome of the election should be unchanged.
We model this situation by permuting votes by a group action of $S_2$:
\[
  \actiontemp{x}{x^2 = e}{\{v_1, v_2\}}{x \cdot (v_1 \rightarrow \alpha_1, v_2 \rightarrow \alpha_2) = (v_1 \rightarrow \alpha_2, v_2 \rightarrow \alpha_1)}
\]
The group element $x$ acts by swapping the votes. Letting $C$ denote the entire program, we would like to prove the triple
\[
\triple{\tuple{S_2}{a}{\{v_1,v_2\}}}{C}{\tuple{E}{e}{\{winner\}}}{\mathbf{e*}}
\]
which says that permuting $v_1$ and $v_2$ preserves the winner.

Intuitively, the number of votes for Candidate-1, and Candidate-2 should remain
unchanged after the permutation. Therefore, we will split the program into two parts. We will first prove that after Lines 1 and 2, the value of $count_1$ is invariant, and then show that after Lines 3 and 4, $count_2$ is invariant. We show the details of this proof in
\iffull \Cref{app: Examples}. \else the full version of the paper. \fi
For now, we focus on highlighting the application of $\textsc{IF}$ rule. If we denote the first four lines as $C_r$, using the regular assignment rule and sequencing, we prove 
\[ \triple{(S_2,a_{\mathit{Vars}})}{C_r}{(S_2,a_{\mathit{Vars}})}{\mathbf{eq}}.\]
Continuing with the rest of the program,
on Line 5, we assign the test to variable $b$. 
We can conclude
\( \triple{(S_2,a_{\mathit{Vars}})}{{b} \coloneq count_1 > count_2;}{(S_2,a_{\mathit{Vars}})}{\mathbf{eq}}\). Next, we move on to the $\mathbf{if}$ statement. 
Using $\textsc{CONST}$, $\textsc{LIFT}$ and $\textsc{ID}$  we can conclude
\(\triple{(S_2,a_{\mathit{Vars}})}{{winner} \coloneq 0}{(S_2,a_{\mathit{Vars}})}{\mathbf{eq}}\)
and similarly, 
\(\triple{(S_2,a_{\mathit{Vars}})}{{winner} \coloneq 1}{(S_2,a_{\mathit{Vars}})}{\mathbf{eq}}\).
Now, since $(S_2,a_{\mathit{Vars}})$ leaves $count_1$ and $count_2$ unchanged, we can prove the entailment 
\( (S_2,a_{\mathit{Vars}}) \overset{\mathbf{e*}}{\rightarrow} (E,e_{\{b\}})\).
Then by the $\textsc{IF}$ rule, 
\[
\triple{(S_2,a_{\mathit{Vars}})}{\textbf{if } b \textbf{ then } {winner} \coloneq 0; \textbf{else }  {winner} \coloneq 1; }{(S_2,a_{\mathit{Vars}})}{\mathbf{eq}}
\]
Finally, we can apply $\textsc{SEQ}$ to conclude
$\triple{(S_2,a_{\mathit{Vars}})}{C}{(S_2,a_{\mathit{Vars}})}{\mathbf{eq}}$.
Also, $(S_2,a_{\mathit{Vars}}) \overset{\mathbf{e*}}{\rightarrow}
(E,e_{\{winner\}})$. Then by $\textsc{CONS-1}$, we conclude the desired triple:
\( \triple{(S_2,a_{\mathit{Vars}})}{C}{(E,e_{\{winner\}})}{\mathbf{e*}}.
\)

\subsection{AAC Flow}
Next, we consider a system describing a three-dimensional incompressible
velocity field called the AAC flow, an example of a fluid flow that can have
chaotic trajectories. The program in \Cref{fig: AAC} simulates this system for $T$ time steps.
The symmetry properties exhibited by this system are described by
\citet{McLachlan_Quispel_2002}. We describe a variant called the ABC flow in the full paper.
\subsubsection*{Symmetry Property}
This system is equivariant to the action of $S_2$ when $S_2$ acts by flipping the sign of $x$, subtracting $y$ from $\pi$ and subtracting $\pi$ from $z$. Formally define the action of $S_2$ as: 
\[ \actiontemp{g}{g^2 = e}{\{x, y, z\}}{g \cdot (x \rightarrow \alpha_1, y \rightarrow \alpha_2, z \rightarrow \alpha_3) = (x \rightarrow -\alpha_1, y \rightarrow \pi -\alpha_2, z \rightarrow z - \pi)}\]
If we let $C$ denote the program in \Cref{fig: AAC}, then  our desired triple is: 
\(
\triple{(S_2,\overline{a})}{C}{(S_2,\overline{a})}{\mathbf{eq}}
\). 
Lines 2, 3, and 4 are injective, so we can compute $\mathsf{POST}$. Using $\textsc{ASSGN}, \textsc{SEQ}$, and the $\textsc{FOR}$ rule we can prove the desired triple.

\subsubsection*{Incorrect Version of AAC}
\begin{figure}
  \begin{center}
    \begin{tabular}{c}
    \begin{lstlisting}
1. $\while(t \coloneq 0; t < T; t \coloneq t + dt)${
2. $\quad \quad x \coloneq (A \cdot \sin{(z)} + C \cdot \cos{(y)})\cdot dt + x;$
3. $\quad \quad y \coloneq  (A \cdot \sin{(x)} + A \cdot \cos{(z)})\cdot dt + y;$
4. $\quad \quad z \coloneq  (C \cdot \sin{(y)} + A \cdot \cos{(x)})\cdot dt + z;$
5. }
     \end{lstlisting}
    \end{tabular}
  \end{center}
    \vspace{-2.5mm}
  \caption{A program simulating the AAC flow}
  \label{fig: AAC}
\end{figure}
While verifying this property, our tool (described in \Cref{sec: impl}) found an
error in the AAC flow described in \citet[Example 28]{McLachlan_Quispel_2002}.
In particular, on Line 4, they use \( z \coloneq  (C \cdot \sin{(y)} + A \cdot
\cos{(\mathbf{z})})\cdot dt + z\), instead of \(z \coloneq  (C \cdot \sin{(y)} +
A \cdot \cos{(\textbf{x})})\cdot dt + z \).  The claimed symmetry property fails
to hold for their system.  While we believe this error was likely a
misprint---\citet{McLachlan_Quispel_2002} analyze the correct version in the
rest of their work---this example highlights the subtlety of reasoning about
symmetry properties and the utility of our method for analyzing symmetry
properties.

\section{Implementation and Evaluation}
\label{sec: impl}
We have developed a prototype tool called $\mathsf{SymVerif}$ for proving
symmetry properties of programs. $\mathsf{SymVerif}$ consists of a verifier and
a synthesizer. Our implementation assumes that any provided group action is
faithful.  
\subsection{Verifying the Assignments}
For non-injective assignments, applying the $\textsc{SEM-ASSGN}$ rule requires checking semantic validity of a triple. 
To ease this task, in
certain cases, we can encode this semantic condition as an equivalent logical
formula. We will use these encodings to verify the application of the
$\textsc{SEM-ASSGN}$ rule for all the examples in Section \ref{sec:examples}
using Z3~\cite{z3}.  

Concretely, we aim to check if a given triple \(\triple{(G,a)}{v \coloneq
exp}{(H,b)}{\varphi}\) is valid. We split this task into two parts. First, we
check if for every group element $g \in G$---of which there might be infinitely
many---there exists a corresponding group action of $h \in H$ on the output.
When $H$ is finite, we can encode this check as a first order formula:
\begin{restatable}{theorem}{FormulaFake}
  \label{def: formula-fake}
  Let $G = \langle g_1, \dots, g_n \mid r_1, \dots, r_j \rangle$ be a finitely
  generated group with $n$ generators, and action $(G,a_{\mathit{Vars}})$. Let $H$ be a
  finite group with $k$ elements and faithful action $(H,b_{\mathit{Vars}})$
    and
  $v \coloneq exp$ be an assignment command.
  Then there exists a first-order formula $F \llparenthesis \triple{(G,a_{\mathit{Vars}})}{v \coloneq exp}{(H,b_{\mathit{Vars}})}{} \rrparenthesis$ which is satisfiable if, and only if, there exists a homomorphism $\tau: G \rightarrow H$
  \[
    \forall g\in G,\;
    \forall \sigma \in \Sigma,\; b_{\tau(g)(\sigma)}(\sem{v \coloneq exp}_{\sigma})
    = \sem{v\coloneq exp}_{a_{g}(\sigma)}
  \]
\end{restatable}

We construct this formula, prove it correct, and extend it to check entailment in
\iffull \Cref{app:encoding}. \else the full version of the paper. \fi
This the formula checks whether for all $g \in G$, there exists $\tau(g) \in H$ satisfying
the target symmetry property, but we will still need to check that this mapping
is equivalent to the given homomorphism $\varphi$. To accomplish this task,
recall that our tool assumes that the provided group actions are faithful. Thus
by \Cref{thm: faithful-unique}, there is at most one homomorphism validating the
triple. Using the validity of the formula from \cref{def: formula-fake}, we can
construct the homomorphism $\tau$.
Finally, we check if $\varphi$ is equivalent to $\tau$ to validate the triple.
In general, this can be undecidable, but in practice, our tool can handle this
efficiently.  If $H$ is an infinite group, our formula encoding in \Cref{def:
formula-fake} must make use of the provided homomorphism $\varphi$. We show
further details of this case in
\iffull \Cref{app: Examples}. \else the full version of the paper. \fi

\subsection{Prototype Implementation}
Our prototype implementation consists of two components: a verifier and a synthesizer. 
\paragraph*{Verifier} 
Given a program, a pre-group action, and a post-group action, our verifier
checks whether the triple is valid. It handles injective assignments using
symbolic computation and non-injective assignments through additional
annotations and SMT solving. Our verifier does not handle loops and general
if-statements, but it is able to handle limited branching via ternary
expressions. 
\paragraph*{Required Annotations}
In addition to the desired pre- and post-condition for the program, for non-injective assignments, a pre-and post-condition annotation is required for applying the $\textsc{SEM-ASSGN}$ rule. If the post-condition group is infinite, the annotation must provide a homomorphism.

Our verifier is implemented in Python, leveraging the SymPy library to compute
$\mathsf{POST}$ for injective assignments. To verify annotations for
non-injective assignments we encode the problem using the formula from
\Cref{def: formula-fake}, which is solved using the Z3 SMT solver~\cite{z3}.
Entailment is checked similarly.

\paragraph*{Synthesizer} Manually finding valid pre- and post-conditions for non-injective assignments can be tedious and error-prone. To improve automation, we also implement a synthesizer that generates a valid pre-condition for a given post-condition and assignment statement.

Our pre-condition synthesizer is based on our construction of the weakest
pre-condition, particularly the group action in \Cref{def:grp-finite-action}.
Recall that to define the group action $\hat{G}((H,b),C)_{\mathit{Vars}}$, we
have to find a bijection $g_i$ for each generator of $h _j$ of $H$ such that $b_{h_j}(\sem{C}_{\sigma}) = \sem{C}_{a_{g_i}(\sigma)}$. We encode the search for such bijections as a syntax-guided synthesis (SyGuS) problem. Specifically, we define a grammar that is heuristically generated based on the post-condition, and assignment statement. The grammar represents expressions that define group actions in the pre-condition. 
The grammar is generated as expressions over all variables and function symbols in the assignment statement and the group action, as well as basic arithmetic operations. For instance, if the post-condition and assignment were a simple addition to the 
variable $x$, the generated sketch might look like $(x \mid {??}) (+ \mid -) ({??}
\mid x )$, where the hole $??$ could be any constant and $|$ represents choice.
So this grammar is generating either $x$ or a constant, plus-or-minus $x$ or a constant. 

We use SKETCH \cite{solar2013program} to synthesize bijections over this grammar
satisfying the constraints.  Since SKETCH only verifies correctness for bounded
values, we verify the soundness of the synthesized bijection for all values of the given type using the encoding in \Cref{def: formula-fake}.
\subsection{Evaluation}
Our empirical evaluation focuses on three key aspects: performance, scalability, and the ability to synthesize pre-conditions for applying the $\textsc{SEM-ASSGN}$ rule.  
\subsubsection{Performance of the Verifier}
\begin{table}[ht]
  \centering
  \begin{tabular}{lcc}
    \toprule
    \textbf{Program} & \textbf{Time (s)} & \textbf{No. of Assignment Statements}  \\ 
    \midrule
    Car Translation                  & 0.096  & 5               \\
    Car Translation 2                & 0.217      &5          \\ 
    Car with Dihedral Group D4       & 0.266         & 3     \\  
    Car with Dihedral Group D6       & 0.493        & 3    \\ 
    Lorenz Attractor System          & 0.115          & 3        \\  
    Gravitational Attractor          & 0.920          & 5          \\   
    AAC Flow                         & 0.496  & 3 \\
    Voting System 2 Voters           & 0.901           & 5           \\   
    Voting System 20 Voters          & 3.274    & 42 \\   
    \bottomrule
  \end{tabular}
  \caption{Verified Programs with Time Taken, and Size of the Program}
  \label{fig: verifier-benchmarks}
  \vspace{-2mm}
\end{table}
\Cref{fig: verifier-benchmarks} presents the performance results of our verifier on a diverse set of example programs. These programs, chosen from various domains, range from simple geometric transformations to more complex dynamical systems and voting protocols, including a system with 20 voters. 
Only the Gravitational attractor program requires annotations with pre- and post-conditions for non-injective assignments. All other programs only required the desired pre- and post-conditions for the entire program. The programs and the assertions are defined formally in \Cref{sec:examples} and the Appendix.
\Cref{fig: verifier-benchmarks} shows that the time taken to verify most programs is mostly under one second, except the voting system with 20 voters that has approximately $2 \times 10^{18}$ group elements in the pre-condition.
\subsubsection{Ability to Synthesize Pre-Conditions}
\begin{table}[ht]
  \centering
  \begin{tabular}{lccc}
    \toprule
    \textbf{Assignment Statement}          & \textbf{Sketching Time (s)} & \textbf{Verification Time (s)} & \textbf{Is Injective?} \\
    \midrule
    car translation assgn $x$              & 4.991                   & 0.45                       & Yes                    \\
    car translation assgn $y$              & 4.992                   & 0.34                       & Yes                    \\
    car translation assgn 3     & 0.872                   & 0.30                       & Yes                    \\
    car translation assgn 4     & 0.706                   & 0.31                       & Yes                    \\
    car translation  assgn 5     & 0.733                   & 0.23                       & Yes             \\
    car in straight line assgn $x$ with D4 & 1.896                      & 1.02                         & Yes                    \\
    car in straight line assgn $y$ with D4 & 2.061                      & 1.17                         & Yes                    \\
    lorenz assgn $x$                       & 1.702                  & 1.221                      & Yes                    \\
    lorenz assgn $y$                       & -                   & -                      & Yes                    \\
    lorenz assgn $z$                       & 0.589                   & 0.42                       & Yes                    \\
    gravity assgn $F$                      & 2.513                   & 0.83                       & No                     \\
    gravity assgn $v_1$                    & 0.836               & 0.145                      & Yes                    \\
    gravity assgn $v_2$                    & -                 & -                 & Yes                    \\
    gravity assgn $x_1$                    & 0.665               & 0.21                      & Yes                    \\
    gravity assgn $x_2$                    & 0.711                  & 0.18                      & Yes                    \\
    
    AAC assgn $z$  & 9.668 & 0.572 & Yes \\
  
    \bottomrule
  \end{tabular}
  \caption{Sketching and Verification Times}
  \vspace{-7mm}
  \label{fig: synth}
\end{table}
To assess the synthesizer, we measured both the performance and the utility of the synthesized pre-conditions. Since the synthesizer is not guaranteed to produce the weakest pre-condition, we evaluated whether the pre-conditions generated were strong enough to prove the desired properties in the example programs. We focus on individual assignment statements from the programs in Table \ref{fig: verifier-benchmarks}. For each assignment statement, we provided the synthesizer with post-conditions from the verifier and set a timeout limit of 2 minutes.  

The results are shown in \Cref{fig: synth}, with sketching and verification times reported. Whenever the synthesizer successfully generated pre-conditions, they were adequate to prove the desired properties of the program. When the assignment is injective, the ground truth weakest pre-condition is known. In all such cases, when the tool did not time out, it was able to find the \emph{weakest pre-condition}. For assignments involving more complex operations-such as piecewise and trigonometric functions—the tool timed out. Overall, the results demonstrate the synthesizer's ability to generate useful pre-conditions, especially when expressions do not involve complex functions. 

\paragraph*{Verification Times} We use the encoding in \Cref{def: formula-fake}
to verify the sketches returned by $\mathsf{SKETCH}$.  The verification time for
assignments is relatively fast---usually taking less than a second---thus demonstrating the efficiency of the encoding \Cref{def: formula-fake}. 
This suggests that our encoding is efficient and can verify a variety of assignment statements.
\subsubsection{Verifier Scalability}
\begin{table}[ht]
  \centering
  \begin{tabular}{lccc}
    \toprule
    \textbf{Program} & \textbf{Time (s)} & \textbf{Size of Pre-Group}  & \textbf{Size of Post-Group}  \\ 
    \midrule
    Car with Dihedral Group D4     & 0.166  & 8 & 8 \\
    Car with Dihedral Group D8     & 0.294 &  16 & 16 \\
    Car with Dihedral Group D32    & 0.307 & 64 & 64 \\ 
    Car with Dihedral Group D512   & 0.419  & 1024 & 1024 \\
    Car with Dihedral Group D1024  &  0.524 & 2048 & 2048 \\
    \bottomrule
  \end{tabular}
  \caption{Verified Programs and Times}
  \label{fig: scale-benchmarks}
\end{table}
We also perform experiments to evaluate the scalability of our verifier on
examples varying the sizes of the groups, keeping the program constant. Since
the verifier computes $\mathsf{POST}$ only for the generators of the
pre-condition, we expect the runtime to scale with the number of generators,
which can be substantially smaller than the total number of group elements. On
the other hand, the complexity of the entailment check depends on both the
number of elements in the post-condition group and the number of generators in
the pre-condition group.

Our experiments show that the verifier scales effectively:  we verify a triple with 2048 elements in both the pre-condition and post-condition groups, demonstrating the ability to handle large groups. These results highlight the ability of our tool to handle larger group sizes efficiently.

\section{Related Work}
\label{sec:rw}
\paragraph*{Symmetry Properties}
\citet{gotlieb2003exploiting} consider a symmetry property very similar to
Definition \ref{def: valid}. Rather than trying to prove the property formally,
their approach uses software testing with a finite set of inputs. In contrast,
our approach formally proves symmetry properties.

\paragraph*{Type Systems}
\citet{atkey2014parametricity} develop a type system to reason about conserved physical quantities via relational parametricity. While this allows reasoning about semantic symmetries, the scope of this work is limited to invariance of physical quantities. Our program logic can handle arbitrary assertions about symmetries, and our imperative language is general purpose.

\paragraph*{Nominal Sets and Nominal Logic}
Nominal sets \cite{nominalsets} provide an elegant mathematical framework for
reasoning about names and bindings, via permutations. Nominal Logic
\cite{pitts2003nominallogic} is a variant of first-order logic that uses nominal
sets to build a proof for reasoning about the syntax of formal languages.  Our
work considers arbitrary group actions acting over the infinite space of program
states rather than variable names, and we focus on program verification.  It
would be interesting to see if there is some more formal connection between our
system and nominal logic.

\paragraph{Hyperproperties and Relational Hoare Logics} The symmetry property we consider can be seen as families of hyperproperties \cite{clarkson2010hyperproperties}, where each group element induces a bijection between program states. Previous work has addressed proving hyperproperties using program logics \cite{benton2004simple, sousa2016cartesian}.

We compare our approach with Relational Hoare Logic (RHL) \citep{benton2004simple}. While RHL can encode each group element as a relational triple and prove them individually, our approach enables proving these properties in a single proof. We will illustrate this point with case study \ref{ex:dihedral}.
We consider the dihedral group $(D_4)$, which has a finite number of elements. Let
$(D_4, a)$ be the group action of the dihedral group, as defined previously.
We aim to prove
the triple $\triple{(D_4, a)}{C}{(D_4, a)}{\mathbf{eq}}$ from Example 2. 
In Relational Hoare Logic, we could encode this triple as: \(\forall g \in D_4, \exists h \in D_4\)
\[
(x_2 =  a_g(x_1), y_2 = a_g( y_1), \theta_2 = a_g(\theta_1))\{C\} 
(x_2 =  a_{\mathbf{eq}(h)}(x_1), y_2 = a_{\mathbf{eq}(h)}( y_1), \theta_2 = a_{\mathbf{eq}(h)}(\theta_1))
\]
This says for every group element, $g$ if we transform the program state by
$a_g$, then the program state without the group acting on it (denoted by the
subscript 1), after executing $C$ is related to the program state with the group
acting on it (denoted by the subscript 2) after executing $C$, by the action of
$a_{\mathbf{eq}(h)}$.
We see three main challenges in verifying symmetry
properties in RHL:
\begin{enumerate}
\item The quantifiers on $g$ are outside the triple and will need
meta-theoretic reasoning to prove. Further, the actions of each group element
can be very different, so there may not be common proof patterns that can be
reused when proving the triple in RHL.
\item Since RHL rules do not preserve group actions, we must verify that each assertion corresponds to a group action. For example, the assignment rule might produce conditions that do not align with a group action, requiring separate checks.
\item RHL rules do not construct a homomorphism, which is required for our target property. Encoding this constraint would require additional meta-theoretic reasoning over all functions from $G$ to $H$. 
\end{enumerate}
Our logic overcomes these challenges by leveraging the group-theoretic structure of our assertions.
\paragraph*{Automated Reasoning}
Previous work has looked at automating Relational Hoare Logic proofs. For example,
\textsc{DESCARTES} \cite{sousa2016cartesian} can be used to verify $k$-safety hyper-properties automatically. However, it cannot be used directly to verify the type of symmetry property we consider. 
There is also a long line of work in symmetry
reduction and its applications to automated reasoning in model checking \cite{clarke1998symmetry, wahl2010replication} and SMT solving \cite{deharbe2011exploiting, areces2013symt}.
In these settings,  the symmetries are often limited to finite structures like
transition systems or control flow graphs.
It is not clear how to apply these approaches to prove our target
properties, which involve richer program states and more expressive programs.
\section{Conclusion}
\label{sec:conc}

We have developed a Hoare logic for reasoning about symmetry properties of
imperative programs. Our logic uses a novel syntax of group actions to
specify the pre- and post-conditions and features a variety of proof rules for
manipulating programs and constructions on group actions. We view this as the
first work in a potentially longer line of research on program verification
methods for symmetry properties. Accordingly, we see a variety of interesting
future directions:

\paragraph*{Theoretical questions.}
Our investigation of weakest pre- and strongest post-conditions shows how to
construct these group actions for single assignment commands, but many questions
remain open. In particular, we do not know if the weakest pre-condition operation
can be extended to a transformer for sequential composition. Addressing these
questions could help shed light on the completeness (or lack of completeness) of
our approach.

\paragraph*{Supporting richer languages.}
Finally, we have considered a simple imperative language for this investigation.
There are many possible directions to extending the expressivity. For instance,
it might be interesting to consider symmetries for effectful (e.g., randomized)
programs. In another direction, it could also be interesting to develop a
higher-order, functional language for symmetries.

\section*{Acknowledgements}
We thank the reviewers for their constructive suggestions. This work benefited
from discussions with Ohad Kammar, and feedback from the Cornell PLDG Seminar
and the UToPiA Seminar at UT Austin.

\section*{Data Availability Statement}
Our tool $\mathsf{SymVerif}$ can be found on Zenodo~\cite{mehta202516921665}.

\bibliographystyle{ACM-Reference-Format}
\bibliography{header,refs}
\iflong
\appendix
\section{Appendix for Section~\ref{sec:assertions}}
In this section of the appendix, we collect the omitted materials from Section~\ref{sec:assertions}. First, we state some important lemmas that we will use in this section. Next, we provide details on the group action construction and missing proofs.
Group actions on a set $X$ can be viewed as homomorphisms into the symmetric
group.
\begin{restatable}{lemma}{ActionHomomorph}
  \label{lem: action-homo}
  Every group action $\tuple{G}{a}{V}$ corresponds to a homomorphism from $G$ into $Sym(\Sigma_V)$.
\end{restatable}

To define a map from a group given as a group presentation, it suffices to
define the map from each generator and check that the relations are preserved.
To prove this, we will use the following lemma:
\begin{restatable}{lemma}{HomoIdentity}
      \label{lemma-homo-respects-identity}
  \cite{johnson1997presentations}
  Given a group with a presentation
  $G = \langle S_G \mid R_G \rangle$, where each relation in $R_G$ is of the form $r_j = e_G$.
  Let $H$ be another group and $f: S_G \rightarrow H$ a function.
  Then $f$ extends to a homomorphism $F: G \rightarrow H$.
  if and only if each $f(r_j) = e_H$.
\end{restatable}
\FiniteGrpThm*
 \begin{proof}
    The statement follows directly from Lemma \ref{lem: action-homo} and Lemma \ref{lemma-homo-respects-identity}.
  \end{proof}
\subsection{Syntax of Group Actions}
Recall the definition of the direct product group action in our syntax: 
\DirProdAction*
The condition for when $(G \times H, a \times b_{V_1 \cup V_2})$ is a valid action is characterized by the following lemma:
\LemDirProdCond*
 \begin{proof}
    ($\implies$) We know $(G \times H, a \times b_{V_1 \cup V_2})$ is a group action. Therefore, $(a \times b)_{(g_1,h_1)}i_{g_2,h_2)}(\sigma) = (a \times b)_{(g_1g_2, h_1h_2)}(\sigma)$. By the definition of $(a \times b)$, this means $\overline{a}_{g_1}(\overline{b}_{h_1}(\overline{a}_{g_2}(\overline{b}_{h_2}(\sigma)))) = \overline{a}_{g_1}(\overline{a}_{g_2}(\overline{b}_{h_1}(\overline{b}_{h_2}(\sigma))))$, which implies the result.
    \\ ($\impliedby$) We know that $\overline{a}_g(\overline{b}_h(\sigma)) = \overline{b}_h(\overline{a}_g(\sigma))$.  Now, $(a \times b)_{(g_1,h_1)}i_{g_2,h_2)}(\sigma) = \overline{a}_{g_1}(\overline{b}_{h_1}(\overline{a}_{g_2}(\overline{b}_{h_2}(\sigma))))$. By commutativity, this is equal to $\overline{a}_{g_1}(\overline{a}_{g_2}(\overline{b}_{h_1}(\overline{b}_{h_2}(\sigma))))$, which means $(a \times b)_{(g_1,h_1)}i_{g_2,h_2)}(\sigma) = (a \times b)_{(g_1g_2, h_1h_2)}(\sigma)$.
  \end{proof}
  \subsection{Free Product Construction}
\label{app: fp-cons}
The free product is an operation that combines two groups $G$ and $H$ into a
larger group denoted by $G * H$, where each element is a word of the form
$g_1h_1g_2 \dots h_n$, where the first and last elements can be from either $G$
or $H$. Every word is \emph{reduced}, which means every instance of the identity
element of $G$ or $H$ is removed, and every consecutive pair $g_1g_2$ or
$h_1h_2$ is replaced by its product.
The group operation of the group $G*H$ is concatenation followed by reduction.
Another useful characterization of the free product is via a so-called universal property:
If $f_1$ and $f_2$ are any functions into a group $I$, $j_1$ and $j_2$ are
injective homomorphisms, then if $G*H$ is a free product, there exists a unique homomorphism $\varphi$ such that the following diagram commutes.
\[\begin{tikzcd}
    && {G*H} \\
    \\
    G && I && H
    \arrow["\varphi"{description}, from=1-3, to=3-3]
    \arrow["{j_2}", from=3-1, to=1-3]
    \arrow["{f_1}", from=3-1, to=3-3]
    \arrow["{j_1}"', from=3-5, to=1-3]
    \arrow["{f_2}"', from=3-5, to=3-3]
  \end{tikzcd}\]
Suppose we have two groups $G$ and $H$ presented as $
  G = \langle S_G \mid R_G \rangle$ and $H =\langle S_H \mid R_H\rangle$. We assume that the symbols in $S_G$ and $S_H$ are disjoint. The set of generators of the free product $G * H$ is the union of $S_G$ and $S_H$, and the relations consist of $R_G$ and $ R_H$.  In our syntax, we can express the free product as:
\[
  G * H  = \langle S_G, S_H \mid R_G, R_H \rangle
\]
\subsection{Free Product Group Action}
\label{app:fp-action}
Consider a free product group $I = G*H$. Suppose $\tuple{G}{a}{V_1}$ is an action of $G$ on $\Sigma_{V_1}$, and $(H,b_{V_2})$ is a action of $H$ on $\Sigma_{V_2}$. Then for every word $g_1h_2 \dots g_nh_n$, its action is defined by $a_{g_1} \circ b_{h_2} \cdots a_{g_n} \circ b_{h_n}$. We first show how to express the group action of the free product in our syntax.
Let $V_1 = \{v_1, \dots, v_k\}$, $V_2 = \{v'_1, \dots, v'_n\}$ and let $a_g$ map each integer $\alpha_i$
to $\alpha'_i$, and $b_h$ map each integer $\beta_j$ to $\beta_j'$.
The action of the free product can be expressed in our syntax as follows:
\begin{align*}
  \actiontemp{S_G, H_G}{R_S, R_H}{V_1 \cup V_2}{ &                                                                                                                                  \\ \forall g \in S_G , &(g)\cdot(v_1 \rightarrow \alpha_1, \dots, v_k \rightarrow \alpha_k, v'_1 \rightarrow \beta_1, \dots, v'_n \rightarrow \beta_n) = \\
                                                 & v_1 \rightarrow \alpha_1', \dots, v_k \rightarrow \alpha'_k, v'_1 \rightarrow \beta_1, \dots, v'_n \rightarrow \beta_n           \\
  \forall h \in S_H ,                            & (h)\cdot(v_1 \rightarrow \alpha_1, \dots, v_k \rightarrow \alpha_k, v'_1 \rightarrow \beta_1, \dots, v'_n \rightarrow \beta_n) = \\
                                                 & v_1 \rightarrow \alpha_1, \dots, v_k \rightarrow \alpha_k, v'_1 \rightarrow \beta'_1, \dots, v'_n \rightarrow \beta'_n}
\end{align*}
Now, we show that this is a valid group action.
\begin{restatable}{lemma}{FP-Actions}
  \label{lem: fp-action}
  Let $I = G*H$, and let $\tuple{G}{a}{V_1}$ be the action of $G$ on $\Sigma_{V_1}$, and $(H,b_{V_2})$ be the action of $H$ on $\Sigma_{V_2}$. Then for every word $g_1h_2 \dots g_nh_n$, its action on $\Sigma_{V_1 \cup V_2}$ is defined by $a_{g_1} \circ b_{h_2} \cdots a_{g_n} \circ b_{h_n}$.
\end{restatable}
  \begin{proof}
    We first note that the action  $\tuple{G}{a}{V_1}$ can be converted to an action  $\tuple{G}{a}{V_1 \cup V_2}$, where $a$ fixes every variable in $V_1 \setminus V_2$.  We can similarly convert $(H,b_{V_2})$ to an action $(H,b_{V_1 \cup V_2})$. By Lemma \ref{lem: action-homo}, we have a homomorphism $f_1: G \rightarrow Sym(\Sigma_{V_1 \cup V_2})$ and $f_2: H \rightarrow Sym(\Sigma_{V_1 \cup V_2}) $. Therefore by the definition of free product, we have a homomorphism $\varphi: G*H \rightarrow Sym(\Sigma_{V_1 \cup V_2})$, which defined a group action of $G*H$. Finally, this homomorphism must extend in a way such that $\varphi(g_1h_2 \dots g_nh_n) = f_{1_{g_1}} \circ f_{2_{h_2}} \cdots f_{1_{g_n}} \circ f_{2_{h_n}}$
  \end{proof}
We will use the notation $\tuple{G * H}{a * b}{V_1 \cup V_2}$ to denote the action of the free product.

\subsection{Syntax of Homomorphisms}
In this section, we prove that the maps defined on generators in our syntax are indeed homomorphisms.
\DirProdHomo*
\begin{proof}
We first show that the group operation is preserved.
Let \( g_1, g_2 \in G \). Then since \( \varphi \) and \( \psi \) are homomorphisms we have:
   \[
   \varphi(g_1g_2) = \varphi(g_1)\,\varphi(g_2) \quad \text{and} \quad \psi(g_1g_2) = \psi(g_1)\,\psi(g_2).
   \]
   Hence,
   \[
   (\varphi\times \psi)(g_1g_2) = \varphi(g_1g_2)\,\psi(g_1g_2) = \varphi(g_1)\,\varphi(g_2)\,\psi(g_1)\,\psi(g_2).
   \]
   Because in the direct product \( H\times I \) the generators coming from \( S_H \) and \( S_I \) commute (i.e. letters from \( S_H \) commute with letters from \( S_I \)), we can rearrange the word:
   \[
   \varphi(g_1)\,\varphi(g_2)\,\psi(g_1)\,\psi(g_2) = \varphi(g_1)\,\psi(g_1)\,\varphi(g_2)\,\psi(g_2).
   \]
   That is,
   \[
   (\varphi\times \psi)(g_1g_2) = (\varphi\times \psi)(g_1)\,(\varphi\times \psi)(g_2).
   \]
  Next, we verify the identity is preserved.  
   Let \( e_G \) denote the identity in \( G \). Since homomorphisms send the identity to the identity, we have:
   \[
   \varphi(e_G) = e_H \quad \text{and} \quad \psi(e_G) = e_I.
   \]
   In our syntax, the identity in \( H \times I \) is represented by the word that corresponds to \( e_H \) (the empty word over \( S_H \)) concatenated with \( e_I \) (the empty word over \( S_I \)), which is just the empty word in the alphabet \( S_H\cup S_I \). Hence,
   \[
   (\varphi\times \psi)(e_G) = \varphi(e_G)\,\psi(e_G) = e_H\,e_I = e_{H\times I}.
   \]

Since both the operation and the identity are preserved, the map \( \varphi\times \psi \) is indeed a group homomorphism.
\end{proof}
\subsubsection{Free Product}
\label{app: fp-hom}
    We now extend our syntax of homomorphism to the free product. Suppose we have two homomorphisms $\varphi: G \rightarrow I$, and $\psi: H \rightarrow J$. Then we define $\varphi * \psi: G * H \rightarrow J * I$ as follows: 
\[
\varphi * \psi \triangleq \Big\langle \forall s\in S_G\cup S_H,\, \varphi * \psi(s) = 
\begin{cases}
\phi(s) & \text{if } s\in S_G,\\[1mm]
\psi(s) & \text{if } s\in S_H,
\end{cases}
\Big\rangle.
\]
The next lemma shows that this is a homomorphism. 
\begin{restatable}{lemma}{FPHomo}
    
    \(\varphi * \psi\) extends uniquely to a homomorphism.
\end{restatable}
\begin{proof}
       The free product \(G * H\) has the universal property that any pair of homomorphisms \(\phi: G \to K\) and \(\psi: H \to K\) (where \(K\) is any group) extends uniquely to a homomorphism \(G * H \to K\). In our case, taking \(K = I * J\) and using the inclusions \(I, J \subset I*J\) we get that \(\varphi * \psi\) is a homomorphism. 
\end{proof}

\subsection{Composition of Homomorphisms}
\label{app: compose}
Given two homomorphisms $\varphi: G \rightarrow H$, and $\psi: H \rightarrow I$, we want to define their composition, denoted by $\psi \circ \varphi : G \rightarrow I$ as a map that first applies $\varphi$ and then applies $\psi$. By definition of \(\varphi\), for every generator \(g \in S_G\), \(\varphi(g) = h_1 \cdots h_k, h_i \in S_H\). We can then apply \(\psi\) to each generator \(h_i\). Thus, we define
$\psi \circ \varphi: G \rightarrow I$ in our syntax as follows: 
\[
 (\psi \circ \varphi) \triangleq \Big\langle\forall g 
 \in S_G, (\psi \circ \varphi)(g) = \psi(h_1) \cdots \psi(h_k) \text{ where } \varphi(g) = h_1\cdots h_k \; h_i \in S_H \Big\rangle
\]
As expected, this extends to a unique homomorphism.
\begin{restatable}{lemma}{CompHomo}
    \label{lem: compose-homomorphism}
    \((\psi \circ \varphi) \) extends uniquely to a homomorphism.
\end{restatable}
\begin{proof}
    Consider any relation in $R_G$. By Lemma \ref{lem:johnson1}, $\varphi(r_j) = e_H$, and since $\psi$ is a homomorphism, $\psi(e_H) = e_I$ then by Lemma \ref{lem:johnson1}  \((\psi \circ \varphi) \) is a homomorphism. 
\end{proof}
\subsection{Entailment of Group Actions}
\ImpProp*
  \begin{proof}
    The proofs are routine.  We only prove (1) and (4) here.  \\
    (1)  Consider any element $(g,h) \in G \times H$. By the definition of $\tuple{G \times H}{a \times b}{V_1\uplus V_2} $, all the variables in $V_1$ are acted on only by $\overline{a}_g$. Similarly, for $\tuple{G}{a}{V_1}$ the variables in \(V_1\) are acted on by \(\overline{a}_g\). The map $\mathbf{proj_1}: G \times H \rightarrow G$ such that $\mathbf{proj_1}((g,h)) = g$ is a homomorphism.
    Therefore
    \[
      \forall (g,h) \in G \times H, \forall \sigma, \forall v\in V_1, a\times \overline{b_{(g,h)}}(\sigma)(v) = \overline{a_{\mathbf{proj_1}(g,h)}}(\sigma)(v)
    \]
    Therefore,  $\tuple{G \times H}{a \times b}{V_1\uplus V_2}  \overset{ \mathbf{proj_1}}{\rightarrow} \tuple{G}{a}{V}$. \\
    (4) Consider any element $g \in G$. By the definition of lifting for every variable $v \in V_1$, $a_g(\sigma)(v) = \overline{a}_g(\sigma)(v)$. Then,
    \[
      \forall g \in G, \forall \sigma, \forall v\in V_1, \overline{a}(\sigma)(v) = \overline{a}_{\mathbf{eq
    }(g)}(\sigma)(v)
    \]
    Therefore, $\tuple{G}{a}{\mathit{Vars}}\overset{\mathbf{eq}}{\rightarrow} \tuple{G}{a}{V_1}$.
  \end{proof}
  \section{Symmetry Logic}
    \subsection{Programming Language Syntax}    
    \label{app:syntax}
\begin{figure}[t]
  \begin{alignat*}{2}
    \mathcal{EXP}  & \ni e &&\coloneqq x \in \mathit{Vars} \mid n \in \mathbb{Z} \mid f(e_1, \dots, e_n), f \in Ops                                                                        \\
    \mathcal{BEXP} & \ni b &&\coloneqq true \mid false \mid e_1 = e_2 \mid b_1 \land b_2 \mid\lnot b_1 \mid e_1 \leq e_2 \mid e_1 = e_2                                           \\
    \mathcal{COM}  & \ni C &&\coloneqq \textbf{skip} \mid x \coloneqq e  \mid C_1; C_2 \mid \textbf{for } {(t \coloneq 0; t < \alpha_1; t \coloneq t + \alpha_2)} \textbf{ do } C \mid \\& && \textbf{if }B\;\textbf{then}\;C_1\;\textbf{else}\;C_2 \mid \textbf{while } B \textbf{ do } C
  \end{alignat*}
  \caption{Core Programming Language Syntax}
  \label{fig:lang-syntax}
\end{figure}
The full programming language is presented in Figure \ref{fig:lang-syntax}.

\subsection{Programming Language Semantics}
\label{app:sem}
    We define three semantics functions $\mathcal{A}\sem{\_}: \mathcal{EXP} \times \Sigma \rightarrow \mathbb{Z}$,  $\mathcal{B}\sem{\_}: \mathcal{BEXP} \times \Sigma \rightarrow Bool$ and $\mathcal{C}\sem{\_}: \mathcal{COM} \times \Sigma \rightarrow \Sigma$, which define the denotational semantics of our programming language in the style of \citet{winskel1993formal}.  
The semantics assume that for every function symbol in the set $Ops$, there is a corresponding mathematical operation.
  We define these functions using structural induction.
  \begin{align*}
    \mathcal{A}\sem{n}_{\sigma}                   & = n         \\
    \mathcal{A}\sem{x}_{\sigma}                   & = \sigma(x) \\
    \mathcal{A}\sem{f(e_1 , \dots,  e_n}_{\sigma} & =
    f(\mathcal{A}\sem{e_1}_{\sigma} , \dots,  \mathcal{A}\sem{e_n}_{\sigma})
  \end{align*}
  For Boolean Expressions:
  \begin{align*}
    \mathcal{B}\sem{true}_{\sigma}          & = true                   \\
    \mathcal{B}\sem{false}_{\sigma}         & = false                  \\
    \mathcal{B}\sem{e_1 = e_2}_{\sigma}     & =
    \mathcal{A}\sem{e_1}_{\sigma}  = \mathcal{A}\sem{e_2}_{\sigma}     \\
    \mathcal{B}\sem{e_1 \leq e_2}_{\sigma}  & =
    \mathcal{A}\sem{e_1}_{\sigma}  \leq \mathcal{A}\sem{e_2}_{\sigma}  \\
    \mathcal{B}\sem{b_1 \land b_2}_{\sigma} & =
    \mathcal{B}\sem{e_1}_{\sigma}  \land \mathcal{B}\sem{e_2}_{\sigma} \\
    \mathcal{B}\sem{\lnot b}_{\sigma}       & =
    \lnot \mathcal{B}\sem{b}_{\sigma}
  \end{align*}
  For program commands:
  \begin{align*}
    \mathcal{C}\sem{\textbf{skip}}_{\sigma}                                                                    & = \sigma                                                                                                           \\
    \mathcal{C}\sem{x \coloneq e}_{\sigma}                                                                     & = \sigma[x \mapsto \mathcal{A}\sem{e}_{\sigma}]                                                                    \\
    \mathcal{C}\sem{c_0; c_1}                                                                                  & = \mathcal{C}\sem{c_1} \circ \mathcal{C}\sem{c_0}                                                                  \\
    \mathcal{C}\sem{\textbf{if } b \textbf{ then } C_1 \textbf{ else } C_1}_{\sigma} &= 
    \begin{cases}
        \sem{C_1}_{\sigma} & \mathcal{B}\sem{b}_{\sigma} = true \\ 
        \sem{C_2}_{\sigma} & \mathcal{B}\sem{b}_{\sigma} = false \\ 
    \end{cases} \\
    \mathcal{C}\sem{\textbf{for }(t \coloneq 0; t < \alpha; t \coloneq t + dt) \textbf{ do } C}_{\sigma} & = \underbrace{\mathcal{C}\sem{C; \dots ;C}_{\sigma}}_\text{n steps till $\mathcal{B}\sem{t < T}_{\sigma} = false$}
    \\
  \mathcal{C}\sem{\textbf{while } b \textbf{ do } C}
  &=
  \operatorname{lfp}(\Phi)
  \end{align*}
  where $\Phi$ is defined as
  \[
  \Phi(  \mathcal{C}\sem{\textbf{while } b \textbf{ do } C})_\sigma \;=\;
  \begin{cases}
       \mathcal{C}\sem{\textbf{while } b \textbf{ do } C}_{\bigl(\mathcal{C}\sem{C}(\sigma)\bigr)}, & \text{if } \mathcal{B}\sem{b}(\sigma) = \text{true},\\[6pt]
     \sigma, & \text{if } \mathcal{B}\sem{b}(\sigma) = \text{false}.
  \end{cases}
\]
Since we assume loops always terminate, we can safely assume the lfp exists.
\subsection*{The $\mathsf{POST}$ Transformer}
Recall the definition of \(\mathsf{POST}\).
\POSTDef*
    The following theorem shows that the map produced by the $\mathsf{POST}$ transformer is a group action.
 \PostGpActionThm*
  \begin{proof}
    Let $\tuple{G}{a^*}{\mathit{Vars}}$ denote the action $\mathsf{POST}((G,a_V),f)$.  Let the $\Sigma_f$ denote the set of program states where $v_1$ is in the range of $f$. Clearly, $\Sigma_f \subseteq \Sigma$. We begin by noting that for any $\sigma \in \Sigma_f$,  $\mathsf{POST}((G,a_V),f)(\sigma) \in \Sigma_f$. Therefore, we can consider two cases. Let $\sigma$ be any program state in $\Sigma$:
    \begin{enumerate}
      \item $\sigma \in \Sigma_f$: We need to verify that the group action axioms hold:
            \begin{enumerate}
              \item Compatibility Axiom:
                    \begin{align*}
                      a^*_h(a^*_g(\sigma)) & = a^*_h(v_1 \rightarrow f(\overline{a}_g(C^{-1}(\sigma))(v_1), \dots, \overline{a}_g(C^{-1}(\sigma))(v_k))                          \\& \dots, v_n \rightarrow \overline{a}_g(C^{-1}(\sigma))(v_n), \dots, v_m \rightarrow \alpha_m) \\
                                           & =v_1 \rightarrow f(\overline{a}_h(\overline{a}_g(C^{-1}(\sigma)))(v_1), \dots, \overline{a}_h(\overline{a}_g(C^{-1}(\sigma))(v_k))) \\& \dots, v_n \rightarrow \overline{a}_h(\overline{a}_g(C^{-1}(\sigma)))(v_n), \dots, v_m \rightarrow \alpha_m\\
                                           & = v_1 \rightarrow f(\overline{a}_{hg}(C^{-1}(\sigma))(v_1), \dots, \overline{a}_{hg}(C^{-1}(\sigma))(v_k))                          \\& \dots, v_n \rightarrow \overline{a}_{hg}(C^{-1}(\sigma))(v_n), \dots, v_m \rightarrow \alpha_m) \\
                                           & = a^*_{hg}(\sigma)
                    \end{align*}
              \item Identity Axiom:
                    \begin{align*}
                      a^*_e(\sigma) & = v_1 \rightarrow f(\overline{a}_e(C^{-1}(\sigma))(v_1), \dots, \overline{a}_e(C^{-1}(\sigma))(v_k)) \\& \dots, v_n \rightarrow \overline{a}_e(C^{-1}(\sigma))(v_n), \dots, v_m \rightarrow \alpha_m \\
                                    & = \sigma.
                    \end{align*}
            \end{enumerate}
      \item $\sigma \notin \Sigma_f$: Compatibility:
\[              a^*_h(a^*_g(\sigma))  = a^*_h(\sigma)     = \sigma         = a^*_{hg}(\sigma)
           \]
            Identity:
            \(
              a^*_e(\sigma) = \sigma.
           \)
    \end{enumerate}
  \end{proof}
\subsection{Omitted Structural Rules }
\label{app:rules}
We present the omitted free product rule for our logic. Intuitively, the $\textsc{FREE-PROD}$ rule `glues' together pre- and post-condition groups to create a new valid triple.
\begin{mathpar}
 \inferrule*[left={$\textsc{FREE-PROD}$},right={}]{\triple{(G_1,a_{\mathit{Vars}})}{C} {(H_1,b_{\mathit{Vars}})}{\varphi} \quad \triple{(G_2,c_{\mathit{Vars}})}{C} {(H_2,d_{\mathit{Vars}})}{\eta}}{(\triple{G_1 * G_2, a * c_{\mathit{Vars}})}{C}{ (H_1 * H_2, b *d_{\mathit{Vars}})}{\varphi* \eta}} 
\end{mathpar}
\subsection{Soundness Of Rules}

We show that the rules presented in Figures
\ref{fig:rules-stmt} and \ref{fig:rules-struct} are sound.
  Before we prove the soundness theorem, we prove an important lemma:
  \begin{restatable}{lemma}{ComposeMultiple}
    \label{lem: compose-multiple}
    Let $(H,b_{\mathit{Vars}})$, $(H_2,k_{\mathit{Vars}})$, $(G_1,a_{\mathit{Vars}})$, and $(G_2, i_{\mathit{Vars}})$ be group actions. If $\forall g_1 \in G_1$
    $\forall \sigma \in \Sigma$, $b_{h_1}(\sem{C}_{\sigma}) = C(a_{g_1}(\sigma))$ and $\forall g_2 \in G_2$, $k_{h_2}(\sem{C}_{\sigma}) = C(i_{g_2}(\sigma))$, then
    $b_{h_1}(k_{h_2}(\sem{C}_{\sigma})) = C(a_{g_1}(i_{g_2}(\sigma)))$
  \end{restatable}
  \begin{proof}
  First, we observe that:
\[        b_{h_1}(k_{h_2}(\sem{C}_{\sigma}))      =  b_{h_1}(\sem{C}_{(i_{g_2}}(\sigma)))\]
Now if we let  $i_{g_2}(\sigma) = \sigma'$, then:
    \(
        b_{h_1}(\sem{C}_{(i_{g_2}(\sigma)}))   =  b_{h_1}(\sem{C}_{\sigma'})       = \sem{C}_{a_{g_1}(\sigma')}           =  \sem{C}_{a_{g_1}(i_{g_2}(\sigma))}    
  \), as desired
  \end{proof}

\Soundness*

  \begin{proof}
    We proceed by induction on the structure of the derivation. \\    \textbf{Base Cases: } \\

    \fbox{\textbf{SKIP :}} By the semantics of $\mathbf{skip}$,
    for any program state $\sigma$, $\sem{\mathbf{skip}}_{\sigma} = \sigma$, and $
      \sem{\mathbf{skip}}_{a_g(\sigma)}= a_g(\sigma)$.  Therefore, \[\forall g \in G, \forall \sigma \in \Sigma, \forall v \in V,   a_{\mathbf{eq}(g)}(\sem{\mathbf{skip}}_\sigma)(v) = \sem{\mathbf{skip}}_{a_g(\sigma)}(v)\] This implies
    $\triple{\tuple{G}{a}{V}}{\mathbf{skip}}{\tuple{G}{a}{V}}{\mathbf{eq}}$.  \\

    \fbox{\textbf{ASSGN :}} We consider assignment statements of the form $v_1 \coloneq f(v_1, \dots, v_k)$.
    We begin by noting that it suffices to consider the program state in $\Sigma_f$, which is the set of states reachable after the assignment command. Let $g$ be any group element in $G$.
    Consider $\mathsf{POST}(\tuple{G}{a}{V_1}, f)_g(\sem{C}_\sigma)(v_1)$. By the definition of $\mathsf{POST}$ and lifting,
    \begin{align*}
      \overline{\mathsf{POST}(\tuple{G}{a}{V_1}, f)}_g(\sem{C}_\sigma)(v_1) & = f(\overline{a}_g(C^{-1}(\sem{C}_{\sigma}))(v_1), \dots, \overline{a}_g(C^{-1}(\sem{C}_{\sigma}))(v_k))(v_1) \\
                                                                            & = f(\overline{a}_g(\sigma)(v_1), \dots, \overline{a}_g(\sigma)(v_k))(v_1)
    \end{align*}
    At the same time,
    \begin{align*}
      \sem{v_1 \coloneq f(v_1, \dots, v_k)}_{\overline{a}_g(\sigma)}(v_1) & = f(\overline{a}_g(\sigma)(v_1), \dots, \overline{a}_g(\sigma)(v_k))(v_1)
    \end{align*}
    Therefore,
    \[
      \overline{ \mathsf{POST}(\tuple{G}{a}{V_1}, f)_g}(v_1) =  \sem{v_1 \coloneq f(v_1, \dots, v_k)}_{\overline{a}_g(\sigma)}(v_1)
    \]
    Now, consider any $v_i \in \mathit{Vars} \setminus \{v_1\}$. If $v_i \in V_1$
    \begin{align*}
      \overline{\mathsf{POST}(\tuple{G}{a}{V_1}, f)}_g(\sem{C}_\sigma)(v_i) & = \overline{a}_g(C^{-1}(\sem{C}_{\sigma}))(v_i)                        \\
                                                                            & = \overline{a}_g(\sigma)(v_i)                                          \\
                                                                            & =  \sem{v_1 \coloneq f(v_1, \dots, v_k)}_{\overline{a}_g(\sigma)}(v_i)
    \end{align*}
    If $v_i \notin V_1$
    \begin{align*}
      \overline{\mathsf{POST}(\tuple{G}{a}{V_1}, f)}_g(\sem{C}_\sigma)(v_i) & = \sigma(v_i)                                                                                                                \\
                                                                            & = \overline{a}_g(\sigma)(v_i)                                          &  & \text{(The lifting leaves } v_i \text{ unchanged)} \\
                                                                            & =  \sem{v_1 \coloneq f(v_1, \dots, v_k)}_{\overline{a}_g(\sigma)}(v_i)
    \end{align*}
    Therefore, $\forall g \in G, \forall \sigma \in \Sigma$,
    \[\forall v_i \in V, \overline{\mathsf{POST}(\tuple{G}{a}{V_1}, f)}_g(v_i) =  \sem{v_1 \coloneq f(v_1, \dots, v_k)}_{\overline{a}_g(\sigma)}(v_i)\]
   Thus, the triple 
    \[\triple{(G, a_{V_1})}{ v_1 \coloneq f(v_1, \dots , v_k)} {\mathsf{POST}(\tuple{G}{a}{V_1}, f)}{\mathbf{eq}}\] holds.  \\

    \fbox{\textbf{CONST :}} By the side condition, for all variables $v_i \in V$, $ \forall \sigma \in \Sigma, \sem{C}_{\sigma}(v_i) = \sigma(v_i)$. Therefore, for any group element $g \in G$, and variable $v_i \in V$, $\overline{a}(\sem{C}_{\sigma})(v_i) = \overline{a}(\sigma)(v_i)$. Meanwhile, $\sem{C}_{\overline{a}_g(\sigma)}(v_i) =  \overline{a}(\sigma)(v_i)$. Therefore,
    \[
      \forall g \in G, \forall \sigma \in \Sigma,
      \forall v_i \in V, \overline{a}_g(\sem{C}_{\sigma})(v_i) = \sem{C}_{\overline{a}_g(\sigma)}(v_i)
    \]
    As before, the map $\varphi(g) = g$ is a homomorphism. Therefore, the triple $\triple{(G, a_{V})}{ C }{(G, a_{V})}{\mathbf{eq}}$ is valid. \\

    \fbox{\textbf{SEM-ASSGN: }} This rule is sound since the premise is the definition of validity.
    \\

    \fbox{\textbf{ID: }} Consider any variable in $\mathit{Vars}\setminus(V \cup \mathit{Vars}(C))$. This variable will be unchanged by the group action, and the command, and therefore it remains unchanged no matter what $(G, a_{V})$ is. Thus, the rule is sound.
    \\
    \textbf{Inductive cases :}\\

    \fbox{\textbf{SEQ :} }
    We need to show that there exists a homomorphism $\alpha: G \rightarrow I$ such that
    \[\forall g \in G, \forall \sigma \in \Sigma,
      \forall v \in V_2 \;\overline{c}_{\alpha(g)}(\sem{C_1; C_2}_{\sigma})(v) = \sem{C_1; C_2}_{\overline{a}_{g}(\sigma)}(v)\]
    By the induction hypothesis, $\vDash (G, a_{V_1}) \;C_1\;(H, b_{\mathit{Vars}})$ and $ \vDash (H, b_{\mathit{Vars}})\; C_2\;(I, c_{V_2})$.
    Therefore there exists a homomorphism $\varphi$ such that,
    \[
      \forall g \in G, \forall \sigma \in \Sigma\; \forall v \in \mathit{Vars},
      b_{\varphi(g)}(\sem{C_1}_{\sigma})(v) = \sem{C_1}_{\overline{a}_g(\sigma)}(v)
    \]
    Since all the variables in the domain of the program state are related, this is equivalent to saying
    \[
      \forall g \in G, \forall \sigma \in \Sigma\;
      b_{\varphi(g)}(\sem{C_1}_{\sigma}) = \sem{C_1}_{\overline{a}_g(\sigma)}
    \]
    Also by the induction hypothesis, there exists a homomorphism $\eta$ such that,
    \[
      \forall h \in H, \forall \sigma \in \Sigma\; \forall v \in V_2,
      i_{\eta(h)}(\sem{C_2}_{\sigma})(v) = \sem{C_2}_{b_h(\sigma)}(v)
    \]
    By the second premise,
    \begin{equation}
      \forall h \in H, \forall \sigma \in \Sigma\; \forall v \in V_2,
      i_{\eta(h)}(\sem{C_2}_{\sem{C_1}_\sigma})(v) = \sem{C_2}_{b_h(\sem{C_1}_\sigma)}(v)
    \end{equation}

    Now, consider any $g \in G$. By the first premise, for any program states $\sigma \in \Sigma$,
    \[
      \sem{C_2}_{\overline{a}_g(\sem{C_1}_\sigma)} = \sem{C_2}_{b_{\varphi(g)}(\sem{C_1}_\sigma)}
    \]
    Therefore, by (2)
    \[
      \forall g \in G, \forall \sigma \in \Sigma\; \forall v \in V_2,
      i_{\eta(\varphi(g))}(\sem{C_2}_{\sem{C_1}_\sigma})(v) = \sem{C_2}_{\overline{a}_g(\sem{C_1}_\sigma)}(v)
    \]
    Now, if we let $\alpha$ be the composition of $\eta \circ \varphi$, then $\alpha$ is also a homomorphism. Therefore,
    \[\forall g \in G, \forall \sigma \in \Sigma,
      \forall v \in V_2\;\overline{c}_{\alpha(g)}(\sem{C_1; C_2}_{\sigma})(v) = \sem{C_1; C_2}_{\overline{a}_{g}(\sigma)}(v).\]

    \fbox{\textbf{IF :}}
    By the induction hypothesis, $(G, a_{V_1}) \overset{\mathbf{e*}}{\rightarrow}
    (E, e_{\{b\}})$. This means that for any group element $g \in G$, and program state
    $\sigma \in \Sigma$, $\sem{x}_{\sigma} = \sem{x}_{a_g(\sigma)}$. 
    We consider two cases:  \begin{enumerate}
        \item If  $\sem{x}_{\sigma} = \sem{x}_{a_g(\sigma)} = true$: \\
        $\sem{\mathbf{if }\;x\;\mathbf{then}\;C_1\;\mathbf{else}\;C_2}_{\sigma} = \sem{C_1}_{\sigma}$, and  $\sem{\mathbf{if }\;x\;\mathbf{then}\;C_1\;\mathbf{else}\;C_2}_{a_g(\sigma)} = \sem{C_1}_{a_g(\sigma)}$. Then by the induction hypothesis, $\sem{\mathbf{if }\;x\;\mathbf{then}\;C_1\;\mathbf{else}\;C_2}_{a_g(\sigma)} = b_{\phi(g)}(\sem{\mathbf{if }\;x\;\mathbf{then}\;C_1\;\mathbf{else}\;C_2}_{\sigma})$.
        \item  If  $\sem{x}_{\sigma} = \sem{x}_{a_g(\sigma)} = false$: \\
        $\sem{\mathbf{if }\;x\;\mathbf{then}\;C_1\;\mathbf{else}\;C_2}_{\sigma} = \sem{C_2}_{\sigma}$, and  $\sem{\mathbf{if }\;x\;\mathbf{then}\;C_1\;\mathbf{else}\;C_2}_{a_g(\sigma)} = \sem{C_2}_{a_g(\sigma)}$. Then by the induction hypothesis,  $\sem{\mathbf{if }\;x\;\mathbf{then}\;C_1\;\mathbf{else}\;C_2}_{a_g(\sigma)} = b_{\phi(g)}(\sem{\mathbf{if }\;x\;\mathbf{then}\;C_1\;\mathbf{else}\;C_2}_{\sigma})$.
    \end{enumerate}
which gives us the desired result. 

\fbox{\textbf{WHILE :}}
By the I.H, we know that $\vDash \triple{\tuple{G}{a}{V}}{C}{\tuple{G}{a}{V}}{\varphi}$. 
Now, since the variables modified by the command $C$ are all in $V$, every variable in $\mathit{Vars} \setminus V$ is unchanged. Further, $\tuple{G}{\overline{a}}{\mathit{Vars}}$ does not change any of the variables in $\mathit{Vars} \setminus V$ either. Therefore, the triple $\triple{\tuple{G}{\overline{a}}{\mathit{Vars}}}{C}{\tuple{G}{\overline{a}}{\mathit{Vars}}}{\varphi}$ and the triple $\triple{\tuple{G}{a}{V}}{C}{\tuple{G}{\overline{a}}{\mathit{Vars}}}{\varphi}$ are both sound. 
Next,  we know $(G, a_{V}) \overset{\mathbf{e*}}{\rightarrow}
    (E, e_{\{x\}})$. This means that for any group element $g \in G$, and program state
    $\sigma \in \Sigma$, $\sem{x}_{\sigma} = \sem{x}_{a_g(\sigma)}$. Since $x \in V$ the guard is unchanged.
By assumption, the loop terminates. Let it terminate after $n$ iterations.
 We note that $\mathbf{eq}^n = \mathbf{eq}$. 
 We know that for any $n$ iterations, by sequencing, 
 \[
 a_{\mathbf{eq^n(g)}(\sigma)}(\sem{C_1;\dots C_1}_{\sigma}) = \sem{C_1;\dots C_1}_{a_g(\sigma)}
 \]
 which implies $\triple{(G,a_{V})}{\textbf{while } B \textbf{ do } C}{(G,a_V)}{\mathbf{eq}}$
 
    \fbox{\textbf{FOR :}}
    Since the variables $t$, $\alpha_1$, and $\alpha_2$ are unchanged, the loop guard is unchanged by the group action. We also know that $\vDash \triple{\tuple{G}{a}{V}}{C}{\tuple{G}{a}{V}}{\varphi}$. This means that there $\exists \phi: G \rightarrow H$ such that for all $g \in G, \forall \sigma \in \Sigma \; \overline{a}_{\phi(g)}(\sem{C}_{\sigma}) = \sem{C}_{\overline{a}_g(\sigma)}$.
    Since the variables modified by the command $C$ are all in $V$, every variable in $\mathit{Vars} \setminus V$ is unchanged. Further, $\tuple{G}{\overline{a}}{\mathit{Vars}}$ does not change any of the variables in $\mathit{Vars} \setminus V$ either. Therefore, the triple $\triple{\tuple{G}{\overline{a}}{\mathit{Vars}}}{C}{\tuple{G}{\overline{a}}{\mathit{Vars}}}{\varphi}$ and the triple $\triple{\tuple{G}{a}{V}}{C}{\tuple{G}{\overline{a}}{\mathit{Vars}}}{\varphi}$ are both sound. By assumption, the loop terminates. Let it terminate after $n$ iterations. Since $\alpha_1$ is a constant, and $t = t+ \alpha_2$ remains unchanged, the loop will iterate for $n$ iteration for all program states.
    Since the loop terminates in a finite number of steps, by the repeated application of the $\textsc{SEQ}$ rule,
    $\triple{\tuple{G}{a}{V}}{C;. \dots, C}{\tuple{G}{\overline{a}}{\mathit{Vars}}}{\varphi^n}$ is sound.  By $\textsc{CONS-1}$, $\triple{\tuple{G}{a}{V}}{C;\dots, C}{\tuple{G}{a}{V}}{\varphi^n}$ is sound. By the semantics of $\textsc{FOR}$,  \((G, a_{V})\; \textbf{for } {(t \coloneq 0; t < \alpha_1; t \coloneq t + \alpha_2)}_{\varphi^n} \textbf{ do } C \;(G, a_{V})\) is sound.

    \fbox{\textbf{LIFT :}} This rule is sound trivially by the definition of validity.

    \fbox{\textbf{DIR-PROD :}}   By the induction hypothesis and inversion, we know $\vDash \triple{(G, a_{V_1})}{C_1}{(H, b_{V_2})}{\varphi}$ and
    $\vDash (G, a_{V_1})\;C_{2\alpha}\;(I, c_{V_3})$.
    Therefore,   \[\forall \sigma \in \Sigma, \forall v \in V_2, \overline{b}_{\varphi(g)}(\sem{C}_{\sigma})(v)= \sem{C}_{\overline{a}_g(\sigma)}(v)\] and
    \[\sigma \in \Sigma, \forall v \in V_3,  \overline{c}_{\alpha(g)}(\sem{C}_{\sigma})(v)= \sem{C}_{\overline{a}_g(\sigma)}(v)\]
    Consider any $v \in V_2 \cup V_3$, and the group element $(h, i)$ in the group $(H \times I)$.
    Since $V_2$ and $V_3$ are disjoint, $v$ is either in $V_2$ or $V_3$.
    By the definition of the group action $(H \times I, b \times c)$, for any program state in $\Sigma$, the program state can be split
    into three disjoint parts. The first part has domain $V_1$, and $b_h$ acts on this.
    The second part has domain $V_2$, and $c_i$ acts on this. The third is the rest of $\sigma$, which is mapped to itself.
    Now, by the first premise if $v$ is in $V_2$,
    $\overline{b}_{\varphi(g)}(\sem{C}_{\sigma})(v)= \sem{C}_{\overline{a}_g(\sigma)}(v)$, and by the second premise,
    if $v$ is in $V_3$,
    $\overline{c}_{\alpha(g)}(\sem{C}_{\sigma})(v)= \sem{C}_{\overline{a}_g(\sigma)}(v)$.
    Since the action of $(H \times I, b \times c)$ is defined such that if
    if $v$ is in $V_2$,  $ (b \times c)_{(h,i)}(\sigma) = b_h(\sigma)$ and similarly for if $V \in V_3$. Additionally, we the map $\eta: G \rightarrow H \times I$ such that $\eta(g) = (\varphi(g), \alpha(g))$ is a homomorphism. Further, $\eta = \varphi 
    \times \alpha$.
    Therefore,
    \[\forall g \in G, \exists \in H \times I \text{ such that } \forall \sigma \in \Sigma \forall v \in V_2 \cup V_3, \overline{(b \times c)}_{\alpha(g)}(\sem{C}_{\sigma})(v)= \sem{C}_{\overline{a}_g(\sigma)}(v)
    \]

    \fbox{\textbf{FREE-PROD : }}
    Consider any element of the group $G * I$. Suppose it is of the form $\alpha = g_1i_2\dots i_n$. Then the action of  $a * c_\alpha(\sigma) = a_{g_1}\cdot c_{i_1} \cdots i_n (\sigma) $.
    Let $h_1d_1, \dots, d_n$ be the word of group elements such that
    $b_{h_1}(\sem{C}_{\sigma}) = \sem{C}_{a_{g_{1}}(\sigma)}$ and $j_{d_1}(\sem{C}_{\sigma}) = \sem{C}_{c_{i_{1}}(\sigma)}$, and so on. Then by Lemma \ref{lem: compose-multiple}, for any $v \in V$,
    $b * d_{h_1d_1, \dots, d_n}(\sem{C}_\sigma)(v)= \sem{C}_{a*c_{\alpha}(\sigma)}(v)$. A symmetric argument holds for an element $\alpha' = i_1g_1\dots g_n$. Finally, we have a homomorphism from $G_1 \rightarrow H_1 * H_2$ by composing the homomorphism from $G_1$ to $H_1$ with the injection of $H_1$ into $H_2$. This gives a homomorphism from $G_1$ into $H_1 * H_2$, and similarly a homomorphism from $G_2 \rightarrow H_2$, which is compatible with the original homomorphism, i.e, the homomorphism is exactly $\varphi * \eta$. Therefore, the triple is sound.

    \fbox{\textbf{CONSEQUENCE 1: }}
    By the definition of $\rightarrow$, there exists a homomorphism $\varphi: H \rightarrow I$ such that for all $h \in H, \forall \sigma \in \Sigma, \forall v \in V_2, \overline{b}_h(\sigma)(v) = i_{\varphi(h)}(\sigma)(v)$. By the induction hypothesis, there exists a homomorphism $\alpha: G \rightarrow H$ such that for all $g \in G, \forall \sigma \in \Sigma, \forall v \in V_1, \overline{b_{\alpha(g)}}(\sigma)(v) = \sem{C}_{\overline{a}_g(\sigma)}(v)$. Since $V_1 \subseteq V_2$,
    $\forall g \in G, \forall \sigma \in \Sigma, \forall v \in V_2, \overline{i_{\varphi(\alpha(g))}}(\sigma)(v) = \sem{C}_{\overline{a}_g(\sigma)}(v)$. Therefore, the rule is sound. \\

    \fbox{\textbf{CONSEQUENCE 2: }}
    By the definition of $\rightarrow$, there exists a homomorphism $\varphi: I \rightarrow G$ such that for all $g \in G, \forall \sigma \in \Sigma, \forall v \in V_1, \overline{a}_g(\sigma)(v) = \overline{c}_{\varphi(g)}(\sigma)(v)$. By the definition of lifting and since both the actions are on the same set, for all $g \in G, \forall \sigma \in \Sigma, \overline{a}_h(\sigma) = \overline{c}_{\varphi(g)}(\sigma)$

    By the induction hypothesis, there exists a homomorphism $\alpha: G \rightarrow H$ such that for all $g \in G, \forall \sigma \in \Sigma, \forall v \in V_2, \overline{b_{\alpha(g)}}(\sigma)(v) = \sem{C}_{\overline{a}_g(\sigma)}(v)$. Therefore,
    $\forall g \in G, \forall \sigma \in \Sigma, \forall v \in V_2, \overline{b_{(\alpha(\varphi(i)))}}(\sigma)(v) = \sem{C}_{\overline{c}_i(\sigma)}(v)$. Therefore, the rule is sound.

  \end{proof}
    \section{Weakest Pre-Condition}
    \label{app: wp}
  
Recall that there are two variants of the weakest pre-condition we consider. We first present the weakest pre-condition, and then prove its existence, before moving on to the weakest faithful pre-condition. The weakest pre-condition is the weakest pre-condition among all faithful group actions, but the weakest faithful pre-condition is the weakest pre-condition among all pre-conditions, provided the post-condition is faithful. We first show that a faithful post-condition implies that the homomorphism \(\varphi\) is unique.
\FaithfulUnique*
  \begin{proof}
  Let $\varphi: G \rightarrow H$ be a map such that $\varphi(g) = h$ and $b_h(\sem{C}_{\sigma}) = \sem{C}_{a_g(\sigma)}$. 
    The uniqueness of \(\varphi\) follows from the fact that $(H,b_{Vars})$ is faithful, so every $h \in H$ acts differently. Therefore, for every group element \(g \in G\), there is exactly one \(h \in H\) such that \(\varphi(g) = h\) would make the triple sound. 
  \end{proof}

Now, we formally define the weakest pre-condition.
\begin{definition}
  \label{def: wp-pre-faithful}
  The weakest pre-condition for a post-condition, $\tuple{H}{b}{\mathit{Vars}}$ and a  command $C$, is a group action, denoted by $\mathsf{wp}(\tuple{H}{b}{\mathit{Vars}}, C)$, such that if
  $\triple{(\mathsf{wp}(\tuple{H}{b}{\mathit{Vars}}, C))}{C}{\tuple{H}{b}{\mathit{Vars}}}{\varphi}$ is valid, and
  for any faithful group action $\tuple{I}{c}{\mathit{Vars}}$, if $\triple{\tuple{I}{c}{\mathit{Vars}}}{C}{\tuple{H}{b}{\mathit{Vars}}}{\tau}$ is valid then there exists a homomorphims \(\eta\) such that $\tuple{I}{c}{\mathit{Vars}} \overset{\eta}{\rightarrow} \mathsf{wp}(\tuple{H}{b}{\mathit{Vars}}, C)$ and $\tau = \varphi \circ \eta$.
\end{definition}

  At a high level, to construct the weakest pre-condition, we will take all the valid faithful pre-conditions and put them together in a group.
  To formalize this construction, we will use a group construction called the \emph{free product}.
  We begin with some basic group-theoretic definitions that generalize the definition of the free product we saw in Section \ref{sec:assertions}.
  \begin{definition}
    \label{def: free-product}
    \cite{rotman2012introduction} Let $\{A_i : i \in I\}$ be a family of groups. A \emph{free product} of the $A_i$ is a group $P$ and a family of homomorphisms $j_i: A_i \rightarrow P$ such that, for every group $G$ and every family of homomorphism $f_i: A_i \rightarrow G$, there exists a unique homomorphism $\varphi: P \rightarrow G$ with $\varphi j_i = f_i$ for all $i$.
    \[\begin{tikzcd}[/tikz/ampersand replacement=\&]
        \& P \\
        {A_i} \& G
        \arrow["\varphi", from=1-2, to=2-2]
        \arrow["{j_i}", from=2-1, to=1-2]
        \arrow["{f_i}"', from=2-1, to=2-2]
      \end{tikzcd}\]
  \end{definition}
  Without loss of generality, for the rest of this section, we assume that each $A^{\#}_i = A_i \setminus \{e\}$ are pairwise disjoint. A free product is unique up to isomorphism:
  \begin{restatable} {theorem}{FPUnique}
    \label{thm: fp-unique}
    \cite{rotman2012introduction}
    Let $\{A_i : i \in I\}$ be a family of groups. If $P$ and $Q$ are each a free product of $A_i$, then $P \cong Q$.
  \end{restatable}  
  Because of Theorem \ref{thm: fp-unique}, we can speak of \emph{the} free product.
  Every element of a free product group is equivalent to a finite sequence made by concatenating elements from each of the $A_i$, and multiplying adjacent elements from the same group.  In other words, every element is of the form $g_1g_2\dots g_n$, where each $g_k$ is an element of $A_i$. More formally:
  \begin{restatable}{lemma}{NormalFP}
    \label{lem:normal-fp}
    Every nontrivial element of a free product of $\{A_i : i \in I\}$ can be
    written uniquely in the form
    $x_1x_2\cdots x_n$
    where each $x_j$ is a nontrivial element of some $A_i$
    and consecutive terms lie in
    different groups.
  \end{restatable}
  Since the free product is unique, it suffices for us to consider elements of the free product to be words of the form described in Lemma \ref{lem:normal-fp}.
  We remark here that Definition \ref{def: free-product} is a generalization of the definition presented earlier in Section \ref{sec:assertions} to an infinite number of groups.
  \iflong
    We now state a series of important lemmas. We refer the reader to a standard group theory textbook like \cite{rotman2012introduction} for a proof.
    \begin{lemma}
      \label{lem:fp-exists}
      Given a family $\{A_i : i \in I\}$ of groups, a free product exists.
    \end{lemma}
    \begin{lemma}
      \label{lem:fp-inj}
      If $P$ is a free product of $\{A_i : i \in I\}$, then the homomorphisms $j_i$ are injections.
    \end{lemma}
    \begin{lemma}
      \label{lem: faithful-inj}
      Every faithful action corresponds to an \emph{injective homomorphism} from $G$ into the symmetric group over $\Sigma_V$, denoted $Sym(\Sigma_V)$.
    \end{lemma}
  \fi
  We now define a group whose action will be the weakest pre-condition.
  \begin{definition}
    \label{def: wp-group}
    Let $\{G_i \mid i \in I\}$ be a family of subgroups of $Sym(\Sigma)$ with group action $\tuple{G_i}{a^i}{\mathit{Vars}}$ such that for every group action $\tuple{G_i}{a^i}{\mathit{Vars}}$, the triple $\triple{\tuple{G_i}{a^i}{\mathit{Vars}}}{C}{\tuple{H}{b}{\mathit{Vars}}}{\varphi^i}$ is valid for some homomorphism \(\varphi\). Then define $W_{(H,b),C}^*$ to be the free product of $\{G_i : i \in I\}$.
  \end{definition}
  Next, we define the action of $W_{(H,b),C}^*$
  \begin{definition}
    \label{def: fp-action}
    Let $x_{g_l,i}$ denote the $i^{th}$ letter of a word belonging to a group $G_l$, and let $a^l_{x_{g_l,i}}$ denote the action of this element on the set $\Sigma$.  For every word $x_{g_i,1}x_{g_j,2} \dots x_{g_k,n}$ in $W_{(H,b),C}^*$ its action is defined by $a^*_{x_{g_i,1}x_{g_j,2} \dots x_{g_k,n}} \triangleq a^i_{x_{g_i,1}} \circ a^j_{x_{g_j,2}} \circ \dots \circ a^k_{x_{g_k,n}}$. We denote the group action by $\mathsf{wp}(\tuple{H}{b}{\mathit{Vars}}, C) \triangleq \tuple{W_{(H,b),C}^*}{a^*}{\mathit{Vars}}$.
  \end{definition}
  The next lemma shows that $\mathsf{wp}(\tuple{H}{b}{\mathit{Vars}}, C)$ is a valid group action.
  \begin{restatable}{lemma}{FpActionDef}
    \label{lem: fp-action-def}
    $\mathsf{wp}(\tuple{H}{b}{\mathit{Vars}}, C)$ is a group action.
  \end{restatable}
  \iflong
    \begin{proof}
     Without loss of generality, we can consider words in their normal form, as described in Lemma \ref{lem:normal-fp}. Let $W_{(H,b),C}^*$ denote the free product group such that $\mathsf{wp}(\tuple{H}{b}{\mathit{Vars}}, C) \triangleq \tuple{W_{(H,b),C}^*}{a^*}{\mathit{Vars}}$. Recall that $W_{(H,b),C}^*$ is a product of groups $\{ G_i \mid i \in I \}$. Each group $G_i$ has an inclusion homomorphism into $Sym(\Sigma)$, since each $G_i$ is a subgroup of $Sym(\Sigma)$.
      By definition \ref{def: free-product}, the free product extends these homomorphisms uniquely from $W_{(H,b),C}^*$ to $Sym(\Sigma)$. Let $\varphi_i : G_i \rightarrow Sym(\Sigma)$ denote the homomorphism of the $i^{th}$ group. Then the homomorphism $\varphi$ from $W_{(H,b),C}^*$ to $Sym(\Sigma)$ extends such that $\varphi(x_{g_i,1}x_{g_j,2} \dots x_{g_k,n}) = \varphi^i(x_{g_i,1}) \circ \varphi^j(x_{g_j,2}) \dots \varphi^k(x_{g_k,n})$, which defines a group action on $Sym(\Sigma_V)$.  The action of any word $x_{g_i,1}x_{g_j,2} \dots x_{g_k,n}$ is defined by the action of each element, or more formally it is equal to $a^i_{x_{g_i,1}} \circ a^j_{x_{g_j,2}} \circ \dots \circ a^k_{x_{g_k,n}}$. Since $\varphi$ is a homomorphism from $W_{(H,b),C}^*$ to $Sym(\Sigma)$, by Lemma \ref{lem: action-homo}, it defines a valid group action.
    \end{proof}
  \fi
  \iflong
    The next few lemmas will help show that the group action $\mathsf{wp}(\tuple{H}{b}{\mathit{Vars}}, C)$ is a valid pre-condition.
    \begin{restatable}{lemma}{MultClosure}
      \label{lem: mult-closure}
      Let $(H,b_{\mathit{Vars}})$ be a group action, and let $h_1$ and $h_2$ be any two elements in $H$ and let $g_1$ and $g_2$ be bijections on $\Sigma$. If $\;\forall \sigma \in \Sigma$, $b_{h_1}(\sem{C}_{\sigma}) = \sem{C}_{g_1(\sigma)}$ and $b_{h_2}(\sem{C}_{\sigma}) = \sem{C}_{g_2(\sigma)}$, then
      $b_{h_1}(b_{h_2}(\sem{C}_{\sigma})) = \sem{C}_{{g_1}({g_2}(\sigma))}$.
    \end{restatable}
    \begin{proof}
      Consider $ b_{h_1}(b_{h_2}(\sem{C}_{\sigma}))$. Since $\;\forall \sigma \in \Sigma, b_{h_2}(\sem{C}_{\sigma}) = \sem{C}_{g_2(\sigma)}$, this can be re-written as $ b_{h_1}(b_{h_2}(\sem{C}_{\sigma})) =  b_{h_1}(\sem{C}_{{g_2}(\sigma)})$. Now, let ${g_2}(\sigma) = \sigma'$. Then, $b_{h_1}(\sem{C}_{{g_2}(\sigma))} = b_{h_1}(\sem{C}_{\sigma'})$. Since $\;\forall \sigma \in \Sigma$, $b_{h_1}(\sem{C}_{\sigma}) = \sem{C}_{g_1(\sigma)}$,  $\sem{C}_{g_1(\sigma')} =  \sem{C}_{{g_1}({g_2}(\sigma))}$.
    \end{proof}
    \begin{restatable}{lemma}{MultClosure2}
      \label{lem: mult-closure2}
      Let $(H,b_{\mathit{Vars}})$ be a group action, and let $h_1, \dots, h_n$ be elements in $H$ and let $g_1, \dots, g_n$ be bijections on $\Sigma$. If
      $\forall \sigma \in \Sigma_V$, $b_{h_i}(\sem{C}_{\sigma}) = \sem{C}_{g_i(\sigma)}$ for all $i$, then
      $b_{h_1} \circ b_{h_2} \circ \cdots \circ b_{h_n}(\sem{C}_{\sigma}) = \sem{C}_{g_1 \circ \dots \circ g_n(\sigma)}$.
    \end{restatable}
    \begin{proof}
      We proceed by induction on $i$. For the base case $i = 1$, this is true by assumption. Next, we consider the inductive case.
    We know that  $b_{h_1} \circ b_{h_2}, \dots \circ b_{h_{i-1}}(\sem{C}_{\sigma}) = \sem{C}_{(g_1 \circ \dots \circ g_{i-1}(\sigma))}$. Since the composition of bijections is a bijection, $g \triangleq g_1 \circ \dots \circ g_{i-1}$ is a bijection. By the group action axioms, this simplifies to $b_{h_1h_2 \dots h_{i-1}}(\sem{C}_{\sigma}) = \sem{C}_{(g(\sigma))}$ Now, we need to show
      $b_{h_1h_2 \dots h_{i-1}}((b_{h_i}(\sem{C}_{\sigma})) = \sem{C}_{(g(g_i(\sigma))}$. This follows from Lemma \ref{lem: mult-closure}
    \end{proof}
  \fi
  The next theorem proves that $\mathsf{wp}(\tuple{H}{b}{\mathit{Vars}}, C)$ is a valid pre-condition.
  \begin{restatable}{theorem}{WPSound}
    \label{thm: wp-sound}
   Let $\varphi : W_{(H,b),C}^* \rightarrow H$ be the unique homomorphism from the free group $W_{(H,b),C}^*$ into $H$. Then triple $\triple{\mathsf{wp}(\tuple{H}{b}{\mathit{Vars}}, C)}{C}{(H,b_{\mathit{Vars}})}{\varphi}$ is valid.
  \end{restatable}  
  \iflong
    \begin{proof}
      Let $W_{(H,b),C}^*$ be the group such that $\mathsf{wp}(\tuple{H}{b}{\mathit{Vars}}, C) \triangleq \tuple{W_{(H,b),C}^*}{a^*}{\mathit{Vars}}$.
      Consider any word  $x_{g_i,1}x_{g_j,2} \dots x_{g_k,n}$ in $W_{(H,b),C}^*$. Without loss of generality, we consider normalized words of the form in Lemma \ref{lem:normal-fp}. By Definition \ref{def: fp-action}, the action of this word is  $a^i_{x_{g_i,1}} \circ a^j_{x_{g_j,2}} \circ \dots \circ a^k_{x_{g_k,n}}$, where $a^i_{x_{g_i,k}}$ denotes the action of the element $x_{g_i,k} \in G_i$.
      Now, by the definition of the group $W_{(H,b),C}^*$, for any letter $x_{g_i,l}$ in this word, there exists an $h_l \in H$ such that $\forall \sigma \in \Sigma$, \[b_{h_l}(\sem{C}_{\sigma}) = \sem{C}_{a^i_{x_{g_i,l}}(\sigma)}\]
     Further, every $a^i_{x_{g_i,l}}$ induces a bijection from $\Sigma \rightarrow \Sigma$.
      Therefore, by Lemma \ref{lem: mult-closure2}, there exists $h = h_1h_2\cdots h_n$ such that
      \[
        b_{h_1h_2\cdots h_k}(\sem{C}_{\sigma}) = \sem{C}_{a^i_{x_{g_i,1}} \cdots a^k_{x_{g_k,n}}(\sigma)}
      \]
      Recall from definition \ref{def: wp-pre-faithful}, the weakest pre-condition group is the free product of only valid pre-conditions. Therefore, each group $G_i$ such that $W_{(H,b),C}^*$ is a free product of $\{G_i \mid i \in I\}$, has a homomorphism $\varphi_i : G_i \rightarrow H$, such that $\varphi_i(x_{g_i,k}) = h_k$. Then by the definition of the free product, we get a unique homomorphism $\varphi: W_{(H,b),C}^* \rightarrow H$ such that  $\varphi(x_{g_i,1}x_{g_j,2} \dots x_{g_k,n}) = \varphi_i(x_{g_i,1})\cdot  \varphi_j(x_{g_j,2}) \dots \varphi_k( x_{g_k,n}) = h_1 \cdot h_2 \cdots h_n = h$. Therefore, the triple $\triple{\mathsf{wp}(\tuple{H}{b}{\mathit{Vars}}, C)}{C}{(H,b_{\mathit{Vars}})}{\varphi}$ is valid.
    \end{proof}
  \fi
  Finally, we are ready to prove our main theorem. This theorem states that $\mathsf{wp}(\tuple{H}{b}{\mathit{Vars}}, C)$ is the weakest pre-condition.
  \begin{theorem}
    The group action $\mathsf{wp}(\tuple{H}{b}{\mathit{Vars}}, C)$ is the weakest pre-condition.
  \end{theorem}  
  \iflong
    \begin{proof}
      Let $W_{(H,b),C}^*$ be a group, and $a^*_{\mathit{Vars}}$ be its action such that $\mathsf{wp}(\tuple{H}{b}{\mathit{Vars}}, C) \triangleq \tuple{W_{(H,b),C}^*}{a^*}{\mathit{Vars}}$.
      From Theorem \ref{thm: wp-sound}, we know $\vDash \triple{\mathsf{wp}(\tuple{H}{b}{\mathit{Vars}}, C)}{C}{(H,b_{\mathit{Vars}})}{\varphi}$, where $\varphi$ is the unique homomorphism from the free group $W_{(H,b),C}^*$ into $H$. Let $\tuple{G}{a}{\mathit{Vars}}$ be any faithful group action such that $\triple{\tuple{G}{a}{\mathit{Vars}}}{C}{\tuple{H}{b}{\mathit{Vars}}}{\psi}$ is valid. We need to prove that \[\tuple{G}{a}{\mathit{Vars}} \overset{\tau}{\rightarrow} \mathsf{wp}(\tuple{H}{b}{\mathit{Vars}}, C)\]
      for some $\tau$ that we shall construct.
      From Lemma \ref{lem: faithful-inj}, there is an injective homomorphism $\eta : G \rightarrow Sym(\Sigma)$. This implies $G$ is isomorphic to a subgroup of $Sym(\Sigma)$. We denote this subgroup by $G_k$, and its action that is induced by the homomorphism $\eta$ by $(G_k, a_{k_{\mathit{Vars}}})$. Now, if $\triple{\tuple{G}{a}{\mathit{Vars}}}{C}{\tuple{H}{b}{\mathit{Vars}}}{\psi}$ is sound, then $\triple{(G_k,a_{k_{\mathit{Vars}}})}{C}{\tuple{H}{b}{\mathit{Vars}}}{\eta \circ \psi}$ is also sound. Therefore, by definition of the weakest pre-condition group, the group $G_k$ must be in the set $\{G_i \mid i
        \in I\}$ since there is a faithful action of $G_k$ which is a sound pre-condition. Since $W_{(H,b),C}^*$ is a free product of $G_i$, by Lemma \ref{lem:fp-inj} we have an injective homomorphism $j_k: G_k \rightarrow W_{(H,b),C}^*$. Now, for every group element $g \in G_k$, for all program states $\sigma \in \Sigma$, $a_{k,g}(\sigma) = a^*_{j_k(x)}(\sigma)$. 
      Therefore,
      $\tuple{G_k}{a}{\mathit{Vars}} \overset{j_k}{\rightarrow} \mathsf{wp}(\tuple{H}{b}{\mathit{Vars}}, C)$. By composing $j_i \circ \eta$, we get a homomorphism, $\tau = j_i \circ \eta$,from $(G, a_{\mathit{Vars}}) \overset{\tau}{\rightarrow} W_{(H,b),C}^*$ such that for all program states $\sigma \in \Sigma$ $a_g(\sigma) = a^*_{j_i(\varphi(g))}(\sigma)$. Therefore,
 \[\tuple{G}{a}{\mathit{Vars}} \overset{\tau}{\rightarrow} \mathsf{wp}(\tuple{H}{b}{\mathit{Vars}}, C)\]
    \end{proof}
  \fi
    \subsection{Weakest Faithful Pre-Condition}
  Recall the definition of the set 
  \( \mathsf{G^*}(\tuple{H}{b}{\mathit{Vars}}, C) \):  
\GStar*
Note that this set is not empty.
In particular, let $g : \Sigma \rightarrow \Sigma $ be the identity bijection. Then, if  $h$ is also the identity, $b_{h}(\sem{C}_{\sigma}) = \sem{C}_{g_1(\sigma)}$. Therefore, $\mathsf{G^*}(\tuple{H}{b}{\mathit{Vars}}, C)$ is not empty.
  Before we prove $\mathsf{G^*}(\tuple{H}{b}{\mathit{Vars}}, C)$ is a group and define its action, we prove a lemma states that $\mathsf{G^*}(\tuple{H}{b}{\mathit{Vars}}, C)$ is closed under inverse.
  \begin{restatable}{lemma}{InverseClosure}
    \label{lem: inverse-closure}
    Let $g$ be a bijection, and $(H,b_{\mathit{Vars}})$ be a group action such that $\exists h \in H$,
    $\forall \sigma \in \Sigma_V$, $b_{h}(\sem{C}_{\sigma}) = \sem{C}_{g(\sigma)}$. Then,
    $b_{h^{-1}}(\sem{C}_{\sigma}) = \sem{C}_{g^{-1}(\sigma)}$
  \end{restatable}
  \begin{proof}
    Consider any $\sigma \in \Sigma_V$ and let $\sigma' = g^{-1}(\sigma)$.
    We know that
    $b_{h}(\sem{C}_{\sigma'}) = \sem{C}_{g(\sigma')}$. Therefore, by taking the inverse on both sides, $\sem{C}_{\sigma'}  = b_{h^{-1}}(\sem{C}_{g(\sigma')})$. This implies, $\sem{C}_{g^{-1}(\sigma)} = b_{h^{-1}}(\sem{C}_{\sigma})$
  \end{proof}
Next, we show that $\mathsf{G^*}(\tuple{H}{b}{\mathit{Vars}}, C)$ is a group.
\WPGroup*
  \begin{proof}
    We will show four properties.
    \begin{enumerate}
      \item \textit{Closure: } Let $g_3 = g_1 \circ g_2$. We know that $\forall \sigma \in \Sigma$, $b_{h_1}(\sem{C}_{\sigma})= \sem{C}_{g_1(\sigma)}$ and $b_{h_2}(\sem{C}_{\sigma}) = \sem{C}_{{g_2}(\sigma)}$. Then by Lemma \ref{lem: mult-closure},
            $b_{h_1}(b_{h_2}(\sem{C}_{\sigma})) = \sem{C}_{{g_1}({g_2}(\sigma))}$, therefore $g_3 \in \mathsf{G^*}(\tuple{H}{b}{\mathit{Vars}}, C)$.
      \item \textit{Associativity: } Function composition is associative.
      \item \textit{Identity: } The identity function $g(\sigma) = \sigma$ is in $\mathsf{G^*}(\tuple{H}{b}{\mathit{Vars}}, C)$.
      \item \textit{Inverse: } By Lemma \ref{lem: inverse-closure}, $g^{-1}$ is also in $\mathsf{G^*}(\tuple{H}{b}{\mathit{Vars}}, C)$.
    \end{enumerate}
    Therefore, $\mathsf{G^*}(\tuple{H}{b}{\mathit{Vars}}, C)$ is a group.
  \end{proof}
We now recall the weakest faithful pre-condition $\mathsf{fwp(\tuple{H}{b}{\mathit{Vars}}, C)}$ based on the group $\mathsf{G^*}(\tuple{H}{b}{\mathit{Vars}}, C)$.
\WpTwoAction*

All that is left to prove is that $\mathsf{fwp}(\tuple{H}{b}{\mathit{Vars}}, C)$  is the
weakest faithful pre-condition.
We first show that $\mathsf{fwp}(\tuple{H}{b}{\mathit{Vars}}, C)$ is faithful. This will
help in proving that $\mathsf{fwp}(\tuple{H}{b}{\mathit{Vars}}, C)$ is a valid
pre-condition, and the weakest faithful pre-condition.
\begin{restatable}{lemma}{GSFaithful}
    \label{lem: gs-faithful}
  The action $\mathsf{fwp}(\tuple{H}{b}{\mathit{Vars}}, C)$ is faithful.
\end{restatable}
\iflong
  \begin{proof}
    Let $\mathsf{G^*}(\tuple{H}{b}{\mathit{Vars}}, C)$ be the group, and let $a_{\mathit{Vars}}$ be its action such that $\mathsf{fwp}(\tuple{H}{b}{\mathit{Vars}}, C) \triangleq (\mathsf{G^*}(\tuple{H}{b}{\mathit{Vars}}, C), a_{\mathit{Vars}})$.
    Suppose $(\mathsf{G^*}(\tuple{H}{b}{\mathit{Vars}}, C), a_{\mathit{Vars}})$ is not faithful. Then there exists two bijections $g_1$, $g_2 \in (\mathsf{G^*}(\tuple{H'}{b'}{V_1}, C))$ such that $\forall \sigma$, $g_1(\sigma) = g_2(\sigma)$, and $g_1 \neq g_2$, which is a contradiction.
  \end{proof}
\fi
Now that we have established that $\mathsf{fwp}(\tuple{H}{b}{\mathit{Vars}}, C)$ is faithful, we prove a series of lemmas that will help prove that  $\mathsf{fwp}(\tuple{H}{b}{\mathit{Vars}}, C)$ is the weakest pre-condition.. The first lemma proves that if two actions have the same transformations and one is faithful, then there is a homomorphism into the faithful group that respects the group action.
\begin{restatable}{lemma}{FaithfulHomomorphism}
  \label{lem: faithful-has-homo-into-it}
  Given two group actions $(G,a_{V_1})$ and $(H, b_{V_2})$, $V_2 \subseteq V_1$,
  where $(H, b_{V_2})$ is faithful, if $\eta: G \rightarrow H$ is a map such that
  for all group elements $g \in G$,  for all program states $\sigma$, for all
  variables  $v_i \in V_2$, we have $\overline{a}_g(\sigma)(v_i) = \overline{b}_{\phi(g)}(\sigma)(v_i)$. Then, $\eta$ is a homomorphism.
\end{restatable}
  \begin{proof}
    By Lemma \ref{lem: action-homo}, the actions $\overline{a}$, and $\overline{b}$ induce a homomorphism $\varphi : G
      \rightarrow Sym(\Sigma)$ and $\psi : H \rightarrow Sym(\Sigma)$. If $(H, b_{V_2})$ is faithful, then $(H, \overline{b}_{\mathit{Vars}})$ is also faithful. Then, by Lemma \ref{lem: faithful-inj}, $\psi$ is injective.  Now, let $f: S_G \rightarrow H$ be a map such that $f(g_i) = h_i$ such that $\forall \sigma \in \Sigma\;\forall v \in V_2$, $\varphi_{g_i}(\sigma)(v) = \psi_{h_i}(\sigma)(v)$. In the case multiple such $h_i$ exist, we can pick anyone.
    Consider some $r_j = g_1 \cdot g_2, \cdots, g_m$.
    We need to show that $f(r_j) = e_H$.  $f(r_j) = h_j$, where $h_j = h_1, \cdots, h_t$. Since $\varphi$ is a homomorphism,
    $$\varphi(r_j) = e_{Sym}$$ This means
    \[\forall \sigma \in \Sigma\;\forall v \in V_2, \varphi_{r_j}(\sigma)(v) = \psi_{h_j}(\sigma)(v)\]
    which means $\psi_{h_j} = e_{Sym}$. Since $\psi$ is injective, $h_j = e_h$, which means $f(r_j) = e_H$. By Lemma \ref{lemma-homo-respects-identity}, $f$ extends to a homomorphism, and by construction, this homomorphism is $\eta$.
  \end{proof}

The next lemma shows that $\mathsf{G^*}(\tuple{H}{b}{\mathit{Vars}}, C)$ has a homomorphism into a  post-condition group that has a faithful group action. This will help prove that $\mathsf{G^*}(\tuple{H}{b}{\mathit{Vars}}, C)$ is a valid pre-condition.
\begin{restatable}{lemma}{LemmaPreHomo}
  \label{lemma-pre-homo}
  Let $C$ be any command, and $\tuple{\mathsf{G^*}(\tuple{H}{b}{\mathit{Vars}}, C)}{a}{\mathit{Vars}}$ be a group action defined in~\Cref{def: wp2-action}.  Suppose $\tuple{H}{b}{\mathit{Vars}}$ is a faithful action on the set $\Sigma_C$. Let $\varphi: G^* \rightarrow H$ be a map such that $\varphi(g) = h$ and $b_h(\sem{C}_{\sigma}) = \sem{C}_{a_g(\sigma)}$. Then, $\varphi$ is a homomorphism.
\end{restatable}
  \begin{proof}
    We need to prove that $\varphi(g_1 \circ g_2) = \varphi(g_1) \cdot \varphi(g_2)$. Let $\varphi(g_1) \cdot \varphi(g_2) = h_1 \cdot h_2$.
    By the group action axioms, we know that \[b_{h_1 \cdot h_2}(\sem{C}_{\sigma}) = b_{h_1}(b_{h_2}(\sem{C}_\sigma))\] Since the group action is on $\mathit{Vars}$, by Lemma \ref{lem: compose-multiple}, $b_{h_1}(b_{h_2}(\sem{C}_{\sigma})) = \sem{C}_{a_{g_1}(a_{g_2}(\sigma))}$. 
    
    Now, let $\varphi(g_1 \circ g_2) = h_3$. By the definition of $\varphi$,  $b_{h_3}(\sem{C}_{\sigma}) = \sem{C}_{a_{g_1}(a_{g_2}(\sigma))}$. Since $(H,b)$ is a faithful action, any two elements with the same group action must be equal.
    This implies that \[h_3 = h_1 \cdot h_2\] Therefore, \[\varphi(g_1 \circ g_2) = \varphi(g_1) \cdot \varphi(g_2)\]
    We still need to show $\varphi(e_G) = e_H$. The element $e_G$ is the identity bijection. If $e_G$ acts on a program state, everything remains unchanged, so the action of $e_H$ relates the two output states. Therefore, $\varphi(e_g) = e_h$.
  \end{proof}
Finally, we can show that we have indeed constructed the weakest faithful
pre-condition.
\GPWFaithfulPre*
  \begin{proof}
    Let $\mathsf{fwp}(\tuple{H}{b}{\mathit{Vars}}, C) \triangleq (\mathsf{G^*}(\tuple{H'}{b'}{V}, C), a_{\mathit{Vars}})$.  We note that for every $g \in G^*(\tuple{H}{b}{\mathit{Vars}}, C)$, there exists an $h$ such that $\forall \sigma \in \Sigma$, $b_h(\sem{C}_{\sigma}) = \sem{C}_{a_g(\sigma)}$, and by Lemma \ref{lemma-pre-homo}, there is exists a homomorphism $\varphi$ such that $\varphi(g) = h$. Therefore, $\triple{\tuple{\mathsf{G^*}(\tuple{H}{b}{\mathit{Vars}}, C)}{a}{\mathit{Vars}}}{C}{\tuple{H}{b}{\mathit{Vars}}}{\varphi}$ is valid.

    Now consider any $(I,c_{\mathit{Vars}})$ such that $ \vDash \triple{(I,c_{\mathit{Vars}})}{C}{(H,b_{\mathit{Vars}})}{\tau}$. By definition of $G^*(\tuple{H}{b}{\mathit{Vars}}, C)$, for every $i \in I$, there exists $g \in G^*(\tuple{H}{b}{\mathit{Vars}}, C)$, such that for all program states $\sigma \in \Sigma$, $c_i({\sigma}) = a_g({\sigma})$.

    Finally, by Lemma \ref{lem: gs-faithful}, $\tuple{\mathsf{G^*}(\tuple{H}{b}{\mathit{Vars}}, C)}{a}{\mathit{Vars}}$ is faithful and by Lemma \ref{lem: faithful-has-homo-into-it}, there exists a homomorphism from $\eta: I \rightarrow G^*(\tuple{H}{b}{\mathit{Vars}}, C)$ such that for every $i \in I$, for all program states $\sigma \in \Sigma$, $c_i({\sigma}) = a_{\eta(i)}({\sigma})$. Therefore, $(I,c_{\mathit{Vars}}) \overset{\eta}{\rightarrow} \tuple{\mathsf{G^*}(\tuple{H}{b}{\mathit{Vars}}, C)}{a}{\mathit{Vars}}$, and $\tau = \varphi \circ \eta$ which means $\tuple{\mathsf{G^*}(\tuple{H}{b}{\mathit{Vars}}, C)}{a}{\mathit{Vars}}$ is the weakest pre-condition.
  \end{proof}
\subsubsection*{Expressibility of the Weakest Pre-Condition}
Next, we detail the missing proofs regarding the expressibility of the weakest pre-condition.
The first lemma proves that this is a valid group action.
\begin{restatable}
    {lemma}{ExpressGroup}
  $(\hat{G}((H,b),C), g^*_{\mathit{Vars}})$ is a group action.
\end{restatable}\begin{proof}
  This follows from the universal property of the free group. In particular, if we let $G$ be the generating set of $\hat{G}((H,b),C)$, we can define a function $f: G \rightarrow Sym(\Sigma)$ as the inclusion map. Then by
  the universal property, we have a homomorphism $\varphi: \hat{G}((H,b),C) \rightarrow Sym(\Sigma)$.
\end{proof}
Next, we show that indeed, this constructed group action is a valid pre-condition.
\FiniteWP*
\begin{proof}
  We know that each for each the generators of $\hat{G}((H,b),C)$ there exists and $h \in H$ such that $b_{h_i}(\sem{C}_{\sigma}) = \sem{C}_{g^*_g(\sigma)}$. By Lemma \ref{lem: mult-closure}, for any word,
  there must exist an $h_i$ such that $b_{h_i}(\sem{C}_{\sigma}) = \sem{C}_{g^*_g(\sigma)}$. We now need to show the existence of the homomorphism. This follows from Lemma \ref{lemma-pre-homo}, and uniqueness also follows from this. 
\end{proof}
\subsubsection{Strongest Post-Condition}
\PostDoesNotExist*
  \begin{proof}
    To prove this statement, we will exhibit a group $G$ with some group element $g$, a start state $\sigma_{start}$, such that there exists a variable $v \in V$ for which for any group action $\tuple{H}{b}{V}$ for every $h \in H$ the following is true:
    \[ b_h(\sem{v_1 \coloneq f(v_1, \vec{z})}_{\sigma_{start}})(v) \neq \sem{v_1 \coloneq f(v_1, \vec{z})}_{a_g(\sigma_{start})}(v)\]

    Let $G$ be the symmetric group $Sym({\Sigma_V})$. This is the group of all bijections from $\Sigma_V \rightarrow \Sigma_V$, and every element can be viewed as a permutation from one state to another.  Let $V$ be the set $\{v_1, \dots, v_n\}$
    Let $a$ be the group action of the symmetric group $Sym(\Sigma_V)$ such that for any group element, $g'$, $a_{g'}$ acts on a given $\sigma$ as $g'(\sigma)$. Let $g$ be an element of $S_{\Sigma_V}$ that
    is a bijection defined as $g(\{v_1 \mapsto z_1, v_2 \mapsto z_2,\dots ,v_n \mapsto z_n \}) = \{v_1 \mapsto z_n, v_2 \mapsto z_1,\dots ,v_n \mapsto z_{n-1} \} $. In other words, $g$ is the bijection that cyclically rotates a given $\sigma$.

    Since $f$ is not injective, we have that there exists $x_1 \neq y_1, \Vec{z}$ such that $f(x_1, \Vec{z}) = f(y_1, \Vec{z})$.
    Let $\sigma_{start} = \{v_1 \mapsto x_1, v_2 \mapsto z_1, \dots , v_n \mapsto z_n\}$ and
    let $\sigma_{start}' = \{v_1 \mapsto y_1, v_2 \mapsto z_1, \dots , v_n \mapsto z_n\}$.

    We now claim that for any group action $(H,b_V)$ no action $b_h$ can satisfy both,
    \begin{equation}
      b_h(\sem{v_1 \coloneq f(v_1, \vec{z})}_{\sigma_{start}})(v_2) = \sem{v_1 \coloneq f(v_1, \vec{z})}_{a_g(\sigma_{start})}(v_2)
    \end{equation}
    and
    \begin{equation}
      b_h(\sem{v_1 \coloneq f(v_1, \vec{z})}_{\sigma_{start}'})(v_2) = \sem{v_1 \coloneq f(v_1, \vec{z})}_{a_g(\sigma_{start}')}(v_2)
    \end{equation}
    In the following, the subscript $g$ denotes the value of a variable after $a_g$ has acted on the program state.
    If any group $H$ with group action $b$ exists, then, for $b_h$ to satisfy (2), it must satisfy the following:
    \begin{align*}
      b_h(\{v_1 \mapsto f(x_1, \Vec{z}), v_2 \mapsto z_2, \dots , v_n \mapsto z_n\})(v_2) & = \{v_1 \mapsto f(x_{1g}, \Vec{z}_g),  v_2 \mapsto z_{2g},\dots ,, v_n \mapsto z_{ng}\}(v_2) \\
                                                                                          & =  \{v_1 \mapsto f(z_n, \Vec{z}_g),  v_2 \mapsto x_1,\dots ,, v_n \mapsto z_{n-1}\}(v_2)     \\
                                                                                          & = x_1
    \end{align*}
    At the same time, it must also satisfy
    \begin{align*}
      b_h(\{ v_1 \mapsto f(y_1, \Vec{z}), v_2 \mapsto z_2,\dots , v_n \mapsto z_n\})(v_2) & = \{v_1 \mapsto f(y_{1g}, \Vec{z}_g),  v_2 \mapsto z_{2g}, \dots , v_{n} \mapsto z_{2ng}\}(v_2) \\
                                                                                          & = \{v_1 \mapsto f(z_n, \Vec{z}_g),  v_2 \mapsto y_1, \dots , v_{n} \mapsto z_{n-1}\}(v_2)
      \\ &= y_1
    \end{align*}
    Both the left-hand sides are the same since $f$ is not injective, but under $a_g$ as defined, the two right-hand sides are different. So, any $b_h$ can map to only of those. If it maps to the first one, then there exists no $b_h$ for which (3) is true, i.e., for $\sigma_{start}'$
    \begin{equation*}
      b_h(\sem{v_1 \coloneq f(v_1, \vec{z})}_{\sigma_{start}'})(v_2) \neq \sem{v_1 \coloneq f(v_1, \vec{z})}_{a_g(\sigma_{start}')}(v_2)
    \end{equation*} Otherwise, we get a symmetric statement for $\sigma_{start}$. Either way, we have shown that there exists a group $G$ with group action $a_g$ such that there is no group action that can be a valid post-condition for the set $V$. Since this also holds for the lifting of $Sym(\Sigma_V, a_V)$, we get the theorem.
  \end{proof}
\PostIsSp*
  \begin{proof}
    Let $(G,a^*_{\mathit{Vars}})$ denote the action of $\mathsf{POST}((G,a_V), f)$.
    Let $\tuple{I}{c}{V_1}$ be a group action such that $\vDash \triple{\tuple{G}{a}{V}}{v_1 \coloneq f(v_1, \dots, v_n)}{\tuple{I}{c}{V_1}}{\eta}$. Then
    \[
      \forall g \in G, \exists i \in I\; \forall \sigma\; \forall v \in V_1 \;a^*_g(\sigma)(v) = \overline{c_i}(\sigma)(v)
    \]
    Since $(I,c_{V_1})$ is faithful, by Lemma \ref{lem: faithful-has-homo-into-it}, there exists a homomorphism $\tau$ such that
    \[
      \forall g \in G, \exists i \in I\; \forall \sigma \;\forall v \in V_1 a^*(\sigma)(v) = \overline{c_{\tau(g)}}(\sigma)(v)
    \]
    Since $\mathsf{POST}((G,a_V), f)$ acts on $\mathit{Vars}$, $V_1 \subseteq \mathit{Vars}$. Therefore,
    $(\mathsf{POST}((G,a_V), f)) \overset{\tau}{\rightarrow} \tuple{I}{C}{V_1}$.  Since
    $\mathsf{POST}((G,a_V), f)$ is also a sound pre-condition with homomorphism $\varphi$ such that $\eta = \tau \circ \varphi$, we have shown it is the strongest post-condition.
  \end{proof}
\section{Examples}
\subsection{Translation of a Car}
We continue with the proof of the translation of a car example.
Lines 5 and 6 are injective, and a similar computation can be carried out to conclude:
\[\triple{\tuple{\mathbb{Z}}{a}{\mathit{Vars}}}{\phi \coloneq u \cdot dt + \phi}{\tuple{\mathbb{Z}}{a}{\mathit{Vars}}}{\mathbf{eq}}\]
and 
\[\triple{\tuple{\mathbb{Z}}{a}{\mathit{Vars}}}{\theta \coloneq \frac{v}{L}\tan{\theta} \cdot dt + \theta}{\tuple{\mathbb{Z}}{a}{\mathit{Vars}}}{\mathbf{eq}}\]
Therefore, using \textsc{SEQ}, we can prove that each statement preserves the group action $\tuple{\mathbb{Z}}{a}{\mathit{Vars}}$. 
\subsection{Lorenz Attractor System}
    \label{app: Examples}
The Lorenz Attractor system is a simplified set of differential equations that
model fluid convections used widely in weather prediction models. The
quantities that characterize the fluid are: $x$ which is the rate of convective
motion, $y$ which is the temperature difference between the ascending and descending currents, and
$z$ which represents the distortion of the vertical temperature profile.
The equations governing $x,y$, and $z$ are:
\begin{equation*}
  \dot x = -px +py, \dot y = -xz + rx -y , \dot z = xy - bz
\end{equation*}
where $p,r$, and $b$ are parameters that depend on conditions like the fluid, the heat input, the size of the container, etc, but they are assumed constant throughout. The notation $\dot x$ means the change of $x$ with respect to time. Each equation describes the rate of change of the quantities $x,y,z$.  As in the car example, we model this system as a program that runs for $T$ time steps, where $T$ is any constant, as shown in Figure \ref{fig: lorentz}. The variable $dt$ represents a tiny constant. In this example, we enrich our programming model with more types. In particular, we assume the variables $p,r,b, dt$, and $T$ are of type $\mathbb{R}_1$, which means they are real numbers between 0 and 1. As a result, the program store is extended to map to these types.

\begin{figure}
  \begin{lstlisting}[frame=tlrb, mathescape=true]
1.  $\while$(t $\coloneq$ 0; t < T; t $\coloneq$ t +  dt){
2.    x $\coloneq$ (-px + py)$\cdot$dt+ x;
3.    y $\coloneq$  (-xz + rx - y)$\cdot$dt + y;
4.    z $\coloneq$  (xy - bz)$\cdot$dt + z;
6.  }
\end{lstlisting}
  \caption{A program representing a Lorenz Attractor System}
  \label{fig: lorentz}
\end{figure}
The Lorenz Attractor System is equivariant to swapping signs of $x$ and $y$. This means that the output states are related by the same action as the input, and the homomorphism from the pre-condition group to the post-condition group is the identity map.
\citet{10.1007/978-3-030-31784-3-6} exploit this equivariance to verify a safety property of a hybrid system based on the Lorenz Attractor. We seek to show this formally. Once again, we start by defining the group action that models our transformation.
Flipping the signs of $x$ and $y$ can be expressed as the following action of
the group $S_2$:
\[ \actiontemp{g}{g^2 = e}{\{x, y\}}{g \cdot (x \rightarrow \alpha_1, y \rightarrow \alpha_2) = (x \rightarrow -\alpha_1, y \rightarrow -\alpha_2)}\]
This is a group action since for any $\sigma \in \Sigma_{\{x,y\}}$, we have $g^2(\sigma) = \sigma$.
The triple we would like to prove is:
\[\triple{\tuple{S_2}{a}{\mathit{Vars}}}{\textbf{for } {(t \coloneq 0; t < T; t \coloneq t + dt)} \textbf{ do } C}{\tuple{S_2}{a}{\mathit{Vars}}}{\varphi}\]
where $\varphi(g) = g$.
We note that every assignment statement in this program is injective.  
Therefore,
we can apply the $\textsc{ASSGN}$ rule at each step, and apply $\textsc{CONS-1}$ to conclude the desired triple.
\paragraph{Line 2}
We begin by re-writing the first statement in the loop body $x \coloneq (-px + py)\cdot dt+ x$ (Line 2) as $x \coloneq f(x,y,p,dt)$. The function $f$ is injective with respect to $x$, i.e., $\forall x_1, x_2, y \in \mathbb{Z}, \forall dt, p\in \mathbb{R}_1$, $f(x_1,y,p, dt) = f(x_2,y,p, dt)$ implies $ x_1 = x_2$.  Therefore, we can invert this assignment statement. We define
  \[
    \hat{f}(k, p,y,dt) = 
    \begin{cases}
    \begin{aligned}
       & \frac{py\cdot dt - k}{p\cdot dt - 1} \quad &p\cdot dt \neq 1 \\
       & \bot   \quad  &p\cdot dt = 1
           \end{aligned}
    \end{cases}
  \]
  By assumption, $p$ and $t$ are always less than 1. Therefore, $pt \neq 1$.
  We can verify that
  $\hat{f}(f(x,y,p,dt), y, p, dt) = x$. Next, we compute $\mathsf{POST}((S_2,a_{x,y}), f) $
  \begin{align*}
    \mathsf{POST}((S_2,a_{x,y}), f) \triangleq \actiontemp{g}{g^2= 1} {\mathit{Vars}} { & g \cdot (x \rightarrow \alpha_f, y \rightarrow \alpha_y, \dots, z \rightarrow \alpha_z) =
    \\
                                                                               & x \rightarrow f(-\hat{f}(\alpha_f, \alpha_y, \alpha_p, dt ),-\alpha_y,\alpha_p, dt) ,     \\ &  y \rightarrow -\alpha_y,
      \dots, z \rightarrow \alpha_z}
  \end{align*}

$\mathsf{POST}((S_2,a_V), f)_{\mathit{Vars}})$ simplifies to:
\begin{align*}
  \actiontemp{g}{g^2= 1} {\mathit{Vars}} {g \cdot ( & x \rightarrow \alpha_f, y \rightarrow \alpha_y, \dots, z \rightarrow \alpha_z) =
  \\
                                           & x \rightarrow -\alpha_f, \rightarrow -\alpha_y,  \dots,  z\rightarrow \alpha_z}
\end{align*}
By the $\textsc{ASSGN}$ rule, we have $$ \triple{\tuple{S_2}{a}{\{x,y\}}}{x
  \coloneq (-px + py) \cdot dt+ x}{(\mathsf{POST}((S_2,a_V), f)_{\mathit{Vars}}}{\varphi}$$.
We observe from the simplified version that $\mathsf{POST}((S_2, a_{x,y}), f)$
acts by changing the sign of $x$ and $y$ and leaves every other variable
unchanged, which is the same action as $(S_2, a_{\mathit{Vars}})$.
Therefore, $\mathsf{POST}((S_2,a_V), f)_{\mathit{Vars}} \rightarrow (S_2, a_{\mathit{Vars}})$. By
$\textsc{CONS-1}$, we have
\[\triple{\tuple{S_2}{a}{\{x,y\}}}{x \coloneq (-px + py) \cdot dt+ x}{(S_2, a_{\mathit{Vars}})}{\varphi}\]
\paragraph{Lines 3 and 4}
A similar computation can be carried out for Lines 3 and 4 in Figure \ref{fig: lorentz}, each of which is injective, and equivariant to $(S_2, a_{\{x,y\}})$, i.e.,
\[ \triple{\tuple{S_2}{a}{\{x,y\}}}{c_3}{(S_2, a_{\mathit{Vars}})}{\varphi} \text{ and } \triple{\tuple{S_2}{a}{\{x,y\}}}{c_4}{(S_2, a_{\mathit{Vars}})}{\varphi}\]
Let $C$ denote the loop body. By the $\textsc{SEQ}$ and $\textsc{LIFT}$ rules,
we have $\triple{(S_2, a_{\mathit{Vars}})}{C}{(S_2, a_{\mathit{Vars}})}{\varphi}$.
Therefore, the loop body (lines 2-4), denoted by $C$ preserves the action $(S_2,
  a_{\mathit{Vars}})$ and the variables modified by $C$ are always a subset of $\mathit{Vars}$.
Using the $\textsc{FOR}$ rule  we can conclude
\[ \triple{\tuple{S_2}{a}{\mathit{Vars}}}{\textbf{for } {(t \coloneq 0; t < T; t
      \coloneq t + dt)} \textbf{ do } C}{\tuple{S_2}{a}{\mathit{Vars}}}{\varphi}  \]
which is the desired triple.

\subsection{Car in a Straight Line}
We prove that the loop body is preserved for the program in Figure \ref{fig: car-straight}, showing the steps missing in the proof. 

\paragraph{Lines 1 and 2} Since none of the variables in these two lines is changed, it is straightforward to see that the $\textsc{SEM-ASSGN}$ and $\textsc{SEQ}$ rule yields:
\[
\triple{\tuple{D_4}{a}{\mathit{Vars}}}{c_1; c_2}{\tuple{D_4}{a}{\mathit{Vars}}}{\mathbf{eq}}
\]
\paragraph{Line 4}
We start with the first line of the loop body $c_4 \triangleq x \coloneq v \cdot \cos{\theta}\cdot dt + x;$. This assignment statement is injective, so we can use the $\textsc{ASSGN}$ rule.
The group action $\mathsf{POST}$, is equal to the following action:
\begin{align*}
  \tuple{D_4}{b}{\{x,y,\theta\}} \triangleq \actiontemp{r,s}{ & r^4 = e, s^2 = e, rs = sr^{-1}}{\{x,y,\theta\}}{
  r\cdot(x \rightarrow \alpha_x, y \rightarrow \alpha_y, \theta \rightarrow  \alpha_\theta) =                                                                 \\&
  \{x \rightarrow -\alpha_v \sin({\alpha_\theta})\cdot dt - \alpha_y, y \rightarrow \alpha_x - \cos{\theta}, \theta \rightarrow \alpha_\theta + 90 \mod 360\} \\
                                                              & s\cdot(x \rightarrow \alpha_x, y \rightarrow \alpha_y, \theta \rightarrow \alpha_\theta) =
    x \rightarrow \alpha_x, y \rightarrow -\alpha_y, \theta \rightarrow -\alpha_\theta) }
\end{align*}
Therefore, we can conclude,
\[\triple{\tuple{D_4}{a}{\{x,y,\theta\}}}{c_4}{\tuple{D_4}{b}{\mathit{Vars}}}{\mathbf{eq}}\] is valid.

Next, we turn to the Line 5 of the program, $c_5 \triangleq y \coloneq v \sin{(\theta)}\cdot dt + y$.
By a similar line of reasoning, and then using the $\textsc{ASSGN}$ rule, we can compute $\mathsf{POST}$. Using $\textsc{CONS-2}$, we derive $\triple{\tuple{D_4}{b}{\mathit{Vars}}}{c_5}{\tuple{D_4}{a}{\mathit{Vars}}}{\mathbf{eq}}$.

Lines 6 and 7 are also injective assignments. Therefore, by applying $\textsc{ASSGN}$ and $\textsc{CONS-1}$, and then composing the triples with the sequencing rules gives us: 
\[ \triple{\tuple{D_4}{a}{\{x,y,\theta\}}}{v
    \coloneq a \cdot dt + v; t \coloneq t + dt}{\tuple{D_4}{a}{\mathit{Vars}}}{\mathbf{eq}}\]
Finally, we use $\textsc{SEM-ASSGN}$ for Line 8, to conclude: 
\[ \triple{\tuple{D_4}{a}{\{x,y,\theta\}}}{c_8}{\tuple{D_4}{a}{\mathit{Vars}}}{\mathbf{eq}}\]

\subsection{Voting Example}
We resume the proof from the main body of the paper. Intuitively, the number of votes for Candidate-1, and Candidate-2 should remain unchanged after the permutation. Therefore, we will split the program into 2 parts. We will first prove that after Lines 1 and 2, the value of $count_1$ is invariant, and then show that after Lines 3 and 4, $count_2$ is invariant.
    We will use $\varphi : S_2 \rightarrow S_2$to denote the homomorphism $\varphi(g) = g$.
We start with the assignment on Line 1, $c_1 \triangleq count_1$ $\coloneq v_1 == 0 \;?
  \;count_1$ + 1 : $count_1$.
This is an injective assignment with respect to $count_1$.
Therefore, we can compute $\mathsf{POST}((S_2,a_{\{v_1, v_2\}}), f)$. To start, we define $\hat{f}: \mathbb{Z} \times \mathbb{Z} \rightarrow \mathbb{Z}$ that inverts the assignment.
\[
  \hat{f}(\alpha_{count_1}, \alpha_{v_1}) =
  \begin{cases}
    \alpha_{count_1} - 1 & \alpha_{v_1} = 0    \\
    \alpha_{count_1}     & \alpha_{v_1} \neq 0 \\
  \end{cases}
\]
We can easily check that $\hat{f}(f( \alpha_{count_1}, \alpha_{v_1}), \alpha_{v_1}) =  \alpha_{count_1}$. Since the action $\tuple{S_2}{a}{\{x,y\}}$, swaps the values of $v_1$ and $v_2$, we need to compute
\[
  f(\hat{f}(f( \alpha_{count_1}, \alpha_{v_1}), \alpha_{v_1}), \alpha_{v_2})
\]

Combining both the cases in $f$ and $\hat{f}$, we get a  piece-wise function $b: \mathbb{Z} \times \mathbb{Z} \times \mathbb{Z} \rightarrow \mathbb{Z}$ such that
\[
  b(\alpha_{count_1}, \alpha_{v_1}, \alpha_{v_2}) = f(\hat{f}(f( \alpha_{count_1}, \alpha_{v_1}), \alpha_{v_1}), \alpha_{v_2}) \triangleq
  \begin{cases}
    \alpha_{count_1}    & \alpha_{v_1} = 0 \land \alpha_{v_2} = 0       \\
    \alpha_{count_1} -1 & \alpha_{v_1} = 0 \land \alpha_{v_2} \neq 0    \\
    \alpha_{count_1}    & \alpha_{v_1} \neq 0 \land \alpha_{v_2} \neq 0 \\
    \alpha_{count_1}+ 1 & \alpha_{v_1} \neq 0 \land \alpha_{v_2} = 0
  \end{cases}
\]
With this, we have all the ingredients to compute $\mathsf{POST}((S_2,a_{\{v_1, v_2\}}), f)$ that acts on $count_1$ by the function $b$, permutes $v_1$ and $v_2$, and leaves $count_2$ unchanged
\begin{align*}
  (S_2, b_{\mathit{Vars}}) \triangleq & \mathsf{POST}((S_2,a_{\{v_1, v_2\}}), f) \triangleq
  \actiontemp{x}{x^2 = e}{\mathit{Vars}}{                                                                              \\&x \cdot ({count}_1 \rightarrow
    \alpha_{count_1}, {count}_2 \rightarrow
  \alpha_{count_2}, v_1 \rightarrow \alpha_1, v_2 \rightarrow \alpha_2, winner \rightarrow \alpha_{winner}) = \\& ({count}_1\rightarrow
    b(\alpha_{count_1}, \alpha_1, \alpha_2) , {count}_2 \rightarrow
    \alpha_{count_2}, v_1 \rightarrow \alpha_2, v_2 \rightarrow \alpha_1, winner \rightarrow \alpha_{winner})}
\end{align*}
By the $\textsc{ASSGN}$ rule,
$\triple{\tuple{S_2}{a}{\{v_1,v_2\}}}{c_1}{(\mathsf{POST}((S_2,a_{\{v_1, v_2\}}), f)_{\mathit{Vars}})}{\varphi}$ is valid. As guaranteed by Theorem \ref{thm:post-is-grp-action}, $(S_2, b_{\mathit{Vars}})$ is a valid group action.

We move on to the next line of the program denoted by $c_2 \triangleq
  \text{count}_1 \coloneq f_2({count}_1, v_2)$. This line is also an injective
assignment. We can compute $\mathsf{POST}((S_2, b_{\mathit{Vars}}), f_2)$. We elide the computation here, but the action simplifies to $(S_2,a_{\mathit{Vars}})$. Therefore, using $\textsc{SEQ}$, we conclude the triple \[\triple{(S_2,a_{\mathit{Vars}})}{c_1;c_2}{(S_2,a_{\mathit{Vars}})}{\varphi}\]
This means that after permuting $v_1$ and $v_2$, $count_1$ remains unchanged.

We need to prove a similar triple for Lines 3 and 4. The proofs of these are analogous to Lines 1-2, since they are the same expression with variables renamed. Therefore, by applying the same reasoning techniques, we can conclude
\[ \triple{(S_2,a_{\mathit{Vars}})}{c_3;c_4}{(S_2,a_{\mathit{Vars}})}{\varphi}\]
    \subsection{Gravitational Simulation}
\begingroup
\setlength{\intextsep}{0pt}
\setlength{\columnsep}{0pt}
\begin{wrapfigure}{r}{0.54\textwidth}
  \vspace{-7.5mm}
  \begin{center}
    \begin{tabular}{c}
    \begin{lstlisting}
1. $\while$(t $\coloneq$ 0; t < T; t $\coloneq$ t + dt){
2. $\quad \quad F_1 \coloneq \frac{(G \cdot m_1 \cdot m_2)(x_1 - x_2)}{(|x_1 - x_2|^3)};$
3. $\quad \quad v_1 \coloneq v_1 + F_1/m_1 \cdot dt;$
4. $\quad \quad v_2 \coloneq v_2 - F_1/m_2 \cdot dt;$
5. $\quad \quad x_1 \coloneq x_1 + v_1 \cdot dt;$
6. $\quad \quad x_2 \coloneq x_2 + v_2 \cdot dt;$
7. }
     \end{lstlisting}
    \end{tabular}
  \end{center}
  \vspace{-6.0mm}
\caption{Model of a Simple Gravitational Simulation} \label{fig:gravity}
\end{wrapfigure}
Objects experience attractive gravitational forces due to their mass which leads to acceleration. This can lead to complicated trajectories, which can be hard to compute analytically. Here we consider a simplified simulation of this model, where we have two objects at initial positions $x_1$, and $x_2$ at rest. 

This system of two objects is modeled in Figure \ref{fig:gravity}. On Line 2, we compute the gravitational forces of attraction on Object 1. The gravitational force depends on masses $m_1$, $m_2$, the absolute value of the distance between objects, and the gravitational constant, $G$.
The rest of the program uses the force to update the velocity and then the position of both objects. To update the velocity, we first compute acceleration using the equation $a = F/m$. On line 4, when we update the velocity of Object 2, we have to be careful to negate the force, $F_1$, since it is the force acting on Object 1.
In a nutshell, this program simulates motion due to the gravitational force of attraction in one dimension (along the x-axis). 

\paragraph*{Symmetry Property. } A property that this system exhibits is that if we swap the initial positions, masses and velocities of Object 1 and Object 2, then their final positions are also swapped. We can express swapping initial positions and masses as the action of the group $S_2$ like so: 
\begin{align*}
    (S_2, a_{\{x_1, m_1, x_2, m_2,v_1, v_2\}}) \triangleq&
\actiontemp{s}{s^2 = 1}{\{x_1, m_1, x_2, m_2, v_1,v_2\}}{\\&s \cdot (x_1 \rightarrow \alpha_1, m_1 \rightarrow \beta_1, x_2 \rightarrow \alpha_2, m_2 \rightarrow \beta_2, v_1 \rightarrow \gamma_1, v_2 \rightarrow \gamma_2) \\&= (x_1 \rightarrow \alpha_2, m_1 \rightarrow \beta_2, x_2 \rightarrow \alpha_1, m_2 \rightarrow \beta_1, v_1 \rightarrow \gamma_2, v_2 \rightarrow \gamma_1)}
\end{align*}
The group $S_2$ acts on the program state by swapping the values of $x_1$ and $x_2$,  $m_1$ and $m_2$ and $v_1$ and $v_2$. Unfortunately, this precondition is not strong enough to establish a loop invariant. Instead, we consider the following action: 
\begin{align*}
    (S_2, a_{\{x_1, m_1, x_2, m_2,v_1, v_2, F_1\}}) \triangleq&
 \actiontemp{s}{s^2 =1}{\{x_1, m_1, x_2, m_2, v_1,v_2,F_1\}}{\\&s \cdot (x_1 \rightarrow \alpha_1, m_1 \rightarrow \beta_1, x_2 \rightarrow \alpha_2, m_2 \rightarrow \beta_2, v_1 \rightarrow \gamma_1, v_2 \rightarrow \gamma_2, F_1 \rightarrow \delta_1) \\&= (x_1 \rightarrow \alpha_2, m_1 \rightarrow \beta_2, x_2 \rightarrow \alpha_1, m_2 \rightarrow \beta_1, v_1 \rightarrow \gamma_2, v_2 \rightarrow \gamma_1, F_1 \rightarrow -\delta_1)}
\end{align*}
This action is virtually the same as before, with the key difference being that we also flip the sign of $F_1$.
If we let $C$ denote the program in Figure \ref{fig:gravity}, then the desired triple is:
\[
\triple{(S_2, a_{\{x_1, m_1, x_2, m_2,v_1, v_2, F_1\}})}{C}{(S_2, a_{\{x_1, m_1, x_2, m_2,v_1, v_2, F_1\}})}{\mathbf{eq}}
\]
\subsubsection*{Line 2}Line 2 computes the force $F_1$ that acts on object 1. We let $c_2$ denote this statement. Unfortunately, this is not an injective assignment statement, so we have to resort to using the $\textsc{SEM-ASSGN}$ rule. By examining the assignment, we can see that swapping $x_1$, and $x_2$ flips the sign of $F$. This is exactly the group action $(S_2, a_{\{x_1, m_1, x_2, m_2,v_1, v_2, F_1\}})$. Using $\textsc{SEM-ASSGN}$, and $\textsc{LIFT}$, we can conclude  \[\triple{(S_2, \overline{a}_{\mathit{Vars}})}{c_2}{(S_2, a_{\{x_1, m_1, x_2, m_2,v_1, v_2, F_1\}}) }{\mathbf{eq}}\]
Eventually, we would like our post-condition to be about the set $\mathit{Vars}$, so we can apply the sequencing rule. In principle, we could use the $\textsc{SEM-ASSGN}$ rule to derive this too. However, since the semantic assignment rule involves checking semantic side conditions, rather than reasoning about the entire program state, we will use the structural rules to reason about the rest of the program state. 

We define the action $(E,e_{\mathit{Vars}})$ that acts like the identity on every variable. By $\textsc{ID}$,
\[\triple{\tuple{ {S_2}}{\overline{a}}{\mathit{Vars}}}{c_2}{\tuple{E}{e}{\mathit{Vars} \setminus \{x_1, m_1, x_2, m_2,v_1, v_2, F_1\}}}{\mathbf{e*}}\]
Now by the $\textsc{DIR-PROD}$ rule:
\begin{mathpar}
    \inferrule*[left={}, right={}]{\triple{\tuple{ {S_2}}{\overline{a}}{\mathit{Vars}}}{c_2}{\tuple{ {S_2}}{a}{\{x_1, m_1, x_2, m_2,v_1, v_2, F_1\}}}{\mathbf{eq}}\;\;\;\triple{\tuple{{S_2}}{\overline{a}}{\mathit{Vars}}}{ c_2}{\tuple{E}{e}{\mathit{Vars} \setminus \{\{x_1, m_1, x_2, m_2,v_1, v_2, F_1\}\}}}{\mathbf{e*}}}{\triple{\tuple{{S_2} 
    }{\overline{a}}{\mathit{Vars}}}{c_2}{\tuple{ {S_2} \times E
    }{a \times e}{\mathit{Vars}}}{\mathbf{eq} \times \mathbf{e*}}}
\end{mathpar}
Since $\tuple{{S_2} \times E
    }{a \times e}{\mathit{Vars}} \overset{\mathbf{eq}}{\rightarrow} \tuple{{S_2}}{\overline{a}}{\mathit{Vars}}$, by $\textsc{CONS-1}$ we can conclude, $\triple{\tuple{{S_2} 
    }{a}{\mathit{Vars}}}{c_2}{\tuple{{S_2}}{\overline{a}}{\mathit{Vars}}}{\mathbf{eq} \circ \mathbf{eq}}$. Since, $\mathbf{eq} \circ \mathbf{eq} = \mathbf{eq}$, so we get $\triple{\tuple{{S_2} 
    }{a}{\mathit{Vars}}}{c_2}{\tuple{{S_2}}{\overline{a}}{\mathit{Vars}}}{\mathbf{eq}}$.

\subsubsection*{Lines 3 and 4}
Lines 3 and 4 update the values of $v_1$ and $v_2$, respectively. It is intuitive that after updating both velocities, their values remain swapped. However, we would like to find an intermediate assertion after $c_3$. Since lines 3 and 4 have injective assignments, we can use the $\mathsf{ASSGN}$ rule, and compute $\mathsf{POST((S_2,a),c_3)}$. We compute $\mathsf{POST((S_2,a),c_3)}$ as follows. We only show the action on the set $x_1, m_1, x_2, m_2, v_1, v_2,F$, since $ \mathsf{POST((S_2,a),c_3)} $ acts like the identity for all other variables. 
\begin{align*}
    \mathsf{POST((S_2,a),c_3)} &\triangleq \actiontemp{s}{s^2=1}{x_1, m_1, x_2, m_2, v_1, v_2,F}{
    \\&s \cdot (x_1 \rightarrow \alpha_1, m_1 \rightarrow \beta_1, x_2 \rightarrow \alpha_2, m_2 \rightarrow \beta_2, v_1 \rightarrow \gamma_1, v_2 \rightarrow \gamma_2, F_1 \rightarrow \delta_1) \\&= (x_1 \rightarrow \alpha_2, m_1 \rightarrow \beta_2, x_2 \rightarrow \alpha_1, m_2 \rightarrow \beta_1,\\& v_1 \rightarrow -dt\cdot \frac{\delta_1}{\beta_2} + \gamma_2,
    v_2 \rightarrow -dt\cdot \frac{\delta_1}{\beta_1} + \gamma_1, F_1 \rightarrow -\delta_1)
    }
\end{align*}
This action swaps $x_1$ and $x_2$, and $m_1$ and $m_2$. For $v_1$ and $v_2$, it inverts the assignment, swaps the values, and computes the assignment again.  

Using $\mathsf{ASSGN}$, and $\textsc{CONS-1}$ we can also prove $\triple{\mathsf{POST((S_2,a),c_3)}}{c_4}{\tuple{{S_2}}{\overline{a}}{\mathit{Vars}}}{\mathbf{eq}}$. Using $\textsc{SEQ}$, we get $\triple{\tuple{{S_2}}{\overline{a}}{\mathit{Vars}}}{c_3;c_4}{\tuple{{S_2}}{\overline{a}}{\mathit{Vars}}}{\mathbf{eq}}$. 

\subsubsection*{Lines 5 and 6} Lines 5 and 6 are analogous to Lines 3 and 4. Both the assignments are injective, so we can use the $\mathsf{ASSGN}$ rule. After computing $\mathsf{POST}$, and using $\textsc{CONS}$, we get $\triple{\tuple{{S_2}}{\overline{a}}{\mathit{Vars}}}{c_5;c_6}{\tuple{{S_2}}{\overline{a}}{\mathit{Vars}}}{\mathbf{eq}}$. 

If we let $C_l$ denote the loop body, then we have proved, 
\[\triple{\tuple{{S_2}}{\overline{a}}{\mathit{Vars}}}{C_l}{\tuple{{S_2}}{\overline{a}}{\mathit{Vars}}}{\mathbf{eq}}\]
Using the $\textsc{FOR}$ rule gives us the desired triple. 

\subsection{Car Translation 2}
We return to the motion of a simple car system from Figure \ref{fig: car-trans}. This time we consider another symmetry property, where if the $x$ and $y$ coordinates are transformed by different amounts, the motion of the car is still equivariant to this translation--if the starting coordinates $(x,y)$ are shifted by some constants $(c_1,c_2)$,
the final coordinates $(x',y')$ will also be shifted by $(c_1,c_2)$ to $(x' + c_1, y' + c_2)$. Once again, the homomorphism is $\mathbf{eq}(g) = g$. 
We model this using the direct product construction from our logic. Formally, this can be expressed in our syntax as follows: 
\begin{align*}
   (\mathbb{Z} \times \mathbb{Z}_{\{x,y\}}) \triangleq \actiontemp{r,s}{rs = sr }{x, y}{\;\;&r \cdot (x \rightarrow \alpha_1, y \rightarrow \alpha_2) =x \rightarrow \alpha_1 + 1, y \rightarrow \alpha_2 \\&
   s \cdot  (x \rightarrow \alpha_1, y \rightarrow \alpha_2) =x \rightarrow \alpha_1, y \rightarrow \alpha_2 + 1}
\end{align*}
The generator $r$ acts by adding 1 to $x$'s value, and the generator $s$ acts by incrementing the value of $y$ by 1. This construction can be derived from the actions $(\mathbb{Z}, a_{\{x\}})$ and $(\mathbb{Z}, a_{\{y\}})$ using Definition \ref{def: dir-prod-action}, and by Lemma \ref{lem:dir-prod-cond}, this is a group action. If we let $C$ denote the program from Figure \ref{fig: car-trans}, then the desired triple is 
\[
 \triple{(\mathbb{Z} \times \mathbb{Z}_{\{x,y\}})}{C}{  (\mathbb{Z} \times \mathbb{Z}_{\{x,y\}})}{\mathbf{eq}}
\]
The proof of this triple is similar to the earlier car translation example. 
    \subsection{ABC Flows}
\label{app: abc}
We now consider a turbulent flow called the ABC Flow taken from \citet{McLachlan_Quispel_2002}. Figure \ref{fig: ABC} describes this system. This system preserves multiple different actions of the group $S_2$. We describe two such symmetries here. 
Consider the action:
\[ \actiontemp{g}{g^2 = e}{\{x, y, z\}}{g \cdot (x \rightarrow \alpha_1, y \rightarrow \alpha_2, z \rightarrow \alpha_3) = (x \rightarrow \alpha_1, y \rightarrow \pi -\alpha_2, z \rightarrow -z)}\]
If we let $C$ denote the program in \Cref{fig: AAC}, then  our desired triple is: 
\(
\triple{(S_2,\overline{a})}{C}{(S_2,\overline{a})}{\mathbf{eq}}
\). 
Similarly, we can define another group action: 
\[ \actiontemp{g}{g^2 = e}{\{x, y, z\}}{g \cdot (x \rightarrow \alpha_1, y \rightarrow \alpha_2, z \rightarrow \alpha_3) = (x \rightarrow -\alpha_1, y \rightarrow \alpha_2, z \rightarrow \pi - z)}\]
If we denote this as \((S_2,{b})\), then we would like to prove the triple \(
\triple{(S_2,\overline{b})}{C}{(S_2,\overline{b})}{\mathbf{eq}}
\).

The proofs for each of these triples follows the same pattern. Since each statement is injective, we can use the $\textsc{ASSGN}$ rule followed by sequencing and the $\textsc{FOR}$ to conclude the triple.
\section{Verifying Assignments - Encoding}
    \label{app:encoding}
    Theorem \ref{def: formula-fake} says there exists a  first-order formula $F \llparenthesis \triple{(G,a_{\mathit{Vars}})}{v \coloneq exp}{(H,b_{\mathit{Vars}})}{\;} \rrparenthesis$ which is satisfiable if, and only if the triple $\triple{(G,a_{\mathit{Vars}})}{v \coloneq exp}{(H,b_{\mathit{Vars}})}{\;}$ is valid. In the following section, we construct such a formula and show its soundness. 
The following lemma shows that it suffices to check a triple only for the generators in the pre-condition. This enables a finite formula even for infinite groups. 
\begingroup
\setlength{\intextsep}{0pt}
\setlength{\columnsep}{0pt}
\begin{wrapfigure}{r}{0.54\textwidth}
  \begin{center}
    \begin{tabular}{c}
    \begin{lstlisting}
1. $\while(t \coloneq 0; t < T; t \coloneq t + dt)${
2. $\quad \quad x \coloneq (A \cdot \sin{(z)} + C \cdot \cos{(y)})\cdot dt + x;$
3. $\quad \quad y \coloneq  (B \cdot \sin{(x)} + A \cdot \cos{(z)})\cdot dt + y;$
4. $\quad \quad z \coloneq  (C \cdot \sin{(y)} + B \cdot \cos{(x)})\cdot dt + z;$
5. }
     \end{lstlisting}
    \end{tabular}
  \end{center}
    \vspace{-2.5mm}
  \caption{A program simulating the ABC flow}
  \label{fig: ABC}
\end{wrapfigure}

\begin{lemma}
  \label{lemma: finite-pre-helper}
  Let the group $G$ be finitely presented as $G = \langle g_1, \dots, g_n \mid r_i, \dots, r_j \rangle$ with group action $(G,a_{\mathit{Vars}})$. Suppose $(H,b_{\mathit{Vars}})$ is a group action, and $C$ is a program statement.
  If for all $1 \leq i \leq n$, there exists $h \in H$ such that for all
  $\sigma \in \Sigma$ we have $b_h(\sem{C}_{\sigma}) = C(a_{g_i}(\sigma))$
  then for all $g \in G$, there exists $h \in H$ such that for all $\sigma \in
    \Sigma$, we have $b_h(\sem{C}_{\sigma}) = C(a_{g}(\sigma))$.
\end{lemma}
\iflong
  \begin{proof}
    We know that for all $1 \leq i \leq n$, $$\text{ there exists } h \in H, \forall \sigma \in \Sigma\; b_h(\sem{C}_{\sigma}) = C(a_{g_1}(\sigma))$$ We will show that
    for all  $g \in G$, $$\text{ there exists } h \in H  \;\forall \sigma \in \Sigma\; b_h(\sem{C}_{\sigma}) = C(a_{g}(\sigma))$$ We proceed by induction.  \\
    \textit{Base Case : n = 1} This is true by assumption.  \\
    \textit{Inductive Step : } Let $g' = g_1 \cdot g_2 \cdot \dots \cdot g_k$. By the inductive hypothesis, there exists $h'$ such that
    $b_{h'}(\sem{C}_{\sigma}) = \sem{C}_{(a_{g_2 \cdot \dots \cdot g_k}}(\sigma))$. Also, there exists $h_1$ such that
    $b_{h_1}(\sem{C}_{\sigma}) =\sem{C}_{(a_{g_1}(\sigma))}$. Then by Lemma \ref{lem: mult-closure}, $$\forall \sigma \in \Sigma, b_{ h'\cdot h_1}(\sem{C}_{\sigma}) = \sem{C}_{(a_{g'}(\sigma))}$$ Therefore, for all
    $g \in G$, $$\text{ there exists } h \in H, \forall \sigma \in \Sigma, b_h(\sem{C}_{\sigma}) = \sem{C}_{(a_{g_1}(\sigma))}$$
  \end{proof}
\fi
\Cref{lemma: finite-pre-helper} also enables checking implications for finitely presented groups in the pre-condition.
\begin{corollary}
\label{cor: entail}
 Let the group $G$ be finitely presented as $G = \langle g_1, \dots, g_n \mid r_i, \dots, r_j \rangle$ with faithful group action $(G,a_{\mathit{Vars}})$. Suppose $(H,b_{\mathit{Vars}})$ is a group action, and $C$ is a program statement.
  If for all $1 \leq i \leq n$, there exists $h \in H$ such that for all
  $\sigma \in \Sigma$ we have, $b_h({\sigma}) = a_{g_i}(\sigma)$
  then $(G,a_{\mathit{Vars}}) \overset{\varphi}{\rightarrow} (H,b_{\mathit{Vars}})$ where $\varphi$ is the unique homomorphism for which the entailment is valid. 
\end{corollary}
\iflong
\begin{proof}
    The statement follows by instantiating \Cref{lemma: finite-pre-helper} with the $\textsc{SKIP}$ statement and applying \Cref{lemma-pre-homo} with $C$ as $\textsc{SKIP}$ gives the corollary. 
\end{proof}
\fi
We exploit Lemma \ref{lemma: finite-pre-helper} to define a simple translation that allows checking triples with assignment statements. 
Each equation in the presentation of the group action is of the form:
\[
  g \cdot (v_1 \rightarrow \alpha_1, \dots, v_k \rightarrow \alpha_k) = (v_1 \rightarrow e_1, \dots, v_k \rightarrow e_k)
\]
where each $\alpha_i$ is a logical variable ranging over all integers, representing $\sigma(v_i)$, and $g$ maps $\sigma$ to a program state where
variable $v_i$ contains $e_i$, where the symbolic expression $e_i$ may depend on
the logical variables $\alpha$. We introduce notation to express this more compactly.
Let $\Sigma'$ be the set of all maps from $\mathit{Vars}$ to symbolic expressions.
Let $\sigma': \mathit{Vars} \rightarrow \mathcal{SEXP}$ be a map from variables to symbolic
expressions. Then, $\sigma'(v_i) = \alpha_i$, and $a_g(\sigma')(v_i) = e_i$.
Since integers are symbolic expressions, $\Sigma \subset \Sigma'$.
We define a function $\Gamma$ that takes a program expression, and a map from variables to symbolic expressions, and returns a symbolic expression.
We assume that program expressions and symbolic expressions have the same function symbols as determined by the set $Ops$. The function
$\Gamma: \mathcal{EXP} \times \Sigma' \rightarrow \mathcal{SEXP}$ is defined as:
\begin{align*}
  \Gamma(n)_{\sigma'}                  & = n                                    \\
  \Gamma(v)_{\sigma'}                  & = \sigma'(v)                           \\
  \Gamma(f(v_1, \dots, v_n))_{\sigma'} & =   f(\Gamma(e_1), \dots, \Gamma(e_n))
\end{align*}
We are now ready to define a formula that captures the semantic condition of the $\textsc{SEM-ASSGN}$ rule.
\begin{definition}
  \label{def: formula-f}
  Let $G = \langle g_1, \dots, g_n \mid r_1, \dots, r_j \rangle$ be a finitely
  generated group with $n$ generators, and action $(G,a_{\mathit{Vars}})$. Let $H$ be a
  finite group with $k$ elements and action $(H,b_{\mathit{Vars}})$, and $\sigma' \in
    \Sigma$ by a map such that $\sigma'(v_i) = \alpha_i$ for every variable. Let
  $v \coloneq exp$ be an assignment command.
  Then the formula $F \llparenthesis \triple{(G,a_{\mathit{Vars}})}{v \coloneq exp}{(H,b_{\mathit{Vars}}}{} \rrparenthesis$ is:
  \begin{multline*}
    \forall \alpha_1, \dots, \alpha_i, \bigwedge\limits_{i=1}^n \left(
    \bigvee\limits_{j}^k \left (\Gamma(exp)_{a_{g_i}(\sigma')} = b_{h_j} (
    \sigma'[v \mapsto \Gamma(exp)_{\sigma'}])(v) \right. \right. \\
    \left. \left.  \land \bigwedge_{v_i \in \mathit{Vars} \land v' \neq v} a_{g_i}(\sigma')(v') = b_{h_j}(\sigma'[v \mapsto \Gamma(exp)_{\sigma'}])(v') \right) \right)
  \end{multline*}
\end{definition}
The formula  checks that for every generator of the group $G$, there exists an element  $h
  \in H$, such that the action of $h$ on the symbolic state after an assignment is the same as the symbolic state after an assignment that has been acted on by $g \in G$. 
  
  We now work towards proving that if $F$ is satisfiable, then the semantic
  condition of $\textsc{SEM-ASSGN}$ is valid.
  First, the expression produced by $\Gamma$ is a logical representation of the
  program expression.
  \begin{lemma}
    \label{lem: conversion}
    For any $\sigma \in \Sigma$, we have $\Gamma(exp)_{\sigma} = \sem{e}_{\sigma}$.
  \end{lemma}
  \begin{proof}
    We proceed by induction on the syntax of $e$
    \begin{enumerate}
      \item $e = n$: In this case, we have $\Gamma(n)_{\sigma'} = n$ and $\sem{n}_{\sigma} = n$
      \item $e = v_i$: we have $\Gamma(v_i)_{\sigma} = \sigma(v_i) = z \in
              \mathbb{Z}$, and $\sem{v \coloneq v_i})_{\sigma}(v) =  \sigma(v_i) = x_i$.
    \end{enumerate}
    Inductive Case: 
    $e = f(e_1, \dots, e_2):$ $\sem{f(e_1, \dots, e_n)}_{\sigma} =
              f(\sem{e_1}_{\sigma}, \dots, \sem{e_n}_{\sigma})$. By the
            induction hypothesis, we get the theorem.
  \end{proof}
  The next few lemmas prove that each inner equality in $F$ ensures that every variable $v$ is related by the group action on the output states.
  \begin{lemma}
    \label{lemma: assgn-formula-1}
    Let $\sigma'$ be a map from variables to expressions that represents the state before an assignment statement, i.e., $\sigma'(v_i) = \alpha_i$. If  \[\forall \alpha_1, \alpha_2, \dots, \alpha_n \in \mathbb{Z}, \Gamma(exp)_{a_g(\sigma')} = b_h(\sigma'[v \mapsto \Gamma(exp)_{\sigma'}])(v)\] then
    $$\forall \sigma \in \Sigma, b_h(\sem{v \coloneq exp})_{\sigma}(v_i)= \sem{v\coloneq exp}_{a_g(\sigma)}(v)$$
  \end{lemma}
  \begin{proof}
    $\Gamma(v_i)_{a_g(\sigma')} = e_1$, and $b_h(\sigma'[v \mapsto \Gamma(exp)_{\sigma'}])(v) = b_h(\sigma'[v \mapsto \Gamma(exp)])(v) = e_2$, and by the premise, $e_1 = e_2$.
    At the same time, for any $\sigma$, $b_h(\sem{v \coloneq exp})_{\sigma}(v) =  b_h(v \rightarrow \sem{exp}_\sigma)(v)$. By Lemma \ref{lem: conversion}, $b_h(v \rightarrow \Gamma({exp})_\sigma)(v) = \sem{e_2}_\sigma$, and $\sem{v\coloneq exp}_{a_g(\sigma)}(v) = \sem{e_1}_{\sigma}$, which must be equal to $\sem{e_2}_{\sigma}$.
  \end{proof}
  \begin{lemma}
    \label{lemma: assgn-formula-2}
    Let $\sigma'$ be the map from variables to expressions that represents the state before an assignment statement. If  $$a_g(\sigma')(v_i) = b_h(\sigma'[v \mapsto \Gamma(exp)_{\sigma'}])(v)$$ then
    $$\forall \sigma \in \Sigma, b_h(\sem{v \coloneq exp})_{\sigma}(v_i)= \sem{v\coloneq exp}_{a_g(\sigma)}(v_i)$$ where $v \neq v_i$
  \end{lemma}
  \begin{proof}
    Consider any $\sigma \in \Sigma$. Suppose that
    ${a_g(\sigma)}(v_i) = \alpha'$. By assumption, $b_h(\sigma[v \mapsto \Gamma(exp)_{\sigma'}])(v_i) = \alpha'$.
    Now, $b_h(\sem{v \coloneq exp})_{\sigma}(v_i) = b_h(\sigma[v \mapsto \sem{exp}_{\sigma}])(v_i)$. By Lemma \ref{lem: conversion}, $b_h(\sigma[v \mapsto \Gamma(exp)_{\sigma}])(v_i) = \alpha'$.
    At the same time, $\sem{v\coloneq exp}_{a_g(\sigma)}(v_i) = \sem{exp}_{a_g(\sigma)}(v_i)$, and by Lemma \ref{lem: conversion}, this is equal to $\Gamma(exp)_{a_g(\sigma)}(v_i)$. Therefore, $$ b_h(\sem{v \coloneq exp})_{\sigma}(v_i)= \sem{v\coloneq exp}_{a_g(\sigma)}(v_i)$$
  \end{proof}
  We are now ready to prove our main theorem.

\begin{restatable}{theorem}{TripleFormSound}
  Let $G = \langle g_1, \dots, g_n \mid r_1, \dots, r_j \rangle$ be a finitely generated group with $n$ generators, and action $(G,a_{\mathit{Vars}})$. Let $H$ be a finite group with $k$ elements and action $(H,b_{\mathit{Vars}})$, and $\sigma' \in \Sigma$ by a map such that $\sigma'(v_i) = \alpha_i$ for every variable.
  \label{thm:triple-formula-sound}
  If the formula $F \llparenthesis \triple{(G,a_{\mathit{Vars}})}{v \coloneq exp}{(H,b_{\mathit{Vars}}}{}\rrparenthesis$ is satisfiable then,  \[\forall g\in G
    \;\exists h \in H,
    \forall \sigma \in \Sigma\; b_{h_j}(\sem{v \coloneq exp})_{\sigma}= \sem{v\coloneq exp}_{a_{g_i}(\sigma)}
  \]
\end{restatable}  
  \begin{proof}
    We first observe that the domain of $\Sigma$ is $\mathit{Vars}$ and since the $\alpha$'s range over all integers, the premise holds for all $\sigma \in \Sigma$. If the above expression is true, then for every $g_n$, the inner disjunct must be true for some $h_k$. Now assume that the inner disjunct is true. Then,
    $ \Gamma(exp)_{a_{g_i}(\sigma')} = b_{h_j}(\sigma'[v \mapsto \Gamma(exp)_{\sigma'}])(v)$ holds, by Lemma \ref{lemma: assgn-formula-1}, we can conclude that    $$\forall \sigma \in \Sigma\; b_{h_j}(\sem{v \coloneq exp})_{\sigma}(v)= \sem{v\coloneq exp}_{a_{g_i}(\sigma)}(v)$$ and we can apply  Lemma \ref{lemma: assgn-formula-2}, to conclude that for $v' \in \mathit{Vars} \setminus \{v\}$, $b_{h_j}(\sem{v \coloneq exp})_{\sigma}(v')= \sem{v\coloneq exp}_{a_{g_i}(\sigma)}(v')$. Therefore, for every $g$ in the set of generators of $G$,
    \[\exists h \in H, \text{such that } \forall \sigma \in \Sigma, \forall v' \in \mathit{Vars}, b_h(\sem{v \coloneq exp})_{\sigma}(v') = \sem{ v\coloneq exp}_{{a_g(\sigma)}}(v') \]
    By Lemma \ref{lemma: finite-pre-helper},
    for every element $g$ in $G$,
    \[\exists h \in H, \text{such that } \forall \sigma \in \Sigma, b_h(\sem{v \coloneq exp})_{\sigma} = \sem{ v\coloneq exp}_{a_g(\sigma)} \]
  \end{proof}
We demonstrate this encoding for a simple example:
\begin{example}
  Consider the triple $\triple{\tuple{S_2}{a}{\{x,y\}}}{max \coloneq x > y \;? \;x : y}{(S_2, \overline{a}_{\mathit{Vars}})}{\mathbf{eq}}$, which is saying that if $x$ and $y$ are permuted, after this assignment statement, $x$ and $y$ remain permuted, and $max$ is unchanged. We show the translation only for the variable $max$.
  The group  $S_2$ has one generator, call it $g$, and $E$ has only one element, $e$, so we need to check
  \begin{align*}
    \forall \alpha_x, \alpha_y, \Gamma(x > y ? x : y)_{a_{g_1}(\sigma')} = b_e(\sigma'[max \rightarrow \Gamma(x > y ? x : y)_{\sigma'}])(max)
  \end{align*}
  which is equivalent to $\forall \alpha_x, \alpha_y, (\alpha_y > \alpha_x
    ?\alpha_y : \alpha_x) = (\alpha_x > \alpha_y ? \alpha_x : \alpha_y)$
  Since this formula is satisfiable, we can establish $\forall g \in S_2, \exists e \in E, \forall \sigma \in \Sigma, e(\sem{C}_{\sigma})(max) = \sem{C}_{a_g(\sigma)}(max)$.
\end{example}
To finish showing that the encoding is satisfiable. we establish that when the post-condition is faithful, the triple $ \triple{(G,a_{\mathit{Vars}})}{v \coloneq exp}{(H,b_{\mathit{Vars}})}{\varphi}$ is valid. 
\begin{corollary} 
  Let $G = \langle g_1, \dots, g_n \mid r_1, \dots, r_j \rangle$ be a finitely generated group with $n$ generators, and action $(G,a_{\mathit{Vars}})$. Let $H$ be a finite group with $k$ elements and action $(H,b_{\mathit{Vars}})$.
  \label{thm:triple-formula-corollary}
  If $(H,b_{\mathit{Vars}})$ is a faithful group action on the reachable states from the assignment command, and $F \llparenthesis \triple{(G,a_{\mathit{Vars}})}{v \coloneq exp}{(H,b_{\mathit{Vars}})}{\;} \rrparenthesis$ is satisfiable,
  then $\vDash \triple{(G,a_{\mathit{Vars}})}{v \coloneq exp }{(H,b_{\mathit{Vars}})}{\varphi}$ where $\varphi$ is the unique homomorphism that makes the triple sound. 
\end{corollary}
  \begin{proof}
    By Theorem \ref{thm:triple-formula-sound} for every element $g$ in $G$,
    \[\exists h \in H, \text{such that } \forall \sigma \in \Sigma, b_h(\sem{v \coloneq exp})_{\sigma} = \sem{ v\coloneq exp}_{a_g(\sigma)} \]
    Define $\varphi: G \rightarrow H$ such that $\varphi(g) = h$ and $b_h(\sem{v \coloneq exp})_{\sigma} = \sem{ v\coloneq exp}_{a_g(\sigma)}$. \\
    We need to prove that $\varphi(g_1 \circ g_2) = \varphi(g_1) \cdot \varphi(g_2)$. Let $\varphi(g_1) \cdot \varphi(g_2) = h_1 \cdot h_2$.
    By the group action axioms, we know that \[b_{h_1 \cdot h_2}(\sem{C}_{\sigma}) = b_{h_1}(b_{h_2}(\sem{C}_\sigma))\] Since the group action is on $\mathit{Vars}$, by Lemma \ref{lem: compose-multiple}, $b_{h_1}(b_{h_2}(\sem{C}_{\sigma})) = \sem{C}_{a_{g_1}(a_{g_2}(\sigma))}$. \\
    Now, let $\varphi(g_1 \circ g_2) = h_3$. By the definition of $\varphi$,  $b_{h_3}(\sem{C}_{\sigma}) = \sem{C}_{a_{g_1}(a_{g_2}(\sigma))}$. Since $(H,b)$ is a faithful action, any two elements with the same group action must be equal.
    This implies that \[h_3 = h_1 \cdot h_2\] Therefore, \[\varphi(g_1 \circ g_2) = \varphi(g_1) \cdot \varphi(g_2)\]
    The identity case is trivial.
    Therefore, $\vDash \triple{(G,a_{\mathit{Vars}})}{v \coloneq exp}{(H,b_{\mathit{Vars}})}{\varphi}$.
  \end{proof}

This corollary proves \Cref{def: formula-fake}.

\paragraph*{Towards finitely presented post-conditions}
A drawback of our encoding is that the post-condition group must be finite. If
we have a homomorphism in mind---for example, when proving that a group action
is preserved, the homomorphism of interest is often the identity
homomorphism---we can use this homomorphism as part of our encoding.  Given a
homomorphism $\varphi: G \rightarrow H$, the encoding from Definition \ref{def:
  formula-f} can be simplified to:
\begin{multline*}
  \forall \alpha_1, \dots, \alpha_i, \bigwedge\limits_{i=1}^n \left( (\Gamma(exp)_{a_{g_i}(\sigma')} = b_{\varphi(g_i)} ( \sigma'[v \mapsto \Gamma(exp)_{\sigma'}])(v) \land  \right. \\ \left. \bigwedge_{v_i \in \mathit{Vars} \land v' \neq v} a_{g_i}(\sigma')(v') = b_{\varphi(g_i)}(\sigma'[v \mapsto \Gamma(exp)_{\sigma'}])(v') \right)
\end{multline*}
This encoding no longer requires enumerating over the elements of $H$, and therefore, can handle finitely generated groups in the post-condition.
\fi
\end{document}